\documentclass[12pt,letterpaper]{article}           % fleqn: align equations left

\usepackage{geometry}                                     % Custom margins for single page, etc.
\usepackage{fullpage}                                     % Use the full page
\usepackage{setspace}                                     % Enables custom margins, doublespacing, etc.
\usepackage{pdflscape}                                    % Use: \begin{landscape} ... \end{landscape}
\usepackage{tikz}

\usetikzlibrary{arrows.meta}

\definecolor{demandblue}{RGB}{31,119,180}
\definecolor{supplygreen}{RGB}{44,150,44}
\definecolor{shadowred}{RGB}{190,40,40}
\definecolor{dwtri}{RGB}{255,140,30}
\definecolor{ceilgray}{RGB}{80,80,80}

\tikzset{dropline/.style={dashed, gray!65, thin}}

\usepackage[utf8]{inputenc}                             % Font definition and input type
\usepackage[T1]{fontenc}                                  % Font output type 
\usepackage{lmodern}                                      % Latin Modern fonts
\usepackage{textcomp}                                     % Supports many additional symbols

\usepackage{mathtools}
\usepackage{amsmath}                                      % Math equations, etc.
\usepackage{amsthm}                                       % Math theorems, etc.
\usepackage{amsfonts}                                     % Math fonts (e.g. script fonts)
\usepackage{amssymb}                                      % Math symbols such as infinity

\usepackage{nicefrac,xfrac}
\usepackage{xcolor}                                          % Enables colored text
\definecolor{darkblue}{rgb}{0.0,0.0,0.66}                 % Custom color: dark blue
\definecolor{myblue}{RGB}{0,114,178}
\definecolor{myred}{RGB}{213,94,0}
\definecolor{myyellow}{RGB}{242,237,194}
\definecolor{mygreen}{RGB}{0,158,115}           % Custom color: dark blue
\usepackage[hyperfootnotes=false,bookmarksopen]{hyperref} % Enable hyperlinks, expand menu subtree
\usepackage{doi}
\hypersetup{                                              % Custom hyperlink settings
    pdffitwindow=false,                                   % true: window fit to page when opened
    pdfstartview={XYZ null null 1.00},                    % Fits the default zoom of the page to 100%
    pdfnewwindow=true,                                    % Links in new window
    colorlinks=true,                                      % false: boxed links; true: colored links
    linkcolor=darkblue,% Color of internal links
    citecolor=darkblue,                                   % Color of links to bibliography
    urlcolor=darkblue  }                                  % Color of external links

\usepackage{graphicx}                                     % Allows .jpg images to be included
\usepackage[section]{placeins}                            % Forces floats to stay in section
\usepackage{float}                                        % Used with restylefloat
\restylefloat{figure}                                     % "H" forces a figure to be "exactly here"
\usepackage[]{caption}   
\usepackage[disable]{todonotes} %[disable]to shut off notes
\usepackage[authordate,bibencoding=auto,strict,backend=biber,natbib]{biblatex-chicago}
\usepackage{csquotes}

\usepackage{tikz}                                        % Timelines and other drawings
\usetikzlibrary{decorations}                             % Formating for Tikz
\usetikzlibrary{calc}
\usetikzlibrary{matrix}
\usetikzlibrary{positioning}

\usepackage{authblk}

\usepackage{graphicx}
\graphicspath{{../output/}{./}}                           % Look for figures in output/ and current dir
\usepackage{caption}
\usepackage{subcaption}

\usepackage{amsthm}
\newtheorem{theorem}{Theorem}
\newtheorem{proposition}{Proposition}
\newtheorem{corollary}{Corollary}
\newtheorem{claim}[theorem]{Claim}

\newtheorem{lemma}{Lemma}

\theoremstyle{definition}

\newtheorem{definition}{Definition}
\newtheorem{assumption}{Assumption}

\newtheorem{remark}[theorem]{Remark}

\usepackage{pgfplots}
\pgfplotsset{compat=1.11}
\usepgfplotslibrary{fillbetween}
\usetikzlibrary{intersections}
\usetikzlibrary{patterns}
\newcommand{\R}{\mathbb{R}}

\usepackage{outlines}
\usepackage{enumitem}
\setenumerate[1]{label=\Roman*.}
\setenumerate[2]{label=\Alph*.}
\setenumerate[3]{label=\roman*.}
\setenumerate[4]{label=\alph*.}
\setlist[itemize]{leftmargin=4em}
\setlist[enumerate]{leftmargin=4em}

\usepackage[]{mathpazo}

\usepackage{abstract}
\renewenvironment{abstract}
{{%
		\setlength{\leftmargin}{0mm}
		\setlength{\rightmargin}{\leftmargin}%
	}%
	\relax}
{\endlist}

\makeatletter
\def\@maketitle{%
	\newpage
	{\fontsize{24}{26}\selectfont\raggedright  \setlength{\parindent}{0pt} \@title \par
	\vskip 20pt\relax \normalsize\fontsize{18}{20}
	\@author \par
	\vskip 20pt
	\normalsize\fontsize{18}{20}
	\@date \par}%
}
\makeatother

\usepackage{longtable,booktabs}
\usepackage[utf8]{inputenc}

\usepackage{graphicx,grffile}
\makeatletter
\def\maxwidth{\ifdim\Gin@nat@width>\linewidth\linewidth\else\Gin@nat@width\fi}
\def\maxheight{\ifdim\Gin@nat@height>\textheight\textheight\else\Gin@nat@height\fi}
\makeatother
\setkeys{Gin}{width=\maxwidth,height=\maxheight,keepaspectratio}

\title{\centering Chaos and Misallocation \\ \vspace{.25em} under Price Controls
\thanks{We thank Robert Bradley, Jeffrey Clemens, Frank Easterbrook, Jonathan Meer and seminar participants at the University of Chicago Law School for helpful comments and suggestions. All errors are our own.}}%

\author{%
Brian C.~Albrecht\thanks{International Center for Law \& Economics. Email: \href{mailto:mail@briancalbrecht.com}{mail@briancalbrecht.com}.},~ %
Alex Tabarrok\thanks{George Mason University. Email: \href{mailto:Tabarrok@gmu.edu}{Tabarrok@gmu.edu}.},~ %
and Mark Whitmeyer\thanks{Arizona State University. Email: \href{mailto:mark.whitmeyer@gmail.com}{mark.whitmeyer@gmail.com}.}%
}

\date{\today}
\begin{document}

% \listoftodos
% \newpage

% \tableofcontents
% \newpage

	% \pagenumbering{arabic}% resets `page` counter to 1 
	%	
		% \maketitle
	
	{% \usefont{T1}{pnc}{m}{n}
		\setlength{\parindent}{0pt}
		\thispagestyle{plain}
		{\fontsize{18}{20}\selectfont\raggedright 
			\maketitle  % title \par  
			
		}
		
	}

	\begin{abstract}
		
		\hbox{\vrule height .2pt width 39.14pc}
		
		\vskip 8.5pt % \small 

		\noindent Price controls kill the incentive for arbitrage. We prove a Chaos Theorem: under a binding price ceiling, suppliers are indifferent across destinations, so arbitrarily small cost differences can determine the entire allocation. The economy tips to corner outcomes in which some markets are fully served while others are starved; small parameter changes flip the identity of the corners, generating discontinuous welfare jumps. These corner allocations create a distinct source of cross-market misallocation, separate from the aggregate quantity loss (the Harberger triangle) and from within-market misallocation emphasized in prior work. They also create an identification problem: welfare depends on demand far from the observed equilibrium. We derive sharp bounds on misallocation that require no parametric assumptions. In an efficient allocation, shadow prices are equalized across markets; combined with the adding-up constraint, this collapses the infinite-dimensional welfare problem to a one-dimensional search over a common shadow price, with extremal losses achieved by piecewise-linear demand schedules. Calibrating the bounds to station-level AAA survey data from the 1973–74 U.S. gasoline crisis, misallocation losses range from roughly 1 to 9 times the Harberger triangle.

\vspace{0.5cm}


\vspace{0.2cm}


		\hbox{\vrule height .2pt width 39.14pc}

	\end{abstract}
	\vfill
    \doublespacing
	\section{Introduction}

In February 1974, the aggregate gasoline shortfall in the United States was roughly 9 percent nationwide.\footnote{A linear trend fit over 1970--1973 implies 1974 shortfalls of 8--9\% in both EIA petroleum products supplied \citep{eia_petroleum_products_supplied_annual} and FHWA net gasoline volume taxed \citep{fhwa_net_volume_gallons_taxed}.} Price controls on oil kept the price at the pump below the market price. If the standard model of price ceilings held and markets are similar, each local market would have experienced a roughly proportional 9 percent reduction, shadow prices would equalize, and the familiar Harberger triangle would measure the welfare cost. Instead, pain fell in a patchwork. As \citet{yergin_prize_1991} describes, ''gasoline was in short supply in major urban areas, but there were more than abundant supplies in rural and vacation areas, where the only shortage was of tourists.'' Why should a 9 percent national shortfall leave some markets with nothing while others have more than abundant supplies? This paper shows that feast or famine, not modest belt-tightening, is the generic outcome of price controls.

The argument rests on a simple observation: price controls kill the incentive for arbitrage. Under market prices, if gasoline is more valuable in Connecticut than in Idaho, traders profit by redirecting supply until shadow prices equalize. A slight cost advantage for one destination captures only the marginal tanker loads, because the price premium in the underserved market rises just enough to offset the cost difference. But when prices are frozen at $\bar{p}$ everywhere, there is no price gap to exploit. Every unit earns identical revenue regardless of destination, so suppliers are indifferent about where to ship. A one-cent cost advantage is as good as a one-dollar advantage: both redirect the entire flow, because there is no price in the destination market pushing back. What would have been a marginal reallocation becomes a categorical one, and the economy loses its capacity for incremental adjustment, instead lurching between all-or-nothing extremes \citep[p.\ 137]{sowell_knowledge_1980}. Allocation is no longer determined by consumer valuations but by whatever breaks the tie: small differences in transportation costs, regulatory discretion, historical consumption patterns.

The formal consequence is what we call the "Chaos" Theorem. Since price is identical across markets, any cost advantage, however small, can dramatically swing allocations. The supplier's problem reduces to cost minimization over the feasible set. Total costs are linear in quantity, the feasible set (bounded by supply and market demands) is a convex polytope. The problem, therefore, admits a vertex solution: markets are served in order of increasing delivery cost, lower-cost markets are filled to their demand caps, at most one cutoff market is partially served, and higher-cost markets receive nothing. The allocation is not guided by willingness to pay or welfare. The ''chaos'' is thus not the mathematical notion of chaos theory but the economic one of disordered, uncoordinated and divorced from welfare. 

This corner-solution structure differs from benchmarks in prior work.
The standard Harberger triangle is the \textit{minimum} welfare loss compatible with a binding ceiling; and it requires the very price variation a ceiling forbids. Corners are worse because they maximize the shadow price gap: starved markets have desperate consumers willing to pay far above the ceiling, while oversupplied markets clear at the controlled price. Each unit sent to the wrong market destroys value equal to this gap. \citet{glaeser_misallocation_2003} quantify misallocation under rent control assuming random allocation, where every apartment has equal probability of reaching any tenant. Random allocation is an interior outcome with its own smoothness: scarcity is spread widely and small parameter changes produce small welfare changes. We show that equilibrium generically occurs at neither the Harberger nor Glaeser-Luttmer benchmark. Cost-minimizing suppliers drive allocations to vertices, not interiors. Corners are not an assumption but an \emph{outcome} about what cost-minimizing suppliers choose. The correct benchmark is corners, not random, and corners generate qualitatively different welfare properties: losses far larger than either efficient or random distributions, and discontinuous jumps from small parameter perturbations.

Figure~\ref{fig:rationing_map} documents this dispersion for the peak of the crisis at the state level, showing the percentage of stations rationing fuel (either completely out of fuel or limiting purchases). Of sampled stations, Connecticut and Massachusetts exceeded 90 percent, while Idaho, Montana, Utah, Hawaii and Wyoming reported no problems. A natural interpretation is that Idaho got lucky and Connecticut got unlucky. But think about what efficient allocation would look like: with a 9 percent national shortfall, every market should experience roughly a 9 percent reduction, or at least \emph{some} reduction. Zero rationing means those states received more than their efficient share. The states with abundant fuel are a sign of the \textit{problem}. The states with zero rationing were, in a sense, overserved, since they had enough to satisfy demand even at the controlled price, not just at the market price. At market prices, vacationers in rural areas would never have outbid commuters in cities for scarce fuel. Under a binding ceiling, both pay the same price, so there is no mechanism to redirect supply. The rationing in Connecticut and the abundance in Idaho are two sides of the same misallocation.

\begin{figure}[!ht]
  \centering
  \includegraphics[width=\textwidth]{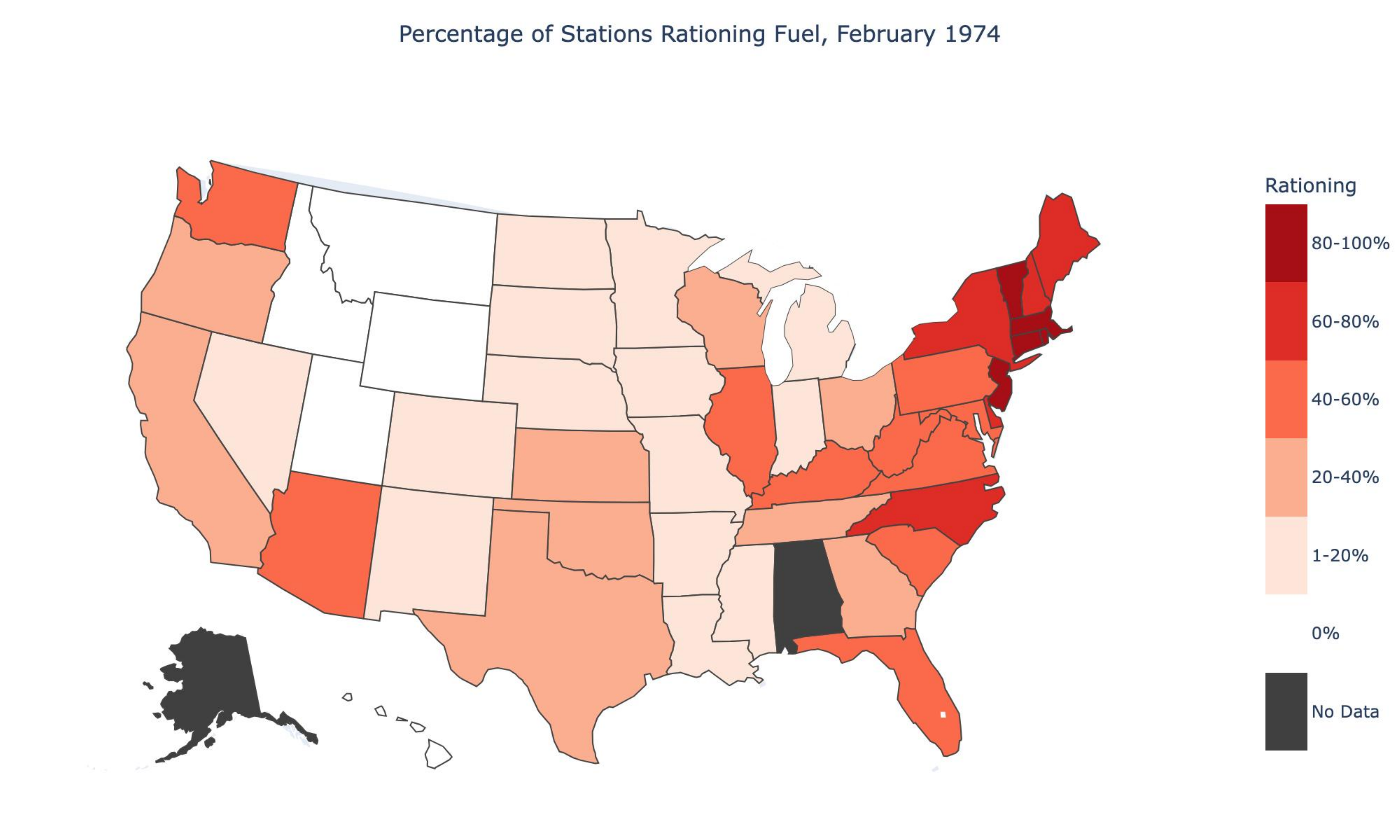}
  \caption{Percentage of gasoline stations rationing fuel by state, February 1974. Rationing includes stations completely out of fuel and those limiting purchases (e.g., maximum gallons per customer). Data are from AAA surveys of sampled stations \citep{aaa_aaa_1974}.}
  \label{fig:rationing_map}
\end{figure}

This reframes the puzzle. The question is not why some states rationed. With a 9 percent national shortfall, rationing was inevitable. The question is why other states exhibited no rationing. Why should some markets clear while others run dry? The answer is that efficient allocation requires price variation, which a binding ceiling explicitly forbids. Once prices cannot adjust, there is no mechanism to direct scarce supplies toward their highest-value uses. The patchwork pattern of Figure~\ref{fig:rationing_map} is not an accident. It is the generic outcome when prices cannot signal scarcity.\footnote{Pure corner solutions are themselves an extreme case. They assume away contractual rigidities, inventory buffers, and other frictions that slow adjustment. But as a limiting benchmark, corners have empirical bite.}

Under a ceiling, equilibrium allocations are generically corners: some markets clear at the controlled price while others are starved. Welfare depends critically on valuing units at these extreme allocations. But demand is identified only locally, near observed price-quantity pairs. Computing welfare at a corner requires integrating inverse demand deep into a region no estimation has reached. Parametric approaches settle that extrapolation by assuming functional form, perhaps bolstered with elasticity sensitivities. We instead ask what welfare statements are possible without committing to a specific extrapolation.

We prove sharp bounds on misallocation losses requiring no parametric demand assumptions. The key structural insight is that efficient allocation equalizes shadow prices across markets, so all markets share a common shadow price $p^*$. The adding-up constraint $\sum_i q_i(p^*) = \bar{Q}$ converts the infinite-dimensional problem over demand functions into a one-dimensional search over the scalar $p^*$. Within each market, the extremal demand curves that attain the bounds are piecewise linear, with slopes at the elasticity-implied endpoints. Instead of being a robustness check, robust bounds are the methodologically appropriate response to a challenge intrinsic to the theoretical problem.

How large are misallocation losses? Our robust interval implies a misallocation-to-Harberger ratio of roughly 1 to 9. Even at the lower bound, misallocation roughly equals the quantity-reduction loss; at the upper bound, the Harberger triangle accounts for barely one-tenth of total welfare cost. The quantity reduction is not the main cost of price controls. The misallocation is.

The remainder of the paper proceeds as follows. Section~\ref{sec:literature} reviews related literature. Section~\ref{sec:two-market} introduces a two-market example. Section~\ref{sec:general-model} presents the general model. Section~\ref{sec:chaos} develops the Chaos Theorem, and Section~\ref{sec:simulation} illustrates the chaos mechanism with simulations. Section~\ref{sec:empirical} applies our framework to the 1973--74 gasoline crisis, developing robust bounds on misallocation losses. Section~\ref{sec:pricetoquantity} documents how the allocation chaos generated demand for direct quantity controls. We defer proofs to Appendices~\ref{omittedproofs} and \ref{app:extremals-proofs}.

\section{Literature Review}
\label{sec:literature}

Our paper connects three distinct strands of research: the classic welfare-economics treatment of price controls, the modern literature on misallocation, and recent work that uses robust or worst-case
methods to bound welfare when key allocation details are unknown.

The canonical approach traces back to \citet{harberger_monopoly_1954}, who measured the efficiency cost of quantity reductions with the familiar ''triangle.'' Early analyses of price ceilings took the triangle as their primary yardstick (e.g. \citet{oi_measuring_1976}, \citet{gordon_response_1973}, and \citet{rockoff_drastic_1984} for U.S. wartime experience). 

A subsequent literature pointed out that the triangle misses the welfare costs of non-price rationing. \citet{barzel_theory_1974} showed that when money prices are suppressed, rationing by waiting emerges: consumers queue until the shadow ''time price'' clears the market. \citet{frech_welfare_1987} and \citet{deacon_rationing_1985,deacon_welfare_1989} applied this framework to the 1973--74 gasoline crisis, finding that queuing costs alone could exceed the entire rent transfer from producers to consumers. \citet{frech_welfare_1987} explicitly consider cross-market allocation, analyzing inefficiencies between urban and rural gasoline markets in California. In their framework, optimal rationing by queuing would equalize \emph{time prices} (money price plus waiting cost) across markets. Misallocation occurs when time prices differ: urban consumers faced longer queues than rural consumers, generating welfare losses from the failure to equalize shadow prices. But crucially, their analysis assumes an interior solution: every market receives positive supply, and inefficiency arises from unequal shadow prices across locations. Our paper identifies a distinct source of welfare loss that this queuing literature cannot capture. The Chaos Theorem predicts corner solutions where some markets receive literally nothing, an allocative loss that exists even in a hypothetical world with zero queuing costs.

Our paper is most closely related to \citet{davis_allocative_2011}, who study misallocation in the market for natural gas. \citet{brunt_moving_2025} makes a similar pedagogical point, arguing that misallocation losses under price controls are often larger than the quantity-reduction triangle. \citet{Murphy1992} analyze misallocation under partial price liberalization, showing that when some markets are freed while others remain controlled, supply flows to uncontrolled markets until controlled markets are starved. We extend their framework by characterizing the generic structure of equilibrium allocations (corners) and the discontinuous dependence on nuisance parameters.

A foundational question is whether prices or rationing better allocate scarce goods. \citet{weitzman_is_1977} shows that the answer depends on the joint distribution of needs and income: the price system dominates when wants are dispersed or incomes are equal, while rationing can dominate when needs are uniform but incomes are unequal. The mechanism-design literature formalizes this tradeoff: \citet{condorelli_market_2013} characterizes when market mechanisms versus non-market mechanisms are optimal, while \citet{akbarpour_redistributive_2024} show that non-market allocation dominates when willingness to pay is uninformative about the designer's welfare weights. \citet{kominers_price_2025} define price gouging as occurring when lowering the price from the market-clearing level would raise utilitarian welfare. Our paper abstracts from redistributive motives, focusing instead on the allocative inefficiency that arises when a binding ceiling segments markets and blocks arbitrage.

Our paper builds on the literature studying random allocation under binding price ceilings. \citet{glaeser_misallocation_2003} show that apartments do not necessarily flow to their highest-value users under rent control, and they quantify the resulting misallocation under random allocation. \citet{Bulow2012} characterize when consumer surplus rises or falls under different demand-curve shapes. 

We focus instead on market segmentation: markets are naturally segmented by geography, by transaction and transport costs, by final use, and by position in the structure of production. We show that equilibrium generically occurs at neither the Harberger benchmark (efficient allocation) nor the Glaeser-Luttmer benchmark (random allocation). Cost-minimizing suppliers drive allocations to vertices, not interiors.

Recent empirical work on rent control, such as \citet{diamond_effects_2019} on San Francisco or the comprehensive survey by \citet{kholodilin_rent_2024}, documents substantial misallocation in housing, though identification of the allocative component remains challenging. \citet{buurma-olsen_quantifying_2025} make progress on this front, finding that Dutch public-housing tenants consume units differing 7.5 percent from their efficient allocation. We provide a general theoretical framework that nests these settings as special cases.

Our robust bounds draw on the literature measuring welfare without fully specifying demand. \citet{hausman_nonparametric_1995} estimate consumer surplus using nonparametric methods; \citet{hausman_individual_2016} show that with unobserved heterogeneity, welfare is only partially identified. \citet{kremer_worst-case_2018} characterize worst-case deadweight loss, showing it can approach the entire surplus under extreme demand shapes. \citet{kang_robustness_2025} show that familiar functional forms are extremal within canonical demand families; \citet{kang_robust_regulation_2026} extends the robustness approach to the prices-versus-quantities choice. We contribute by applying these methods to misallocation losses across segmented markets, and by showing that the bound-finding problem reduces to a one-dimensional optimization with piecewise-linear extremals.

\section{Two Market Example}
\label{sec:two-market}

Consider a national market divided into two submarkets with demand functions $D_1(p)$ and $D_2(p)$. A binding price ceiling $\bar{p}$ creates a total shortage of $D(\bar{p}) - S(\bar{p})$, where $S(p)$ is aggregate supply. We denote total supply under the ceiling as $\bar Q \equiv S(\bar{p})$. The quantities allocated must sum to this total: $q_1 + q_2 = \bar Q$. No market can receive more than it demands at the ceiling price, otherwise the market price would fall below the ceiling. This gives upper bounds $q_1 \leq D_1(\bar{p}) \equiv \bar q_1$ and $q_2 \leq D_2(\bar{p}) \equiv \bar q_2$. Together, these constraints define the feasible set.

Figure~\ref{fig:misallocation_two_market} illustrates the welfare consequences of different allocations. In Panel~A, the shortage is allocated efficiently, and shadow prices equalize across markets at $p^*$, yielding deadweight losses $a$ and $b$ that sum to the national Harberger triangle $c$. The Harberger triangle is the \textit{smallest} welfare cost compatible with a binding ceiling.

Panel~B shows an equally-feasible corner allocation where Market~1 receives its full demand at the controlled price and Market~2 gets whatever is left. The welfare loss is nearly an order of magnitude larger than measured by the Harberger triangle. 

\begin{figure}[htbp]
  \centering
  % ================================================================
%  UNIFORM SCALE: xscale = 2.6 cm/unit, yscale = 2.8 cm/unit
%
%  This makes q1 + q2 = Q̄ hold visually: the physical distance
%  of q1 from origin plus q2 from origin equals Q̄ from origin
%  in the aggregate panel (all measured in the same cm/unit).
%
%  Demand functions:
%    Individual:  P_i(q) = 1.2 − q,   max q = 1.2
%    Aggregate:   P(Q)   = 1.2 − Q/2, max Q = 2.4
%
%  Key values:
%    p̄  = 0.4  →  y = 1.120
%    p* = 0.7  →  y = 1.960
%    p2 = 1.0  →  y = 2.800   [Panel B, market 2 shadow price]
%    q* = 0.5  →  x = 1.300   [each submarket, Panel A]
%    D_i(p̄)=0.8 → x = 2.080  [each submarket]
%    Q̄  = 1.0  →  x = 2.600  [aggregate]   1.300+1.300 = 2.600 ✓
%    D(p̄)=1.6  → x = 4.160  [aggregate]   2.080+2.080 = 4.160 ✓
%
%  Panel B:
%    q1 = 0.8  →  x = 2.080  q2 = 0.2 → x = 0.520  (sum=1.0 ✓)
%    p_2(0.2) = 1.0 → y = 2.800
%
%  Layout: col1=0  col2=5.0  col3=10.0   rowA=5.6  rowB=0
% ================================================================

\begin{tikzpicture}[
  >=Stealth,
  font=\small,
  line width=0.65pt
]

% ──────────────────────────────────────────────────────────────────
%  Column headings
% ──────────────────────────────────────────────────────────────────
\node[font=\normalsize\bfseries] at (1.80, 9.85) {Market 1};
\node[font=\normalsize\bfseries] at (6.80, 9.85) {Market 2};
\node[font=\normalsize\bfseries] at (12.80,9.85) {Aggregate Market};

% ──────────────────────────────────────────────────────────────────
%  PANEL A  —  efficient allocation, shadow prices equalised
% ──────────────────────────────────────────────────────────────────
\node[font=\normalsize\bfseries, anchor=east] at (-0.3,7.05) {Panel~A};

%-- A1: Market 1 -------------------------------------------------
\begin{scope}[shift={(0.0,5.6)}]
  \draw[->] (0,0) -- (3.85,0) node[right] {$q_1$};
  \draw[->] (0,0) -- (0,3.80) node[above] {$P$};
  % Demand: (0,3.360)→(3.120,0)
  \draw[demandblue, thick] (0,3.360) -- (3.120,0);
  \node[demandblue, anchor=south west] at (1.10,2.10) {$D_1(p)$};
  % Ceiling
  \draw[ceilgray, dashed, thin] (0,1.120) -- (3.65,1.120)
    node[right, ceilgray, font=\footnotesize] {$\bar{p}$};
  % Shadow price
  \draw[shadowred, dashed, thin] (0,1.960) -- (3.65,1.960)
    node[right, shadowred, font=\footnotesize] {$p^*$};
  % DWL triangle a: (q*,p*)=(1.300,1.960), (D(p̄),p̄)=(2.080,1.120)
  \fill[dwtri, opacity=0.80]
    (1.300,1.960) -- (2.080,1.120) -- (1.300,1.120) -- cycle;
  \node[font=\bfseries] at (1.590,1.460) {$a$};
  % Drop lines
  \draw[dropline] (1.300,1.960) -- (1.300,0);
  \draw[dropline] (2.080,1.120) -- (2.080,0);
  \draw[dropline] (0,1.960) -- (1.300,1.960);
  % Dot
  \fill[black] (1.300,1.960) circle (1.8pt);
  % Ticks
  \draw (1.300,0.06) -- (1.300,-0.06) node[below, font=\footnotesize] {$q_1^*$};
  \draw (2.080,0.06) -- (2.080,-0.06) node[below, font=\footnotesize] {$\bar{q}_1$};
  \draw (0.06,1.960) -- (-0.06,1.960)
    node[left, font=\footnotesize, shadowred] {$p^*$};
\end{scope}

%-- A2: Market 2 -------------------------------------------------
\begin{scope}[shift={(5.0,5.6)}]
  \draw[->] (0,0) -- (3.85,0) node[right] {$q_2$};
  \draw[->] (0,0) -- (0,3.80) node[above] {$P$};
  \draw[demandblue, thick] (0,3.360) -- (3.120,0);
  \node[demandblue, anchor=south west] at (1.10,2.10) {$D_2(p)$};
  \draw[ceilgray, dashed, thin] (0,1.120) -- (3.65,1.120)
    node[right, ceilgray, font=\footnotesize] {$\bar{p}$};
  \draw[shadowred, dashed, thin] (0,1.960) -- (3.65,1.960)
    node[right, shadowred, font=\footnotesize] {$p^*$};
  \fill[dwtri, opacity=0.80]
    (1.300,1.960) -- (2.080,1.120) -- (1.300,1.120) -- cycle;
  \node[font=\bfseries] at (1.590,1.460) {$b$};
  \draw[dropline] (1.300,1.960) -- (1.300,0);
  \draw[dropline] (2.080,1.120) -- (2.080,0);
  \draw[dropline] (0,1.960) -- (1.300,1.960);
  \fill[black] (1.300,1.960) circle (1.8pt);
  \draw (1.300,0.06) -- (1.300,-0.06) node[below, font=\footnotesize] {$q_2^*$};
  \draw (2.080,0.06) -- (2.080,-0.06) node[below, font=\footnotesize] {$\bar{q}_2$};
  \draw (0.06,1.960) -- (-0.06,1.960)
    node[left, font=\footnotesize, shadowred] {$p^*$};
\end{scope}

%-- A3: Aggregate ------------------------------------------------
% Demand: (0,3.360)→(6.240,0) but draw to x≈5.5 for visual clarity
\begin{scope}[shift={(10.0,5.6)}]
  \draw[->] (0,0) -- (5.60,0) node[right] {$Q$};
  \draw[->] (0,0) -- (0,3.80) node[above] {$P$};
  \draw[demandblue, thick] (0,3.360) -- (5.200,0.560);
  \node[demandblue, anchor=south west] at (.9,2.7) {$D(p)$};
  % Supply: Q̄=1.0 → x=2.600
  \draw[supplygreen, line width=1.2pt] (2.600,0) -- (2.600,3.65);
  % Ceiling
  \draw[ceilgray, dashed, thin] (0,1.120) -- (5.35,1.120)
    node[right, ceilgray, font=\footnotesize] {$\bar{p}$};
  % Shadow price
  \draw[shadowred, dashed, thin] (0,1.960) -- (5.35,1.960)
    node[right, shadowred, font=\footnotesize] {$p^*$};
  % Harberger triangle c: bounded by demand curve above, p̄ below, Q̄ left
  \fill[dwtri, opacity=0.80]
    (2.600,1.960) -- (4.160,1.120) -- (2.600,1.120) -- cycle;
  \node[font=\bfseries] at (3.115,1.415) {$c$};
  % Drop lines
  \draw[dropline] (2.600,1.960) -- (2.600,0);
  \draw[dropline] (4.160,1.120) -- (4.160,0);
  \draw[dropline] (0,1.960) -- (2.600,1.960);
  \fill[black] (2.600,1.960) circle (1.8pt);
  % Ticks
  \draw (2.600,0.06) -- (2.600,-0.06)
    node[below, font=\footnotesize] {$\bar{Q}{=}S(\bar{p})$};
  \draw (4.160,0.06) -- (4.160,-0.06)
    node[below, font=\footnotesize] {$D(\bar{p})$};
  \draw (0.06,1.960) -- (-0.06,1.960)
    node[left, font=\footnotesize, shadowred] {$p^*$};
  % Annotation
  \node[align=center, font=\footnotesize\itshape] at (4.50,3.30)
    {$a+b=c$\\[1pt]{\normalfont(efficient)}};
\end{scope}

% ──────────────────────────────────────────────────────────────────
%  PANEL B  —  corner allocation
% ──────────────────────────────────────────────────────────────────
\node[font=\normalsize\bfseries, anchor=east] at (-0.3,1.75) {Panel~B};

%-- B1: Market 1 (fully served: q1=D1(p̄)=0.8, a=0) --------------
\begin{scope}[shift={(0.0,0.0)}]
  \draw[->] (0,0) -- (3.85,0) node[right] {$q_1$};
  \draw[->] (0,0) -- (0,3.80) node[above] {$P$};
  \draw[demandblue, thick] (0,3.360) -- (3.120,0);
  \node[demandblue, anchor=south west] at (1.10,2.10) {$D_1(p)$};
  % p̄=p1 in red (it is the shadow price for this fully-served market)
  \draw[shadowred, dashed, thin] (0,1.120) -- (3.65,1.120)
    node[right, shadowred, font=\footnotesize] {$\bar{p}{=}p_1$};
  % Drop lines
  \draw[dropline] (2.080,1.120) -- (2.080,0);
  \draw[dropline] (0,1.120) -- (2.080,1.120);
  % Dot
  \fill[black] (2.080,1.120) circle (1.8pt);
  % a=0 label upper-right of dot, outside demand curve
  \node[font=\bfseries, anchor=west] at (2.22,1.52) {$a=0$};
  % Ticks
  \draw (2.080,0.06) -- (2.080,-0.06)
    node[below, font=\footnotesize] {$q_1{=}\bar{q}_1$};
  \draw (0.06,1.120) -- (-0.06,1.120)
    node[left, font=\footnotesize, shadowred] {$\bar{p}{=}p_1$};
\end{scope}

%-- B2: Market 2 (residual: q2=0.2, large b) ---------------------
\begin{scope}[shift={(5.0,0.0)}]
  \draw[->] (0,0) -- (3.85,0) node[right] {$q_2$};
  \draw[->] (0,0) -- (0,3.80) node[above] {$P$};
  \draw[demandblue, thick] (0,3.360) -- (3.120,0);
  \node[demandblue, anchor=south west] at (1.10,2.10) {$D_2(p)$};
  \draw[ceilgray, dashed, thin] (0,1.120) -- (3.65,1.120)
    node[right, ceilgray, font=\footnotesize] {$\bar{p}$};
  % Shadow price p2=1.0
  \draw[shadowred, dashed, thin] (0,2.800) -- (3.65,2.800)
    node[right, shadowred, font=\footnotesize] {$p_2$};
  % Large DWL triangle b: (0.520,2.800)–(2.080,1.120)–(0.520,1.120)
  \fill[dwtri, opacity=0.80]
    (0.520,2.800) -- (2.080,1.120) -- (0.520,1.120) -- cycle;
  \node[font=\bfseries] at (1.020,1.870) {$b$};
  % Drop lines
  \draw[dropline] (0.520,2.800) -- (0.520,0);
  \draw[dropline] (2.080,1.120) -- (2.080,0);
  \draw[dropline] (0,2.800) -- (0.520,2.800);
  % Dot
  \fill[black] (0.520,2.800) circle (1.8pt);
  % Ticks
  \draw (0.520,0.06) -- (0.520,-0.06) node[below, font=\footnotesize] {$q_2$};
  \draw (2.080,0.06) -- (2.080,-0.06) node[below, font=\footnotesize] {$\bar{q}_2$};
  \draw (0.06,2.800) -- (-0.06,2.800)
    node[left, font=\footnotesize, shadowred] {$p_2$};
\end{scope}

%-- B3: Aggregate (Harberger c unchanged; a+b > c) ---------------
\begin{scope}[shift={(10.0,0.0)}]
  \draw[->] (0,0) -- (5.60,0) node[right] {$Q$};
  \draw[->] (0,0) -- (0,3.80) node[above] {$P$};
  \draw[demandblue, thick] (0,3.360) -- (5.200,0.560);
  \node[demandblue, anchor=south west] at (.9,2.7) {$D(p)$};
  \draw[supplygreen, line width=1.2pt] (2.600,0) -- (2.600,3.65);
  \draw[ceilgray, dashed, thin] (0,1.120) -- (5.35,1.120)
    node[right, ceilgray, font=\footnotesize] {$\bar{p}$};
  % Shadow price (faint — p* exists but shadow prices diverge in submarkets)
  \draw[shadowred, dashed, thin] (0,1.960) -- (5.35,1.960)
    node[right, shadowred, font=\footnotesize] {$p^*$};
  % Harberger triangle c: bounded by demand curve above, p̄ below, Q̄ left
  \fill[dwtri, opacity=0.80]
    (2.600,1.960) -- (4.160,1.120) -- (2.600,1.120) -- cycle;
  \node[font=\bfseries] at (3.115,1.415) {$c$};
  % Drop lines
  \draw[dropline] (2.600,1.960) -- (2.600,0);
  \draw[dropline] (4.160,1.120) -- (4.160,0);
  \draw[dropline] (0,1.960) -- (2.600,1.960);
  \fill[black] (2.600,1.960) circle (1.8pt);
  \draw (2.600,0.06) -- (2.600,-0.06)
    node[below, font=\footnotesize] {$\bar{Q}{=}S(\bar{p})$};
  \draw (4.160,0.06) -- (4.160,-0.06)
    node[below, font=\footnotesize] {$D(\bar{p})$};
  \draw (0.06,1.960) -- (-0.06,1.960)
    node[left, font=\footnotesize, shadowred] {$p^*$};
  % Annotation
  \node[align=left, font=\footnotesize] at (4.40,3.30)
    {$a+b>c$\\[1pt]{\itshape(misallocation}\\{\itshape adds DWL)}};
\end{scope}

\end{tikzpicture}
  \caption{\textbf{Two Feasible Allocations Under Price Controls.}  A price ceiling is at $\bar{p}$, the quantity supplied is \(\bar Q = S(\bar{p})\). Panel~A:  Supply is allocated efficiently and shadow prices equalize at $p^*$ across submarkets and losses sum to the aggregate Harberger triangle; $a + b = c$. Panel~B: An equally feasible allocation. Market~1 is fully served at $\bar{p}$. Market~2 receives only the residual. Shadow prices diverge and $a + b > c$.}
  \label{fig:misallocation_two_market}
\end{figure}

Corner allocations are feasible but why should we expect corner allocations? Figure~\ref{fig:corners} shows the geometry. The constraint $q_1 + q_2 = \bar Q$ defines a downward-sloping line; the upper bounds $q_i \leq \bar q_i$ truncate it to a line segment. The endpoints $E_1$ and $E_2$ are corner solutions.

\begin{figure}[h!]
    \centering
    \includegraphics[width=0.6\linewidth]{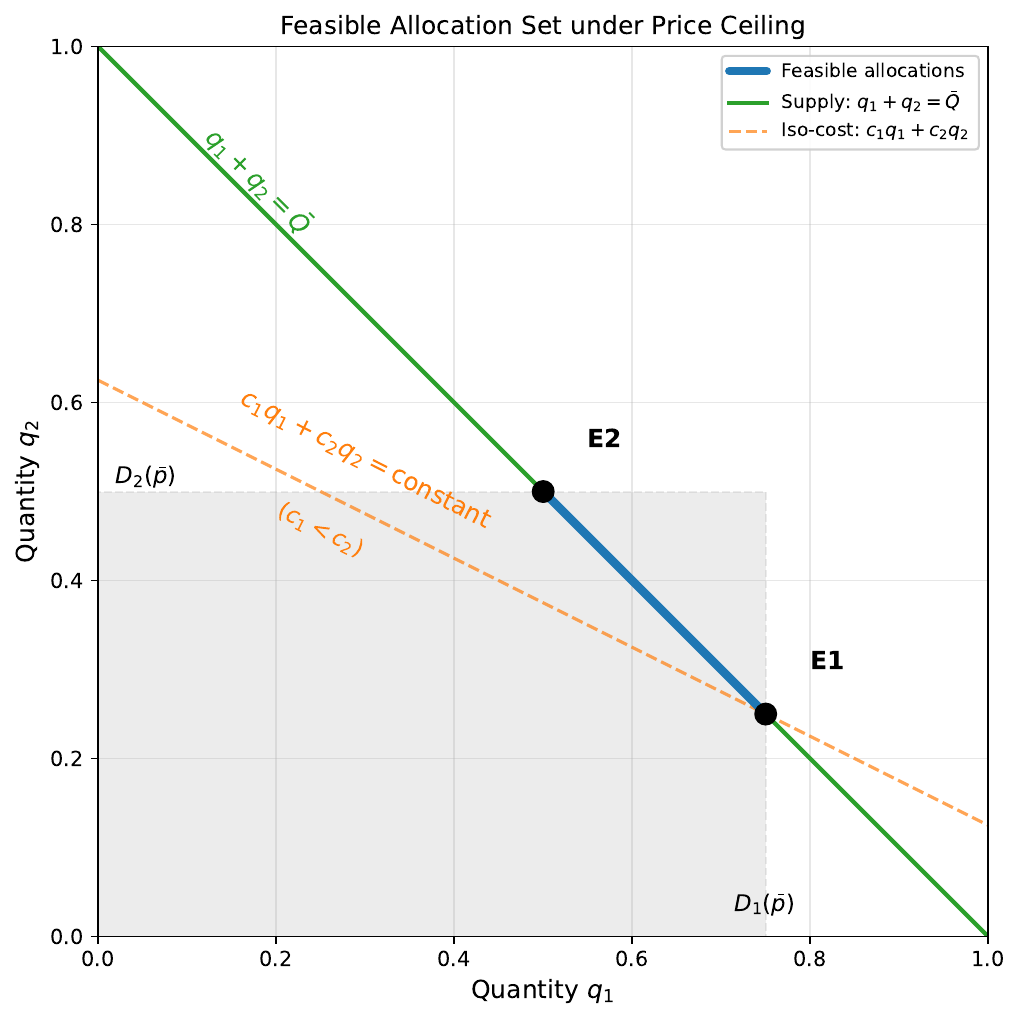}
    \caption{The feasible set is a line segment. Corners $E_1$ and $E_2$ are the only generic equilibria under price controls.}
    \label{fig:corners}
\end{figure}

Between $E_1$ and $E_2$ lies a continuum of feasible allocations. Without price controls, arbitrage would push the economy toward the efficient interior. But a binding ceiling eliminates this incentive: with money prices frozen at $\bar{p}$ everywhere, sellers are indifferent across all feasible allocations. As argued in the introduction, this indifference is resolved by cost minimization. Suppose transporting gasoline to Market~$i$ costs $c_i$ per unit. If $c_1 < c_2$, Market~1 is fully served; if $c_2 < c_1$, Market~2 is fully served. The continuum of feasible allocations collapses to a corner.

When $c_1$ and $c_2$ are nearly equal, the economy sits near a knife-edge. Any small change in relative costs can flip the entire allocation from $E_1$ to $E_2$. Because welfare varies along the feasible set, these discontinuous allocation jumps produce discontinuous welfare jumps. We formalize this in Section~\ref{sec:chaos}.

The term 'submarket' is deliberately flexible. Submarkets may be geographic--states, cities, countries--but can also span temporal segments, end uses, stages of production, and quality tiers. Likewise, 'transport cost' means any cost of moving product between segments. Set a ceiling at the "summer price" and you eliminate the incentive to 'transport' oil to winter via storage. Cap refined-product prices and minute differences in refining costs can redirect nearly all crude toward gasoline, starving diesel and jet fuel. Segments are manifold but the mechanism is the same: price controls induce seller indifference across allocations so any friction, however small, selects a corner; and small shifts in friction flip the system between corners. 

\section{General Model}
\label{sec:general-model}

We now generalize the two-market example to an arbitrary number of submarkets with general demand functions. The goal is to characterize, for any market structure and any binding price ceiling, the full range of welfare outcomes, from the efficient allocation that minimizes deadweight loss to the worst-case allocation that maximizes it.

Formally, a binding ceiling $\bar p$ fixes the aggregate quantity at $\bar Q=S\left(\bar{p}\right)$, where $S$ is the aggregate supply curve. Let $D$ denote aggregate demand. We assume that $\bar{p}$ is binding in the sense that $D\left(\bar{p}\right) > S\left(\bar{p}\right)$. There are $n\geq 1$ segmented submarkets. In market $i$, inverse demand $P_i(\cdot) \geq 0,$ is continuous and nonincreasing on $\left[0, D_i\left(\bar{p}\right)\right]$.

We formalize market segmentation via a decomposition that partitions aggregate demand into submarket demands. A decomposition specifies how many consumers in each submarket are willing to buy at any given price. The inverse demand $P_i(q_i)$ in submarket $i$, the price at which exactly $q_i$ units would be demanded, represents the marginal willingness to pay, or \emph{shadow price}, when quantity $q_i$ is allocated to that market.

\begin{definition}\label{def:decomp}
A collection of direct demand functions $\{D_i\}_{i=1}^n$ with $D_i$ continuous and nonincreasing in price is a \emph{decomposition} if $D(p)=\sum_i D_i(p)$ on $[0, \infty)$.
\end{definition}

The price ceiling $\bar{p}$ constrains feasible allocations in two ways. First, total quantity is constrained to equal supply: $\sum_i q_i = \bar Q \equiv S(\bar{p})$. Second, no submarket can receive more than its demand at the ceiling price, or the local price would fall below the ceiling. Define $\bar q_i \equiv D_i(\bar p)$ as the maximum quantity submarket $i$ can absorb at price $\bar{p}$. The feasible set is then:
\[
\mathcal{F} \coloneqq \left\{q \in\mathbb R_+^n\colon \ 0 \leq q_i \leq \bar{q}_i\ \forall \ i, \ \sum_{i=1}^n q_i=\bar Q\right\}\text{.}
\]

This is a convex polytope \citep{grunbaum_convex_2003}, meaning the intersection of a hyperplane (the adding-up constraint) with a box (the upper and lower bounds). Its extreme points, or vertices, correspond to allocations where all but at most one submarket is either completely served ($q_i = \bar q_i$) or completely starved ($q_i = 0$).

To measure welfare, we focus on consumer surplus. In submarket $i$, consumer surplus at allocation $q_i$, measured relative to the ceiling price $\bar p$, is the integral of willingness to pay above the price paid:
\[
W_i(q_i)\coloneqq\int_{0}^{q_i}\left[P_i(x)-\bar p\right]dx\text{.}
\]
Aggregating across submarkets gives total consumer surplus: $W(q) \coloneqq \sum_{i=1}^n W_i(q_i)$. For simplicity, we focus on welfare as consumer surplus from the demand side: how much surplus is lost when the fixed quantity $\bar{Q}$ is not allocated to its highest-value uses. If delivery costs $c_i$ differed enough across submarkets we would want to keep track of that as well. For example, the efficient benchmark in Proposition~\ref{prop:equalize} would equalize $P_i(q_i) - c_i$ rather than $P_i(q_i)$. Since the Chaos Theorem holds for arbitrarily small but distinct costs, the allocation result and the welfare measure are consistent in this limit.

It will sometimes be convenient to work with \emph{gross} surplus, or the area under the inverse demand curve without subtracting the price paid:
\[
\widetilde W_i(q_i)\coloneqq\int_{0}^{q_i}P_i(x) dx\text{,}
\qquad
\widetilde W(q) \coloneqq \sum_{i=1}^n \widetilde W_i(q_i)\text{.}
\]
Since $W(q) = \widetilde W(q) - \bar p \bar Q$ and total quantity $\bar Q$ is fixed on $\mathcal F$, maximizing consumer surplus $W$ is equivalent to maximizing gross surplus $\widetilde W$. This equivalence simplifies the analysis, so we can ignore the constant $\bar p \bar Q$ term and focus on gross surplus alone.

What allocation maximizes consumer surplus? Under standard conditions, surplus is maximized when the marginal value of the good is equalized across all submarkets that receive positive supply. If one market has a higher shadow price than another, we could increase total surplus by reallocating a unit from the low-value market to the high-value one. This intuition guides the following definition.
\begin{definition}\label{def:eq-shadow}
    An allocation $q \in \mathcal{F}$ \emph{equalizes shadow prices} if there exists $p^* \geq \bar{p}$ such that for every $i$, 
    \begin{enumerate}
        \item \label{def:eq-shadow1} $0 < q_i < \bar{q}_i$ $\Longrightarrow$ $P_i(q_i) = p^*$,
        \item \label{def:eq-shadow3} $q_i = \bar{q}_i$ $\Longrightarrow$ $P_i(q_i) \geq p^*$, and 
        \item \label{def:eq-shadow4} if $q_i = 0$, then $P_i(0) \leq p^*$.
    \end{enumerate}
\end{definition}
\begin{proposition}\label{prop:equalize}
$q^*$ maximizes $W$ over $\mathcal{F}$ if and only if $q^*$ equalizes shadow prices.
\end{proposition}

\noindent The proof is standard: if shadow prices differ, shifting a unit from the low-shadow-price market to the high-shadow-price market increases total surplus by the price gap, contradicting optimality. Conversely, since $W$ is concave on the convex set $\mathcal{F}$, any allocation satisfying these first-order conditions is a global maximizer.% and appears in Appendix~\ref{prop:equalize:proof}.

 Letting $ q^*\in\mathcal F$ denote an equal-shadow-price allocation from Proposition~\ref{prop:equalize}, we now note the following decomposition.

\begin{definition}\label{def:misalloc}
Let $q^* \in \mathcal{F}$ denote an equal-shadow-price allocation from Proposition~\ref{prop:equalize}. The \emph{misallocation deadweight loss} of a feasible allocation $q \in \mathcal{F}$ is the welfare shortfall relative to $q^*$:
\[
\mathcal{L}_{Mis}(q)
\coloneqq W(q^*)-W(q)
= \sum_{i=1}^n \left[W_i(q_i^*)-W_i(q_i)\right].
\]
\end{definition}

This definition isolates the welfare loss from \emph{how} the fixed total $\bar Q$ is distributed across submarkets. It holds constant the aggregate quantity reduction (the source of the Harberger triangle) and asks how much additional surplus is destroyed by allocating them inefficiently?

To build intuition, consider the expression in integral form:
\[
\mathcal{L}_{Mis}(q) = \sum_{i=1}^n \int_{q_i}^{ q_i^*} \left[P_i(x)-\bar p\right] dx
= \sum_{i=1}^n \int_{q_i}^{ q_i^*} P_i(x) dx,
\]
where the second equality uses $\sum_i(q_i^*-q_i)=0$ to cancel the $\bar p$ terms. Each integral measures the value of the units that submarket $i$ receives under the efficient allocation but not under $q$ (if $q_i < q_i^*$), or the value of excess units that $i$ receives under $q$ but not under the efficient allocation (if $q_i > q_i^*$). Markets that are under-served lose high-value units; markets that are over-served gain low-value units. The difference is pure waste.

The misallocation deadweight loss has an important economic interpretation: it represents the value that would be created if arbitrageurs could freely reallocate the fixed supply $\bar Q$ across markets. Under normal market conditions, price differences incentivize such reallocation. But when prices are frozen at $\bar p$, these incentives vanish. The welfare loss is precisely the forgone gains from trade. They are the dollars left on the table when goods do not flow to their highest-value uses.

%%We now show that $\mathcal{L}_{Mis}(q)$ is exactly the \emph{gain} from reassigning units until shadow prices equalize; i.e., the value created by arbitrage when prices cannot do the job. 

\subsection{Worst-Case Allocations}

We have established that equalizing shadow prices across submarkets minimizes deadweight loss for any fixed aggregate quantity $\bar Q$. We now characterize the opposite extreme. Which allocation \emph{maximizes} deadweight loss?

Recall that
\[V(q) \equiv \mathcal{L}_{Mis}( q)=\underbrace{\sum_i\int_0^{q_i^*}P_i(x) dx}_{\text{constant in }q}-\sum_i\int_0^{q_i}P_i(x) dx=W(q^*)-W(q)\text{.}
\]
As we noted earlier, maximizing misallocation $V(q)$ is equivalent to minimizing the gross consumer surplus $\widetilde W(q) \coloneqq \sum_i\int_0^{q_i}P_i(x) dx$ over the feasible set $\mathcal{F}$.

Intuitively, the worst case concentrates supply in markets where it generates the least value. Markets with high choke prices receive nothing; markets with low shadow prices are flooded.

\begin{theorem}%[Worst-case allocation structure]
\label{thm:worst-vertex}
Assume each $P_i$ is continuous and weakly decreasing on $[0,\bar q_i]$ and $\mathcal{F} \neq \emptyset$. Then,
\begin{enumerate}[label=(\roman*)]
\item\label{worstitem1} $V$ is maximized at an extreme point of $\mathcal{F}$, where $q_i\in\{0,\bar q_i\}$ for all but at most one \(i\).

\item\label{worstitem2} For any worst-case allocation \(q^w\in \mathcal{F}\), there exists a cutoff
\(\lambda_{\mathrm{worst}}\in[0,\infty)\) such that for every market \(i\),
\[
\begin{aligned}
&\text{if } \quad 0<q_i^w<\bar q_i, \quad \quad &&\text{then } P_i(q_i^w) = \lambda_{\mathrm{worst}},\\
&\text{if } \quad q_i^w = 0,  \quad \quad     &&\text{then } P_i(0) \ge \lambda_{\mathrm{worst}},\\
&\text{if } \quad q_i^w = \bar q_i, \quad \quad &&\text{then } P_i(\bar q_i) \leq \lambda_{\mathrm{worst}}.
\end{aligned}
\]
\item\label{worstitem3} If $\bar Q\le\bar q_j$ for some $j$, then among single-market vertices $(\bar Q,0,\dots,0)$ the worst case is any
\[
j\in\operatorname*{arg\,min}_i\int_0^{\bar{Q}} P_i(x) dx ,
\]
that is, the market with the smallest average value up to $\bar Q$ (not necessarily the smallest marginal price $P_i(\bar Q)$).
\end{enumerate}
\end{theorem}
We defer the formal proof of this theorem to Appendix \ref{app:vertex-proof}. Part~\ref{worstitem2} provides necessary KKT-type conditions for any worst-case allocation; identifying the global worst among candidate vertices requires comparing average values as in part~\ref{worstitem3}. Overall, the worst outcome is a form of all-or-nothing. Markets with choke prices above the threshold receive nothing; those below it receive everything. This corner structure maximizes the shadow price gap, hence the welfare loss.

\section{The Chaos of Price Controls}
\label{sec:chaos}

The previous section characterized the range of possible welfare outcomes: from the best case (equalizing shadow prices) to the worst case (concentrating supply in low-value markets). But this leaves open a crucial question: \emph{which} allocation should we expect to observe in practice?

Formally, we show small perturbations in underlying parameters can cause discontinuous jumps in allocations and welfare under price controls. Once price variation is removed, equilibrium allocations are determined by arbitrarily small differences in costs or capacities.

\subsection{Setup}

Let $\theta\in \Theta\subset\mathbb{R}^m$ be a compact parameter space governing submarket inverse demands $P_i(x;\theta)$, unit costs $c_i(\theta)$, and total quantity $\bar{Q}(\theta)>0$. We assume (i) \textit{Demand regularity}: for each $i$ and every $\theta$, the map $x\mapsto P_i(x;\theta)$ is continuous and strictly decreasing; (ii) \textit{Parameter continuity}: for each $x$, the map $\theta\mapsto P_i(x;\theta)$ is continuous. We also assume (iii) \textit{Cost continuity}: submarket unit costs $c_i(\theta)$ and the total quantity $\bar{Q}(\theta)>0$ are continuous in $\theta$.

Without price controls, the unique welfare-maximizing allocation $q^*(\theta)$ varies continuously with $\theta$ by Berge's Maximum Theorem. Small changes in parameters produce small changes in allocations and welfare.

\subsection{Discontinuity Under Price Controls}

Now impose a binding ceiling $\bar p$. Define $\bar q_i(\theta) \coloneqq D_i(\bar p;\theta)$ and $\bar{Q}\left(\theta\right) \coloneqq S(\bar p;\theta)$. Under price controls, profit-maximizing suppliers minimize delivery costs since all units sell at the same price $\bar{p}$. Thus, the allocation is determined by cost minimization:
\[
\min_{q\in\mathcal{F}(\theta)}\ c(\theta)\cdot q,
\]
where $\mathcal{F}(\theta)=\{q\in\mathbb{R}^n_+ \colon \ \sum_{i=1}^n q_i=\bar{Q}\left(\theta\right),\ 0 \leq q_i \leq \bar q_i(\theta)\}$.

%Under generic conditions (distinct costs across markets, and total quantity not exactly equal to any subset sum of capacities), the optimizer has a stark structure: markets are filled in order of increasing cost, with at most one market partially filled. Crucially, when two markets have nearly equal costs, an infinitesimal parameter change can flip the entire allocation between them.

Our next proposition records the generic structure of the controlled allocation away from cost ties and boundaries. That is, for our next result, we make the following genericity and nondegeneracy assumptions:
\begin{assumption}\label{ass:52}
    \begin{enumerate}
\item \label{ass1a} $0<\bar{Q}\left(\theta\right)<\sum_{i=1}^n \bar q_i(\theta)$;
\item \label{ass1b} $\bar{Q}\left(\theta\right)\neq \sum_{i\in S}\bar q_i(\theta)$ for every $S\subset\{1,\dots,n\}$;
\item $c_i(\theta)\neq c_j(\theta)$ for $i\neq j$.
\end{enumerate}
\end{assumption}
We order costs as $c_{(1)}(\theta)<\cdots<c_{(n)}(\theta)$ and let $k$ denote the unique index with $\sum_{s=1}^{k-1} \bar q_{(s)}(\theta) <\bar{Q}\left(\theta\right)< \sum_{s=1}^k \bar q_{(s)}(\theta)$. We say that an allocation, \(q\), \textit{fills markets in increasing order of cost, with one partially served market} if 
\[
q_{\left(r\right)}=
\begin{cases}
\bar q_{\left(r\right)}(\theta),&r<k,\\
\bar Q(\theta)-\sum_{s=1}^{k-1}\bar q_{\left(s\right)}(\theta),&r=k,\\
0,&r>k.
\end{cases}
\]
\begin{proposition}\label{prop:greedy-allocation}
Let \(\theta\in\Theta\) be such that Assumption \ref{ass:52} holds. Then, the price-control optimizing allocation \(q\) is unique and fills markets in increasing order of cost, with one partially served market.
\end{proposition}

\subsection{Chaos Result}
\label{sec:Chaos Result}

We continue to maintain the first two parts (\ref{ass1a} and \ref{ass1b}) of Assumption \ref{ass:52}, but jettison the third in order to cleanly state our main result on the jumps in allocations that can be produced by price controls. Let two feasible allocations $v$ and $w$ differ only in which markets receive supply--$v$ fills certain low-cost markets while $w$ fills others. Specifically, suppose allocation $w$ gives market $s$ more units than $v$ does ($w_s > v_s$), while market $r$ receives fewer units ($w_r < v_r$), with the same total quantity. Define the welfare jump--which, here, equals the difference in gross surplus--between them as
\[
\Delta W = \int_{v_s}^{w_s} P_s(x) dx - \int_{w_r}^{v_r} P_r(x) dx,
\]
where $r$ and $s$ are the markets whose allocations differ.

The theorem formalizes the knife-edge intuition. When costs are nearly equal, infinitesimal parameter changes flip which market is served and which is starved, producing discrete welfare jumps from continuous parameter paths. 

Suppose that there is some parameter value \(\theta^* \in \Theta\) such that two adjacent vertex allocations (vertices of $\mathcal{F}$ differing in exactly two coordinates) $v$ and $w$ are both optimal. Let \(r\) and \(s\) be the two markets on which they differ so that \(w_i = v_i\) for all \(i \notin \left\{r,s\right\}\).% and \(\Delta \coloneqq w_s-v_s = v_r-w_r > 0\).
\begin{definition}\label{ass:local-cost-only}
    We say that there is a local cost-only perturbation at \(\theta^* \in \Theta\) if there exists an open neighborhood \(O\subset\mathbb{R}^n\) of \(c\left(\theta^*\right)\) and a smooth function \(\phi \colon O\to\Theta\) such that \(\phi\left(c\left(\theta^*\right)\right)=\theta^*\) and for every \(c' \in O\), the parameter value \(\theta \coloneqq \phi(c')\) satisfies \(c\left(\theta\right)=c'\), \(\mathcal{F}\left(\theta\right)=\mathcal{F}\left(\theta^*\right)\), \(P_r\left(\cdot;\theta\right)=P_r\left(\cdot;\theta^*\right)\) and \(P_s\left(\cdot;\theta\right)=P_s\left(\cdot;\theta^*\right)\).
\end{definition}
That is, near \(\theta^*\), nearby cost vectors can be generated by nearby parameter changes without changing the feasible set, and without changing the two demand curves that enter \(\Delta W\).
\begin{theorem}[Chaos]\label{thm:chaos-main}
Suppose costs are such that two adjacent vertex allocations \(v\) and \(w\) are both optimal at some parameter \(\theta^*\), and suppose there is a local cost-only perturbation at \(\theta^*\) for \(v\) and \(w\).\footnote{Our cost-only perturbation assumption is much stronger than necessary. For this theorem, it would suffice that, in every neighborhood of \(\theta^\ast\), there exists a smooth local parameter change that leaves the feasible set and the two inverse-demand curves entering \(\Delta W\) unchanged, and along which \(v\) is the unique optimizer on one side while \(w\) is the unique optimizer on the other. We impose the stronger assumption only for expositional convenience.} If \(\Delta W\neq 0\) (a generic condition),\footnote{Here, ``generic'' is conditional on both vertex allocations \(w\) and \(v\) being optimal. Given this, the case we exclude, \(\Delta W = 0\), defines a closed and nowhere-dense subset of the space of continuous and strictly decreasing pairs \((P_r, P_s)\), in the sup-norm topology.} then for every neighborhood \(U\) of \(\theta^*\) there exist smooth parameter paths through \(U\) along which
\begin{enumerate}
\item The optimal allocation jumps discontinuously between \(v\) and \(w\);
\item Welfare jumps by \(+\left|\Delta W\right|\) along one path and \(-\left|\Delta W\right|\) along another.
\end{enumerate}
In particular, welfare is not locally monotone in parameters; arbitrarily small perturbations can move welfare up or down by discrete amounts.
\end{theorem}

\begin{corollary}
\label{cor:chaos-allocation}
For every neighborhood \(U\) of \(\theta^\ast\), there exist parameters \(\theta^-,\theta^+\in U\) such that \(v\) is an optimal allocation at \(\theta^-\) and \(w\) is an optimal allocation at \(\theta^+\).
\end{corollary}

Equivalently, the allocation changes by reallocating \(w_s-v_s=v_r-w_r\) units from market \(r\) to market \(s\), with all other market quantities unchanged. Appendix~\ref{app:chaos-proofs} provides the proof, showing how arbitrarily small perturbations can switch the optimum from one vertex allocation to the other.

The Chaos Theorem goes beyond standard LP sensitivity. Yes, linear programs with non-unique optima exhibit discontinuous selections, but the economic content here is \emph{why} the supplier's problem becomes linear. Under market clearing, welfare maximization has a strictly concave objective (diminishing marginal valuations), guaranteeing continuous allocations via Berge's Maximum Theorem. Price controls eliminate the price variation that creates this curvature: with revenue equalized, profit becomes linear in quantity. The degeneracy is the generic outcome when the arbitrage mechanism is shut down.

This result explains the patchwork pattern of Figure~\ref{fig:rationing_map}. Under price controls, which markets receive supply depends on cost parameters that may shift unpredictably with weather, refinery maintenance, or regulatory discretion. A small change--say, a pipeline repair that marginally reduces delivery costs to New Jersey--can discontinuously redirect supply away from New York, even though both states' underlying demands are nearly identical.\footnote{The Chaos Theorem assumes agents are rational cost-minimizers, which pushes allocations to corner solutions. In practice, long-term contracts, sticky relationships, or regulatory inertia may prevent immediate jumps to vertices. However, over time, competitive pressures tend to push allocations toward these extreme points, making the chaos result increasingly relevant as the control regime persists. In discussions of the crisis in the Congressional Record it is clear that allocations often change dramatically as contract terms end.}

The theorem also explains why which markets are starved is often volatile over time. Week-to-week fluctuations in transportation costs, inventory positions, or bureaucratic priorities could flip allocations between equilibria. Contemporary reports document exactly this volatility.

More broadly, the result shows that price controls introduce a fundamental unpredictability into markets. When prices cannot signal value, allocations become hypersensitive to ''nuisance parameters'' that would be irrelevant under market clearing. This hypersensitivity, chaos in the economic sense, is not a market failure but a direct consequence of suppressing the price mechanism. Ironically, the chaos of price controls often fuels demands for even tighter controls and increasingly politicized and bureaucratized allocations.

How much smoothing would be required to bring misallocation toward Harberger magnitudes? If half of supply were locked into historical patterns via long-term contracts or regulatory inertia, the remaining half would still face corner incentives, roughly halving misallocation losses but still leaving them several times the triangle. Only near-complete insulation from cost-based reallocation would eliminate the mechanism entirely. The Chaos Theorem is a limiting case, but the qualitative prediction that misallocation losses dominate quantity-reduction losses is robust to substantial smoothing.

To see why, suppose delivering $q$ costs $c(\theta)\cdot q + \frac{\kappa}{2}\|q\|^2$ for some $\kappa > 0$. The quadratic term makes the supplier's objective strictly convex, so the optimum moves continuously in parameters. But continuously does not mean gently. A cost-only perturbation leaves the feasible set unchanged, and $\|q^\kappa(\theta) - q^\kappa(\theta')\| \leq \frac{1}{\kappa}\|c(\theta) - c(\theta')\|$. When $\kappa$ is small, a small cost wedge still generates a large reallocation. Theorem~\ref{thm:chaos-main} is the sharp discontinuous limit of a broader sensitivity that persists whenever cost curvature is modest.

How likely are near-ties in practice? In large markets, they are typical. Suppose submarket costs are drawn i.i.d. from a continuous distribution. Under price controls, the greedy cost-minimization rule serves markets in order of increasing cost and stops when supply runs out. The relevant margin is the cutoff between the last served market and the first excluded one. As the number of submarkets grows, order statistics compress and the cost gap at this cutoff shrinks. The economy sits close to a tie, even though exact ties occur with probability zero. Thus, even modest shocks are likely to move the economy across the cutoff and produce large reallocations.

\subsection{Simulation}
\label{sec:simulation}

We illustrate the Chaos Theorem with simulations of 100 gasoline markets arranged on a $10\times 10$ grid. Each city has a per-unit delivery cost drawn uniformly from $[0, 0.1]$. We use a demand specification with a finite choke price, so welfare is well-defined even when markets receive zero.\footnote{We use a Hill (sigmoidal) demand function.} Aggregate supply is fixed at $Q = 150$ units, and we impose a price control at 80\% of the market-clearing price.

Under price controls, the allocation follows a greedy rule: fill the lowest-cost city to capacity, then the next-lowest, until supply is exhausted. This produces a vertex allocation where served cities receive their full demand while roughly 30 cities receive nothing.

Figure~\ref{fig:HeatMap1} shows allocations under three random cost draws. The top row shows free-market allocations; the bottom row shows price-control allocations for the same costs. Three patterns emerge. First, moving across the two rows we see that free-market allocations are nearly identical across scenarios--small cost differences produce small allocation differences. Second, moving down a column we see that the shift to price controls creates a very large, discontinuous shift in the allocation with many markets receiving zero (yellow)--note this discontinuous shift is a direct consequence of vertex allocation under price controls. Third, reading across the bottom row we see that different cost draws produce radically different allocations. This is the Chaos Theorem: small changes in nuisance parameters dramatically change the final allocation.
\begin{figure}[ht!]
    \centering
    \includegraphics[width=1\linewidth]{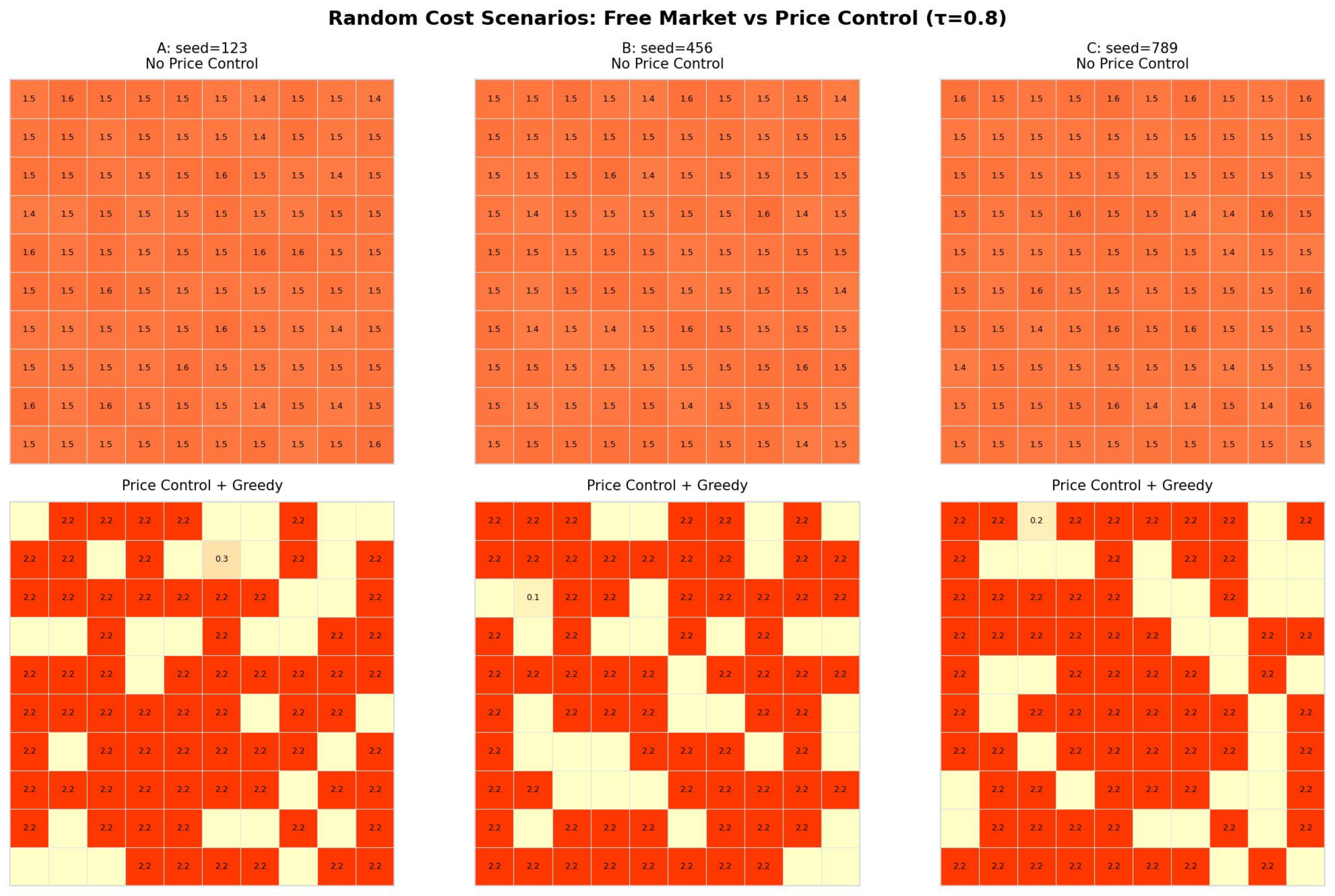}
    \caption{Allocation of a fixed supply ($\bar{Q} = 150$) across 100 cities under free markets (top) versus price controls (bottom). Each column is a different random draw of delivery costs. Under free markets, quantities vary smoothly. Under price controls, low-cost cities fill to capacity while approximately 30 cities go unserved.}
    \label{fig:HeatMap1}
\end{figure}

Cost structures may also have systematic components; locations close to refineries may have lower costs. These produce different spatial patterns but the same corner-solution logic.

\section{Robust Bounds on Misallocation}
\label{sec:robust-bounds}

The Chaos Theorem tells us that under price controls we should  expect corner allocations, far from market equilibria. Elasticity estimates identify demand locally near equilibrium, not at these depressed allocations. Rather than extrapolating via functional form assumption, we develop a partial identification approach. Given only the observed allocation, the ceiling price, bounds on local demand elasticities, and an anchored point on each demand curve (or a bounded interval for that anchor), we derive sharp bounds on misallocation. This is not sensitivity analysis within a parametric family. We do not commit to any functional form. Instead, we ask: across all inverse demand functions consistent with slope bounds, what are the largest and smallest possible welfare losses? The output is an identified interval, not a point estimate, and the interval covers welfare under all robustly possible allocations.

In market $i$, the unknown inverse demand $P_i\colon \left[0, q_i^{\max}\right]\to \mathbb{R}_+$ is continuous and nonincreasing in quantity $q$, where each bound $q_i^{\max}\in(0,\infty)$ is an exogenous upper bound on feasible quantities (representing physical capacity, infrastructure limits, or maximum historical consumption), distinct from the demand-derived $\bar{q}_i = D_i(\bar{p})$ used elsewhere. A binding ceiling $\bar{p} > 0$ is imposed. We observe the delivered allocation
$q^{\mathrm{obs}}=(q^{\mathrm{obs}}_1,\dots,q^{\mathrm{obs}}_n)$, where $$\bar Q \coloneqq \sum_{i=1}^n q^{\mathrm{obs}}_i \in \left[0,\sum_i q_i^{\max}\right]$$.
\begin{assumption}
    For all $i$, $P_i$ is nonincreasing and satisfies    
\begin{enumerate}
    \item \label{ass:a1} $P_i(q_i^{\mathrm{obs}})=p_{0,i}=\bar p-b_i \quad\text{with}\quad b_i\in[\underline b_i,\bar b_i].$
  \item \label{ass:a2} There exist $g_{i,L} < g_{i,U}<0$ with $g_{i,L}\le P'_i(q)\le g_{i,U}$ for a.e. $q\in\left[0, q_i^{\max}\right]$.
  \item \label{ass:a3} $P_i(0)\le M_i<\infty$
  \item \label{ass:a4} $P_i\in W^{1,\infty}\left(\left[0, q_i^{\max}\right]\right)$, the space of bounded Lipschitz-continuous functions on $\left[0, q_i^{\max}\right]$.
\end{enumerate}
\end{assumption}

The first three assumptions are the economically meaningful ones, whereas the last one is technical, guaranteeing that the various objects we construct are well-defined. Let $\mathcal P_i$ be the set of $P_i$ satisfying Assumption \ref{mainass} and $\mathcal P \coloneqq \times_{i=1}^n\mathcal P_i$. We assume $\mathcal P\neq\varnothing$.

Let \(q_i(\cdot)\) denote the left-continuous generalized inverse of inverse demand \(P_i(\cdot)\):
\[
q_i(p)\coloneqq
\begin{cases}
\inf\left\{x\in\left[0,q_i^{\max}\right]\colon P_i(x)\le p\right\},&\text{if the set is nonempty},\\
q_i^{\max},&\text{otherwise}.
\end{cases}
\]
Let \(\Phi(P)\) denote the misallocation loss:
\[
\Phi(P)\coloneqq\sum_{i=1}^n\int_{q_i^{\mathrm{obs}}}^{q_i^{*}(P)}\left(P_i(x)-p^{*}(P)\right)d x,
\]
where \(q^*(P)\) is the equal-shadow-price split of total quantity \(\bar{Q}\) and \(p^*(P)\) is its Lagrange multiplier. Define
\[
\overline{\Phi} \coloneqq \sup_{P\in\mathcal{P}} \Phi(P)
\qquad\text{and}\qquad
\underline{\Phi} \coloneqq \inf_{P\in\mathcal{P}} \Phi(P).
\]
We sketch the result here, as the formal result is quite notation heavy. The fully detailed statement and its proof are given in the appendix.

Within each market, the logic follows \citet{kang_robustness_2025}: extremal welfare is achieved by demand curves at the slope bounds. Our contribution is recognizing that efficient allocation equalizes shadow prices across markets, so all markets share a common $p^*$. This condition, combined with the adding-up constraint, converts the infinite-dimensional problem over demand functions into a one-dimensional search over the scalar $p^*$.

\begin{proposition}[Informal]\label{prop:reductioninformal}Posit an interiority assumption. Then there exist extremizers \(P^\ast\in\mathcal{P}\) attaining \(\underline{\Phi}\) and \(\overline{\Phi}\) such that each \(P_i^\ast\) is continuous and piecewise affine, with \(P_i^{\ast\prime}(q)\in\left\{g_{i,L},g_{i,U}\right\}\) for a.e. \(q\). Moreover, on the quantity range that matters for the bound (between \(q_i^{\mathrm{obs}}\) and the equal-shadow allocation), each market can be taken to have at most one kink.

Computing \(\underline{\Phi}\) and \(\overline{\Phi}\) reduces to a one-dimensional search over a candidate common shadow price \(p\). For each fixed \(p\), the problem decomposes across markets except for the single adding-up constraint \(\sum_{i=1}^n q_i(p)=\bar{Q}\).
\end{proposition}

Proposition~\ref{prop:reductioninformal} is the heart of this robust bounds approach. It turns the bound-finding problem over the huge class \(\mathcal{P}\) into a one-dimensional program over \(p \in \mathcal{I}\), and shows that the optimizer hits the slope endpoints everywhere except at kinks. The key simplification is that, conditional on \(p\), the objective is \textit{linear} in the inverse-quantity functions \(q_i(\cdot)\), so the optimum must lie at the pointwise boundary of the admissible set.

\begin{theorem}[Informal]\label{thm:extremalsinformal}
    Posit an interiority assumption. Then, both \(\overline{\Phi}\) and \(\underline{\Phi}\) are attained by some \(P^*\in\mathcal{P}\). Moreover, each \(P_i^*\) may be chosen continuous and piecewise affine on \(\left[0,q_i^{\max}\right]\), with a.e. derivative \(P_i^{* \prime}(q)\in\left\{g_{i,L},g_{i,U}\right\}\), and satisfying \(P_i^*(0)\le M_i\).
\end{theorem}

\noindent The proof appears in Appendix~\ref{app:extremals-proofs}.

\begin{corollary}\label{prop:anchors}
Suppose we also optimize over $p_{0,i}\in\left[\bar p-\bar b_i, \bar p-\underline b_i\right]$. If the candidate shadow price is at least as large as the price implied by the observed quantity in market $i$; \textit{viz.}, $p \geq P_i(q_i^{\mathrm{obs}})$ and none of the upper or lower bounds on quantities bind over the relevant price range, the upper bound is produced by $b_i = \bar{b}_i$ and the lower bound by $b_i = \underline{b}_i$.
\end{corollary}

In all empirical robust-bound results below, we implement the interval-anchor extension of Corollary~\ref{prop:anchors}: $p_{0,i}$ is allowed to vary over its admissible interval.

Putting these results together, the computation proceeds as follows:
\begin{enumerate}
    \item For each market $i$, the slope bounds and anchor $(q_i^{\mathrm{obs}},  p_{0,i})$ pin inverse demand within a cone: any admissible $P_i$ must lie between the steepest and flattest lines through the anchor. Inverting this cone yields the quantity envelopes $\ell_i(p)$ and $u_i(p)$: the smallest and largest quantities market $i$ could absorb at shadow price $p$ under any admissible demand. When a choke cap $P_i(0) \le M_i$ binds, the upper envelope is trimmed further (Theorem~\ref{thm:extremalsinformal}).
    \item The set of candidate common shadow prices is $\mathcal{I} = \{p : \sum_i \ell_i(p) \le \bar{Q} \le \sum_i u_i(p)\}$, i.e.\ the prices at which per-market quantity intervals are jointly compatible with the aggregate constraint.
    \item For each candidate $p \in \mathcal{I}$, the conditional problem pushes each market to a boundary envelope, but the relevant endpoint depends on both the bound direction (upper versus lower) and whether $p$ lies above or below $p_{0,i}$; equivalently, each market has a boundary target $c_i(p)\in\{\ell_i(p),u_i(p)\}$.
    \item These boundary targets need not satisfy adding-up exactly, i.e.\ $\sum_i c_i(p)$ need not equal $\bar{Q}$. Feasibility is restored by solving a convex quadratic projection onto $\{x:\sum_i x_i=\bar Q,\ \ell_i(p)\le x_i\le u_i(p)\}$ that penalizes deviations from $c_i(p)$, with weights $\kappa_i = 1/g_{i,L} - 1/g_{i,U}$. Markets with wider slope uncertainty (larger $\kappa_i$) absorb more of this adjustment.
    \item The resulting value $\Phi(p)$ is optimized over $p \in \mathcal{I}$ to obtain $\overline{\Phi}$ (or $\underline{\Phi}$). When anchor prices are interval-valued, this becomes a joint optimization over $p$ and the anchor vector $p_0$ on the admissible hyper-rectangle \(\prod_i [p_{0,i}^{L},p_{0,i}^{U}]\) (Corollary~\ref{prop:anchors}).
\end{enumerate}

Theorem~\ref{thm:extremalsinformal} establishes that the bounds $[\underline\Phi, \overline\Phi]$ are \textit{sharp}: they are attained by demand functions in the permitted class. This sharpness distinguishes partial identification from sensitivity analysis. Sensitivity analysis with linear demand would vary the slope $b$ and report how misallocation changes, yielding the same numerical interval when the choke-price constraint is non-binding. With a binding $P(0) \le M$, the robust bounds can differ because the extremal piecewise-affine demands can kink to satisfy the choke while steeper linear extrapolations cannot. But sensitivity analysis cannot rule out the possibility that some nonlinear demand (convex, concave, or oscillating within the slope bounds) produces misallocation outside that interval.

The theorem closes this gap. By the fundamental theorem of calculus, any demand function with $P'(q) \in [g_L, g_U]$ must lie within a cone emanating from the observed point. Since misallocation is an integral of the demand curve, it inherits these bounds. The theorem confirms that the extreme misallocation is achieved at the boundary of this cone, by piecewise affine demand with slopes in $\{g_L, g_U\}$, not by some exotic function in the interior. The bounds from linear demand are, therefore, not conservative approximations, but are exact over the infinite-dimensional class of all Lipschitz functions satisfying the slope constraints.

\section{Illustrating Chaos and Misallocation: The 1973--74 Gasoline Crisis}
\label{sec:empirical}

In August of 1971, President Nixon froze all wages and prices in the United States. Many prices were unfrozen after the 90-day scheduled period, but petroleum remained subject to controls with various modifications until August of 1973 when Special Rule No. 1 established mandatory controls for petroleum specifically \footnote{See \citet{bradley_oil_1997} for a definitive account of the crisis.}. 

Price controls on oil and the sporadic shortages were thus already in place well before the Yom-Kippur war and the subsequent Arab Oil Embargo which began in October of 1973. The embargo, however, dramatically increased the market price of oil above the controlled price, transforming what had been moderate distortions into severe shortages. From October of 1973 to January of 1974 the price of oil tripled. Over the first shock, quantity consumed fell by about 9\% relative to trend \citep{yergin_prize_1991}.

Quickly following the embargo, the Emergency Petroleum Allocation Act of 1973 (EPAA) ratified President Nixon's earlier executive orders with legislated price controls and an allocation system based pro-rata on 1972 levels. If total supply fell to 90\% of 1972's volume, for example, each buyer would receive 90\% of their 1972 allocation. This basic system was then adjusted with numerous exceptions, prioritizing industries such as national defense, essential services, agriculture, and independent refiners, all under an overarching "fair and equitable" guideline \citep{bradley_oil_1997}.

We begin with a formal analysis of the most visible portion of the petroleum allocation under price controls. Namely, the allocation of gasoline across geographies. We then turn to an informal discussion of misallocation across time, sector and other segmentations.

\subsection{AAA Data}

We digitize station-level survey data collected by the AAA during February 1974. Figure \ref{fig:rationing_map} displays the cross-sectional variation in rationing (out-of-fuel closures + purchase limits). The mean out-of-fuel rate by state was 8.1\%, but this average masks striking heterogeneity. The survey reports that 10.1\% of stations were closed (out of fuel), 27.6\% were limiting purchases, and 62.3\% were operating normally. We combine this station-level classification with 1972 gasoline station counts and sales data from the U.S. Census of Retail Trade. Our sample includes 48 states: we exclude Alabama, Alaska, and DC.\footnote{Alabama and Alaska lacked AAA survey coverage.}

The AAA data are corroborated by extensive historical accounts of the misallocation chaos. In Maine, for example, heating oil was in such short supply that the state Office of Civil Defense directed each community to prepare emergency shelters \citep{turkel_2023}. In Maryland and Connecticut, lines for gasoline stretched up to five miles long. Yet, \textit{at the same time}, Texas, the Deep South, and the Great Plains were "virtually awash with gasoline" \citep{time_shortages_1974}. Disparities this stark do not occur in an integrated market.\footnote{We focus on the 1972-1974 crisis but the same patterns were in evidence in 1979 when an even smaller national quantity reduction ${\sim}3.5\%$ led at its height in June of 1979 to the closing of almost every gasoline station in New Jersey, Connecticut and New York City \citep{gupte_fuel_1979}.}

\subsection{Welfare Measures}

The Harberger deadweight loss measures the welfare cost of the aggregate quantity reduction, assuming the reduced supply is allocated efficiently:
\[\mathcal{L}_{Harb} = \widetilde{W}(q^{base}) - \widetilde{W}(q^*) - p^{base} \cdot (Q^{base} - \bar{Q})\]
Here $\widetilde{W}(q) \coloneqq \sum_i \int_0^{q_i} P_i(x)  dx$ denotes gross surplus (the area under demand), in contrast to the net surplus $W$ defined earlier. The final term subtracts baseline expenditure on the missing quantity at $p^{base}$, yielding the standard Harberger triangle benchmark. We express these welfare losses as percentages of \textit{baseline expenditure}: total consumer spending before the crisis, $p^{base} \times Q^{base}$. With our normalization ($p^{base} = 1$, $Q^{base} = 1$), baseline expenditure equals unity. This metric provides economic interpretation: a misallocation loss of 3.6\% means consumers lost welfare equivalent to 3.6\% of their pre-crisis gasoline spending. For our calibration, with a 9\% aggregate shortage ($\bar{Q} = 0.91$) and $\bar{p} = 0.8$, the Harberger triangle is approximately $\mathcal{L}_{Harb} \approx 2.0\%$ of baseline spending.

The misallocation loss $\mathcal{L}_{Mis} = W(q^*) - W(q^{obs})$ captures the additional welfare destruction from inefficient distribution of the reduced supply, as defined in Section~\ref{def:misalloc}. Our key statistic is the \textit{Misallocation Ratio}:
\[\mathcal{R} = \frac{\mathcal{L}_{Mis}}{\mathcal{L}_{Harb}}\]
A ratio $\mathcal{R} = 1$ means misallocation doubles the welfare cost relative to a standard Harberger analysis; $\mathcal{R} > 1$ means misallocation is the dominant source of loss. With efficient rationing, $\mathcal{R} =0.$

\subsection{Allocation}
\label{sec:allocation}

We aggregate stations into two markets: open (62.3\%) and closed/limiting (37.7\%).  ''Out of fuel'' is a point-in-time status from the AAA survey, not a period-average of $q=0$. Our two-market calibration aggregates closed/limiting stations into a single partially-served submarket. The key economic mechanism is that the controlled price $\bar{p}$ lies below the market-clearing price. At $\bar{p}$, quantity demanded exceeds quantity supplied. Open stations satisfy their customers; closed and limiting stations cannot.

Under a vertex allocation, open stations consume their full demand at the controlled price while closed stations receive the residual. We assume $\bar{p} = 0.8$, consistent with \citep{frech_welfare_1987}. With elasticity bounds $\varepsilon \in [0.2,0.4]$ \citep{hughes_evidence_2008}, baseline-anchored linear demand implies
\[
q_O \in [1.04, 1.08].
\]
Closed/limiting stations receive the residual
\[
q_C = \frac{\bar{Q} - (1-s)\cdot q_O}{s}
\quad\Rightarrow\quad
q_C \in [0.63, 0.70].
\]
For the baseline calibration, we use $\varepsilon=0.3$, which produces $q_O=1.06$ and $q_C\approx0.66$. This allocation satisfies the aggregate constraint and exhibits exactly the corner-solution structure predicted by the Chaos Theorem.

\subsection{Station-Level Robust Bounds}
\label{sec:station-robust}

Two features of our setting make partial identification essential. First, extreme allocations are where functional form does most of the work. Our welfare calculations integrate inverse demand over quantity ranges extending to $q \approx 0.66$ for closed/limiting stations. Elasticity estimates come from observed price-quantity variation around market equilibrium, identifying the slope of demand locally near $q = 1$, not at $q = 0.66$. To compute welfare at these allocations, we must extrapolate beyond where demand is identified. Linear extrapolation from an elasticity of 0.2 implies $P(0) = 6$; from 0.4, it implies $P(0) = 3.5$. CES demand diverges entirely as $q \to 0$, requiring an arbitrary choke-price normalization. In our calibration, this extrapolation uncertainty moves misallocation estimates by many percentage points of baseline spending (from roughly 2.3\% to 18.6\% in the pooled no-choke benchmark).

The theory predicts allocation dispersion, the data show dispersion, and dispersion is where functional form assumptions bite hardest. We do not want our welfare estimates hostage to how we extrapolate demand into regions no estimation has ever reached. The robust bounds approach sidesteps this problem. We impose only local, transparent restrictions (slope bounds from elasticity estimates and a finite choke price) and optimize over all inverse demands consistent with these assumptions.

The elasticity bounds $\varepsilon \in [0.2, 0.4]$ imply slope bounds $g \in [-5, -2.5]$.\footnote{Throughout, ''steep'' and ''flat'' refer to the inverse demand curve $P(Q)$. A steep inverse demand (large $|g|$, low elasticity) means price falls rapidly as quantity rises. This corresponds to a \emph{flat} demand curve $Q(P)$, where quantity responds little to price changes. The figures plot $P$ against $Q$, so visual steepness matches $|g|$.} To visualize how these bounds translate into welfare uncertainty, we aggregate stations into a pooled two-market benchmark (open versus closed/limiting) where the extremal demand curves can be drawn directly. Figure~\ref{fig:station_joint_extremals_nochoke} shows the demand curves consistent with the slope bounds and observed allocation when no choke constraint is imposed. Solid segments are pinned down by the data and slope restrictions; dashed segments are admissible extensions.

\begin{figure}[htbp]
    \centering
    \includegraphics[width=\textwidth]{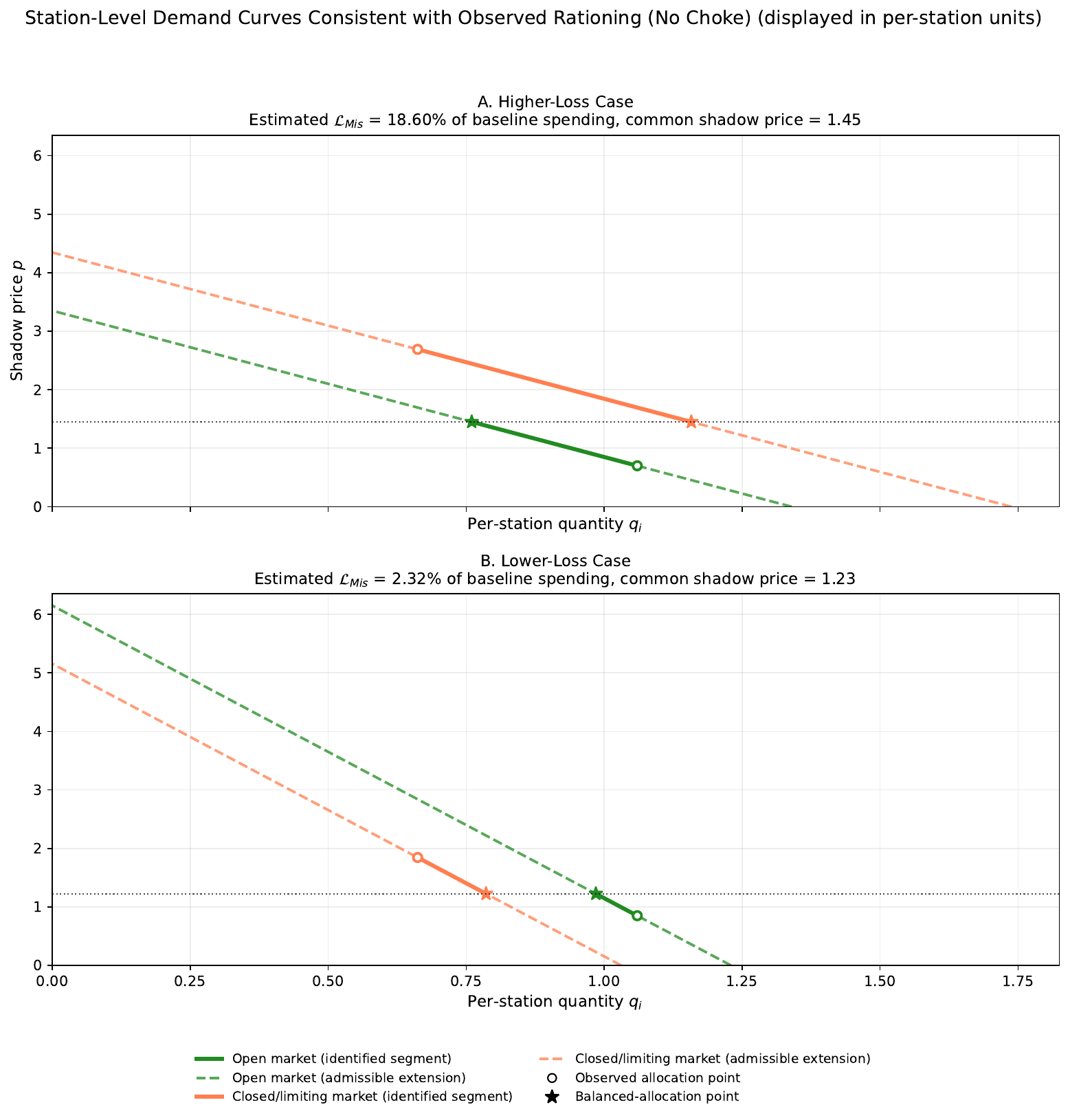}
    \caption{\textbf{Station-Level Demand Curves Without Choke.} The two panels show high-loss and low-loss demand configurations consistent with observed open and closed/limiting station quantities when no choke cap is imposed, using the baseline calibration $q_O=1.06$ and $q_C=0.66$. Solid segments are identified by the observed allocations and slope bounds; dashed segments are admissible extensions.}
    \label{fig:station_joint_extremals_nochoke}
\end{figure}

Figure~\ref{fig:station_joint_extremals_choke} adds a choke constraint $P(0)\le M=4$. The choke operates through two channels. First, a geometric kink: in the lower-loss panel, the steep branch would extrapolate to $P(0)=6>4$, so the admissible curve must bend to hit the choke cap. Second, an anchor-price restriction: even in the higher-loss panel, where both markets use the flat slope and extrapolate to $P(0)=3.5<4$ (no visible kink), the requirement $P(0)\le M$ narrows the set of admissible anchor prices $p_{0,i}$, trimming the closed/limiting upper shadow-price bound at observed quantity (about $2.60\rightarrow 2.34$) and lowering the upper welfare endpoint. We, therefore, present no-choke results as the headline interval and use the choke case as a disciplined robustness check.

\begin{figure}[htbp]
    \centering
    \includegraphics[width=\textwidth]{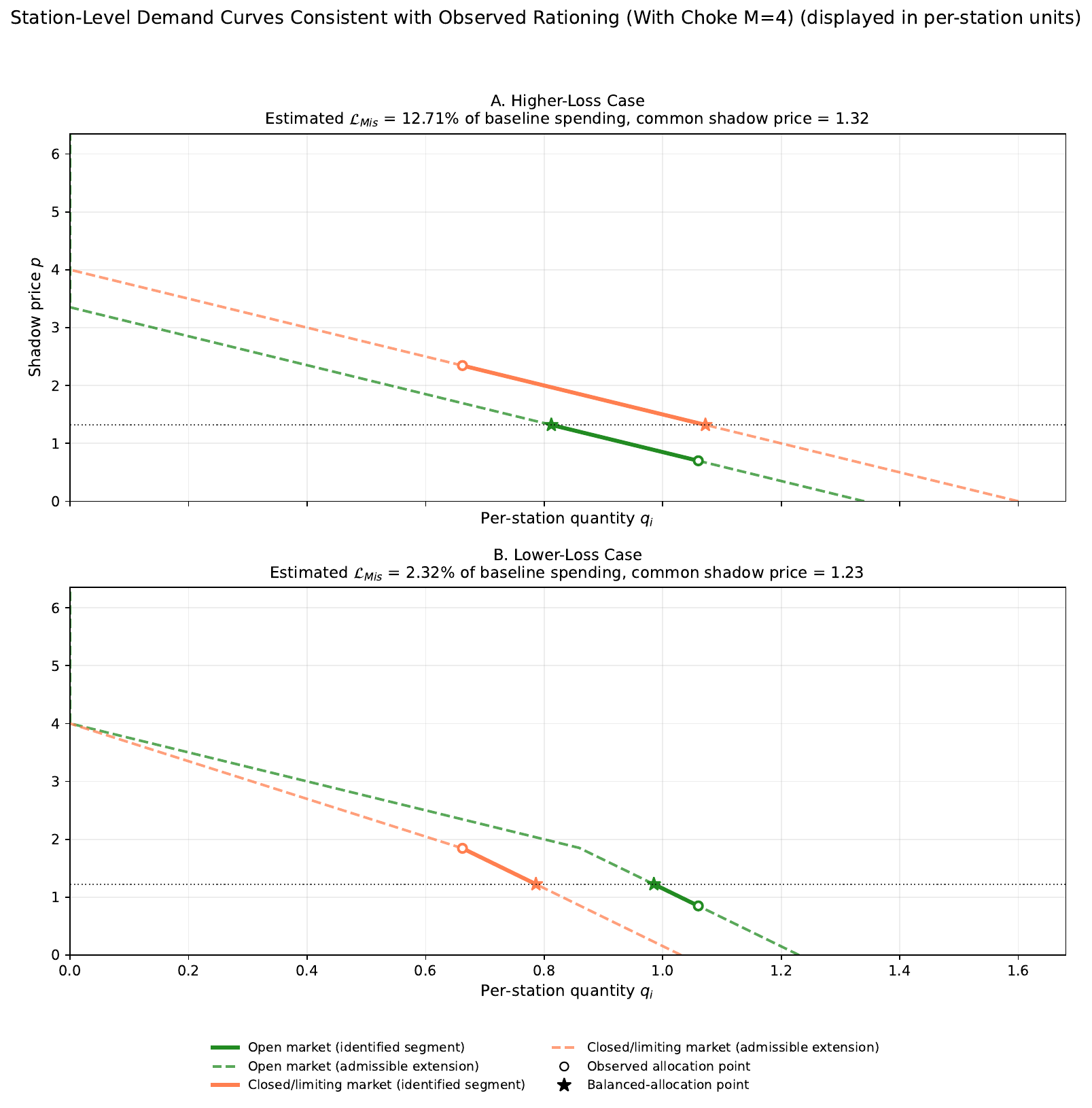}
    \caption{\textbf{Station-Level Demand Curves With Choke.} Same construction as Figure~\ref{fig:station_joint_extremals_nochoke}, using the same baseline calibration and now imposing $P(0)\le M=4$. The choke cap narrows admissible shadow-price ranges, which tightens the upper bound on misallocation losses.}
    \label{fig:station_joint_extremals_choke}
\end{figure}

The pooled two-market benchmark visualized above implies no-choke bounds of $\mathcal{R}\in[1.15, 9.18]$. This is a useful diagnostic object because the extremal demand curves can be visualized directly. Our main empirical specification is different: it disaggregates to the state-by-status level, treating open and non-open stations within each state as separate markets and imposing adding-up jointly across those cells. Each state's weight is proportional to its gasoline consumption (gallons), not its station count, so high-consumption states carry more weight than their station share alone would imply.

\subsection{State-Level Shadow Prices}
\label{sec:state-robust}

Figure~\ref{fig:rationing_map} showed substantial geographic variation in rationing severity. We can translate this variation into state-level shadow prices using the state-by-status robust model. For each state $i$ with rationing share $r_i$ (limiting plus out-of-fuel), the state-average shadow price is
\[
\bar p_{0,i}=(1-r_i) p_{0,i}^{open}+r_i p_{0,i}^{nonopen}.
\]
Figure~\ref{fig:shadow_price_map} maps these state averages for the upper-bound (top) and lower-bound (bottom) configurations. States with high rationing shares (Connecticut, Massachusetts) load heavily on the non-open shadow price and appear with higher state-average shadow prices; low-rationing states (Idaho, Montana, Wyoming) load mostly on the open shadow price and appear lower.

\begin{figure}[htbp]
    \centering
    \includegraphics[width=\textwidth]{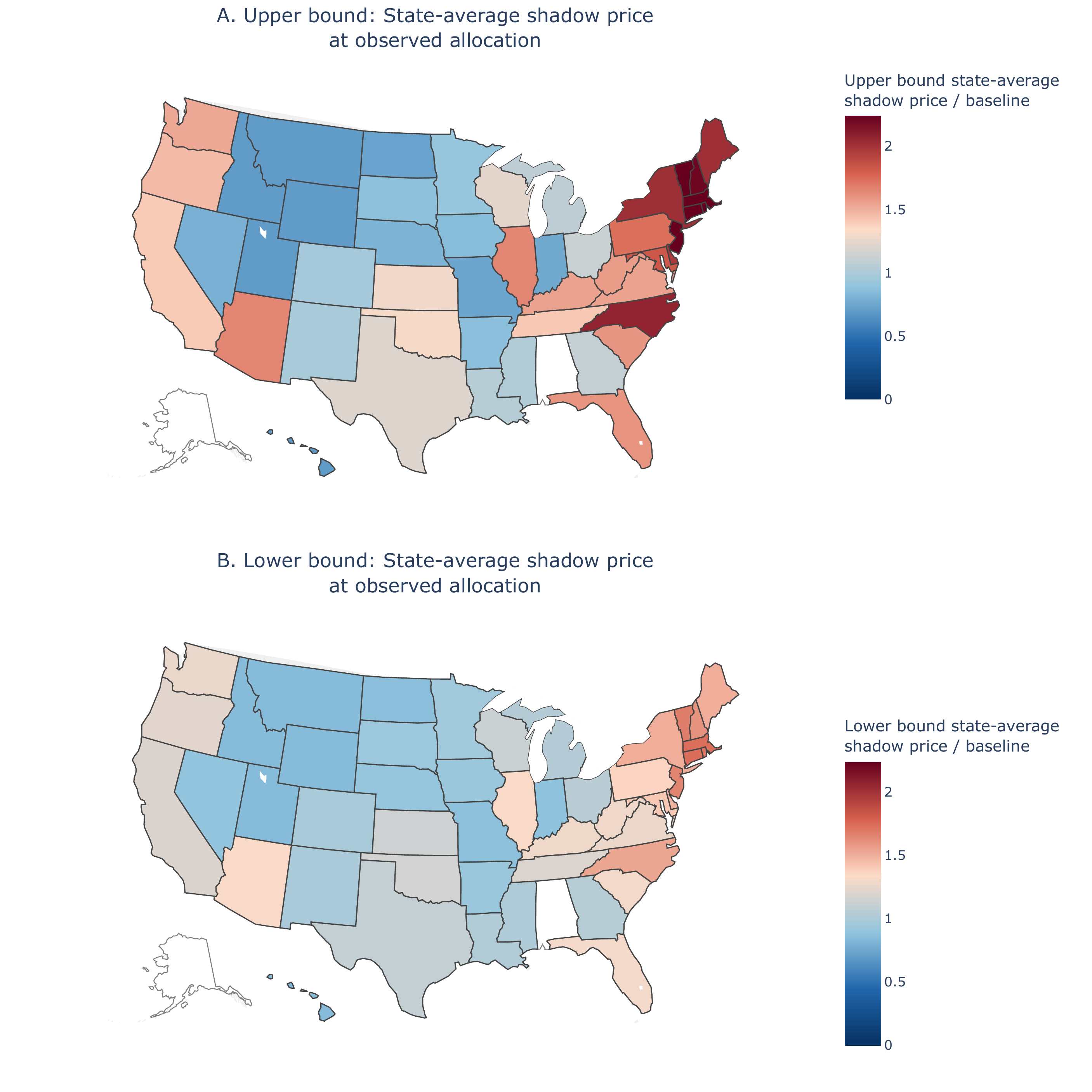}
    \caption{\textbf{State-Average Shadow Prices.} Top panel: upper bound. Bottom panel: lower bound. Colors show each state's average shadow price, computed as the rationing-share-weighted average of open and non-open shadow prices from the state-by-status robust bounds.}
    \label{fig:shadow_price_map}
\end{figure}

The geographic pattern mirrors Figure~\ref{fig:rationing_map}. Where rationing was severe, shadow prices are higher; where rationing was limited, shadow prices are lower. This is the open/non-open mechanism visualized in geographic form.

The economic meaning of these shadow prices is direct. In the upper-bound configuration, the state-average shadow price in Connecticut is 2.5 times the baseline price versus 0.7 times baseline in Montana, where no stations reported shortages. Connecticut consumers valued gasoline at 3.5 times what Montana consumers were paying. Shipping a barrel from Montana to Connecticut would have more than tripled your money. The arbitrage opportunity that price controls created but prevented.

This state-by-status specification uses 91 active state-status markets rather than the pooled two-market benchmark. For each active cell, the anchor price $p_{0,i}$ is allowed to vary over its admissible interval (implied by baseline anchors and slope bounds), and we numerically search jointly over these intervals and the common shadow-price condition. Relative to pooled two-market bounds, this additional structure slightly tightens both endpoints:
\[
\mathcal{L}_{Mis}/\mathcal{L}_{Harb}\in[1.09, 8.75].
\]
Since $\mathcal{L}_{Harb}\approx2.0\%$ of baseline spending, this ratio implies $\mathcal{L}_{Mis}\in[2.2\%, 17.7\%]$. Even at the lower endpoint, misallocation losses exceed the Harberger triangle; at the upper endpoint, misallocation remains close to an order of magnitude larger.

This slight tightening is reasonable but not mechanical. The state-by-status model adds many cell-level feasibility restrictions that must all be satisfied jointly with a common shadow price and adding-up, which can reduce the set of feasible extremal reallocations relative to a pooled two-market benchmark. At the same time, interval anchor uncertainty in $p_{0,i}$ adds flexibility, so monotonic tightening is not guaranteed ex ante. In our calibration, the state-status primitives are fairly similar across states, so the difference from pooled bounds is modest rather than dramatic.

Table~\ref{tab:assumption_decomposition} decomposes the identified interval for $\Phi$ by assumption. Quantities are calibrated once: $q_{open}=1.06$ and $q_{non\text{-}open}=0.67$ follow from $\varepsilon_{open}=0.3$ and the 9\% shortage, and are held fixed throughout. What varies is the \emph{welfare elasticity}, the demand slope used to compute the surplus integral.

Row~1 assigns every market the same welfare elasticity. Sweeping it over $[0.2,\,0.4]$ changes $\Phi$ by a factor of two but leaves $\mathcal{R}$ invariant at 4.4 because the same slope scales both misallocation loss and the Harberger triangle. Row~2 allows each market its own slope within the same range, which modestly widens the identified interval for $\Phi$. Allowing interval-anchor uncertainty (row~3) is the main widening step. Adding a choke cap at $M=4$ (row~4) leaves the lower endpoint unchanged but trims the upper endpoint, because the lower-loss extremal already satisfies the cap while the high-loss extremal is where the choke binds. The identified object in the table is always $\Phi$; the $\mathcal{R}$ column is a normalization used to put these losses on Harberger units. In row~1 that normalization is matched to the same common elasticity, while in rows~2--4 it is reported relative to the conservative Harberger benchmark with $\varepsilon=0.2$.

\begin{table}[htbp]
    \centering
    \small
    \begin{tabular}{lccc}
        \toprule
        Assumptions & $\underline{\Phi}$ (\%) & $\overline{\Phi}$ (\%) & $\mathcal{R}$ \\
        \midrule
        Common $\varepsilon\in[0.2,\,0.4]$ (fixed anchors) & 4.43 & 8.86 & 4.37 \\
        Heterogeneous $\varepsilon_i\in[0.2,\,0.4]$ (fixed anchors) & 4.98 & 9.97 & [2.46, 4.92] \\
        + Anchor uncertainty ($p_{0,i}$ intervals) & 2.21 & 17.72 & [1.09, 8.75] \\
        + Choke constraint ($M=4$) & 2.21 & 12.40 & [1.09, 6.12] \\
        \bottomrule
    \end{tabular}
    \caption{\textbf{Assumption-to-Interval Decomposition.} Quantities calibrated at $\varepsilon_{open}=0.3$ ($q_{open}=1.06$) and held fixed; welfare elasticity (demand slope for surplus computation) varies as indicated. $\Phi$ is the identified object and is reported as misallocation loss relative to baseline spending. $\mathcal{R}\equiv\mathcal{L}_{Mis}/\mathcal{L}_{Harb}$ is shown only as a normalization: row~1 uses the same common $\varepsilon$ in numerator and denominator, while rows~2--4 divide by the conservative Harberger benchmark with $\varepsilon=0.2$.}
    \label{tab:assumption_decomposition}
\end{table}

These magnitudes are qualitatively consistent with the queuing equilibria documented by \citet{deacon_rationing_1985}. When money prices cannot ration demand, time prices must rise to clear the market. At the height of the crisis, hour-long waits were routine at hard-hit stations, with multi-hour waits common at peak times. \citet{deacon_rationing_1985} estimate the value of waiting time at roughly the after-tax wage rate. For a typical motorist, a two-hour wait represents a time cost comparable to the fill-up itself, effectively doubling or tripling the price paid. The upper-bound shadow prices in our framework, where closed-station consumers face valuations well above the controlled price, correspond to exactly such queuing equilibria.

\subsection{Misallocation beyond Geography}

Misallocation during the energy crisis went far beyond geographic misallocation there was misallocation across product mix, time, input-output segments and sector. We give some examples to illustrate the generality of "segment" and "transport".

A barrel of crude can be refined into a slate of outputs: propane, butane, gasoline, jet fuel, heating oil, petrochemical feedstocks, and others. On the margin, tradeoffs exist: a given quantity of crude oil can yield more gasoline at the expense of less fuel oil, or vice versa. When prices are controlled product-by-product, these tradeoffs interact with small ''nuisance'' wedges (refining costs, accounting conventions, transportation frictions, and priority categories) to generate large shifts in the realized fuel mix (recall the quadratic cost function discussion in Section \ref{sec:Chaos Result}). During the controls, scarcity, therefore, appeared not only across space but across products. Propane shortages, diesel shortages, and heating-oil shortages rarely arrived as a single, stationary ''energy shortage;'' instead the binding constraint rotated across fuels as demand and costs moved, as the chaos theorem predicts.

Heating oil illustrates both product-mix and temporal misallocation. Heating-oil prices were initially frozen at (low-demand) summer 1971 levels, weakening incentives to carry inventories into winter and to direct marginal refinery yield toward distillate. Shortages emerged in the winters of 1972--1973 and 1973--1974, contributing to school closures. Freezing prices led to freezing people \citep{poole_fuel_1973,verleger_us_1979, deacon_price_1980}.

Input–output misallocation also arose outside fuels whenever controlled output prices met uncontrolled input costs (or vice versa); \citet{mulligan_supply_chain_2025} formalizes this supply-chain mechanism. In the summer of 1973, chicken farmers gassed, drowned, and suffocated roughly a million baby chicks \citep{time_controls_1973,new_york_times_baby_1973,associated_press_chicks_1973}. Retail chicken prices were controlled, but feed costs were not. That did more than reduce poultry output: it blocked the price system from allocating supply across production stages. Think of the chick stage as market 1, transport as an intertemporal production stage and chickens as market 2. Chicks fed now becomes broilers weeks later. With revenues capped along that grow-out path, relatively small increases in feed costs pushed producers to a corner: abandon the “raise chicks into chickens” segment and liquidate immediately. As Madison Clements of Waco put it, ``It's cheaper to drown em than to put em down and raise em," \citep{associated_press_chicks_1973}.

The same temporal segmentation appeared in livestock where the markets segments were even more well-defined: dairy farmers slaughtered cows and hog farmers culled breeding stock, temporarily increasing meat supply while reducing future flows of milk and pork. In each case, controls flattened relative returns across time and uses, so small cost wedges selected a vertex allocation, with all supply routed to the "now" segment and essentially zero to the "later" segment.

Finally, supply-chain complexity amplified inefficiencies in the output mix. Allocation rules designated priority end uses, but policymakers underestimated input-output linkages. Oil production was prioritized, yet some upstream inputs were not. Propane, for instance, was essential for producing plastic piping used in oil extraction, but the plastics industry initially lacked priority designation. The resulting pipe shortages disrupted the very oil production that allocation policy sought to protect \citep{federal_energy_office_review_1974}.

\section{From Price Controls to Quantity Controls}\label{sec:pricetoquantity}

Allocation by fiat is a natural consequence of price controls because, as the Chaos Theorem indicates, once prices are prevented from doing their work, “the market” no longer selects a coherent allocation. Instead, it delivers corner outcomes with no welfare ordering. Price controls, therefore, tend to metastasize into quantity controls. In that sense, the price control is the ultimate cause of the misallocation; administrative quantity allocation is the proximate response.

The debates during the oil crisis make this progression visible. What triggered the demand for mandatory allocation was not scarcity or shortages alone, but the visibly uneven and disruptive pattern of goods with zero allocations that emerged once price-guided adjustment was suppressed. Walter Mondale, for example, testified to the Joint Economic Committee in 1973 \citep{jec1973gasoline}:
\begin{quote}  
[T]he voluntary plan simply is not working. In our state, over 150 independent service stations have closed, and the number that have been able to reopen since the voluntary plan went into effect has been small. [\ldots]  Major oil suppliers, including Gulf and Sun Oil have indicated their desire to pull out of the Minnesota market [\ldots] Just yesterday, we began debate in the Senate on legislation to impose a mandatory allocation system. I believe that the time to plan for this type of system has come, and that only strong measures will work to ensure that farmers, governments and other consumers have adequate supplies over the coming months and that prices for fuel products are kept down...We need this action now at the Federal level--and I am confident that we will get it.
\end{quote}
Even more directly, a report presented to the Federal Energy Administration in 1974 by the New England Fuel Institute \citep{zotero-item-8347} argued against lifting the Mandatory Allocation Program without also lifting price controls, precisely because price controls without quantity allocation would produce utter "chaos" and "confusion." 

\begin{quote}
The main reason for this chaotic and confused supply situation that was depriving many  distributors of any normal or reasonable supply, or any supply at all, was due to the fact that price controls were in force, but there was no Mandatory Allocation Program.
\end{quote}

The Mandatory Allocation Program was an attempt to solve a problem that price controls had created. The misallocation we measure under the EPAA is a downstream consequence of the Chaos Theorem, not a separate phenomenon.

\section{Conclusion}

In markets without price controls, arbitrage pushes goods toward their highest-value uses, equalizing shadow prices across space, time, and sector. A binding price ceiling breaks this. Once money prices are frozen, the incentive to move goods to higher-value segments vanishes, and the efficient allocation becomes a knife-edge point inside a much larger feasible set.

The Chaos Theorem formalizes the consequence. Equilibrium generically lands on corners where some markets are fully served and others receive nothing. Infinitesimal changes in nuisance parameters such as transportation costs can flip the economy between corner allocations, and discontinuous welfare jumps follow from continuous parameter shifts. The patchwork of the 1973–74 gasoline crisis is not an accident but the generic outcome when prices cannot signal scarcity.

Our robust bounds analysis quantifies these welfare costs without parametric demand assumptions. Using only observed allocations, the ceiling price, and empirically plausible slope bounds, we find misallocation losses on the order of 1 to 9 times the Harberger triangle; the quantity reduction accounts for a minority of total welfare loss.

The same logic extends beyond geography. Heating-oil prices frozen at summer levels eliminated incentives to store winter inventories; controlled chicken prices with uncontrolled feed costs drove farmers to a corner on the input–output, intertemporal margin, leading them to destroy over a million chicks. Whenever controls suppress price variation across segments, seller indifference turns small cost wedges into large misallocations. The resulting chaos generates demand for direct quantity controls.

The broader lesson is that whenever a ceiling fragments an integrated market--gasoline, rental housing, agriculture, medical care--the main cost is not the familiar triangle, but the hidden misallocation behind it.

\appendix
\onehalfspacing
\section{Omitted Proofs}\label{omittedproofs}

\subsection{Theorem \ref{thm:worst-vertex} Proof}\label{app:vertex-proof}
\begin{proof}[Proof of Theorem \ref{thm:worst-vertex}]
    Because $\widetilde W_i(q_i)=\int_0^{q_i}P_i(x) dx$ is concave (as an integral of a nonincreasing function), the sum $\widetilde W=\sum_i \widetilde W_i$ is concave on the convex polytope $\mathcal{F}$. A concave function attains its minimum at an extreme point, proving \ref{worstitem1}.

    For \ref{worstitem2}, fix a minimizer \(q^w\in \mathcal{F}\) of \(\widetilde{W}\). Since each \(P_i\) is continuous
on \(\left[0,\bar{q}_i\right]\), \(\widetilde{W}\)
is continuously differentiable on  \(\mathcal{F}\), with
\(\frac{\partial \widetilde{W}}{\partial q_i}(q) =
P_i(q_i)\) for all \(q\in \mathcal{F}\). The standard Karush-Kuhn-Tucker conditions apply: at \(q^w\), there exist
\(\lambda,\{\mu_i\},\{\nu_i\}\) such that
\begin{align}
\text{(stationarity)}\qquad
& P_i(q_i^w) + \lambda - \mu_i + \nu_i = 0
&&\text{for all } i,
\label{eq:KKT-stationarity}\tag{\(A1\)}
\\
\text{(complementary slackness)}\qquad
& \mu_i q_i^w = 0,
\qquad
\nu_i (q_i^w - \bar q_i) = 0
&&\text{for all } i,
\label{eq:KKT-slackness}\tag{\(A2\)}
\\
\text{(primal feasibility)}\qquad
& \sum_i q_i^w = \bar Q,
\quad
0 \le q_i^w \le \bar q_i
&&\text{for all } i,
\label{eq:KKT-primal}\tag{\(A3\)}
\\
\text{(dual feasibility)}\qquad
& \mu_i,\nu_i \ge 0
&&\text{for all } i.
\label{eq:KKT-dual}\tag{\(A4\)}
\end{align}
We define \(\lambda_{\mathrm{worst}} \coloneqq -\lambda\) and observe that there are three possible cases for each market \(i\).

We go through them one-by-one.

\medskip

\noindent\textbf{Case 1: $0 < q_i^w < \bar q_i$.}
Both constraints $-q_i\le 0$ and $q_i - \bar q_i \le 0$ are slack at $q^w$,
so by complementary slackness \eqref{eq:KKT-slackness} we must have
\(\mu_i = \nu_i = 0\). We substitute these into \eqref{eq:KKT-stationarity} to deduce \(P_i(q_i^w) = \lambda_{\mathrm{worst}}\).

\medskip

\noindent\textbf{Case 2: $q_i^w = 0$.}
Here the constraint $q_i - \bar q_i \le 0$ is slack (since $q_i^w - \bar q_i < 0$),
so $\nu_i = 0$ by \eqref{eq:KKT-slackness}, while $\mu_i\ge 0$ is free.
Equation \eqref{eq:KKT-stationarity} becomes \(P_i(0) = \lambda_{\mathrm{worst}} + \mu_i \geq \lambda_{\mathrm{worst}}\).

\medskip

\noindent\textbf{Case 3: $q_i^w = \bar q_i$.}
Now the constraint $-q_i \le 0$ is slack (since $-q_i^w = -\bar q_i < 0$),
so $\mu_i = 0$ by \eqref{eq:KKT-slackness}, while $\nu_i\ge 0$ is free.
Equation \eqref{eq:KKT-stationarity} becomes
\(P_i(\bar q_i) = \lambda_{\mathrm{worst}} - \nu_i \leq \lambda_{\mathrm{worst}}\).

As each $P_i(\cdot)$ is weakly decreasing and nonnegative, necessarily $\lambda_{\mathrm{worst}}\geq 0$.

Part \ref{worstitem3} follows immediately: for a single-market allocation with $q_j=\bar{Q}$ and $q_{k\ne j}=0$, we have $\widetilde W=\int_0^{\bar{Q}} P_j(x) dx$, so minimizing across $j$ yields the stated criterion.
\end{proof}

\subsection{Chaos Theorem}
\label{app:chaos-proofs}

In this appendix, we develop the technical machinery for Theorem~\ref{thm:chaos-main}. We first discuss the continuity of the allocation in the parameter, when there are no price controls. Next, we introduce a binding ceiling then prove Proposition \ref{prop:greedy-allocation}, Theorem \ref{thm:chaos-main} and Corollary \ref{cor:chaos-allocation}, in sequence.

Let $\theta\in \Theta\subset\mathbb{R}^m$ be a compact, nonempty parameter space.
There are submarkets $i=1,\dots,n$ with inverse demands $P_i(x;\theta)$.
We assume:
\begin{enumerate}
    \item For each $i$ and every $\theta$, inverse demand $P_i(x;\theta)$ is continuous and strictly decreasing in \(x\).
    \item For each \(i\) and every $x$, inverse demand $P_i(x;\theta)$ is continuous in \(\theta\).
    \item For each \(i\), the submarket unit cost $c_i(\theta)$ is continuous in $\theta$.
    \item The total quantity $\bar{Q}\left(\theta\right)>0$ is continuous in $\theta$.
\end{enumerate}
Note that the separate continuity in \(\theta\) and monotonicity in \(x\) deliver the joint continuity of each \(P_i\).

The vector of allocated quantities $q=(q_1,\dots,q_n)$ must satisfy $\sum_{i=1}^n q_i=\bar{Q}\left(\theta\right)$ and $q_i \geq 0$. We term such a vector \textit{feasible} and define
\[
\widetilde{W}(q;\theta)=\sum_{i=1}^n \int_0^{q_i} P_i(x;\theta) dx.
\]

Consider
\[
\max_{q \geq 0} \sum_{i=1}^n \int_0^{q_i} P_i(x;\theta)dx - \sum_{i=1}^n c_i(\theta) q_i
\quad\text{s.t.}\quad \sum_{i=1}^n q_i=\bar{Q}\left(\theta\right).
\]
As each $P_i(\cdot;\theta)$ is strictly decreasing, the objective is strictly concave in $q$, and so the maximizer $q^*(\theta)$ is unique.
The feasible correspondence \[\Gamma(\theta) \coloneqq \left\{q\in\mathbb{R}^n_+ \colon \ \sum_i q_i=\bar{Q}\left(\theta\right)\right\}\text{,}\] has compact values and a closed graph by the continuity of $\bar Q$. \(\Gamma\) is also continuous: if \(\theta_n\to\theta\) and \(q\in\Gamma(\theta)\), then \(q_n\coloneqq \frac{\bar Q(\theta_n)}{\bar Q(\theta)}q\in\Gamma(\theta_n)\) and \(q_n\to q\).
The objective is continuous in $(q,\theta)$ (by the joint continuity of $P_i$ and continuity of $c_i$).
Therefore, by Berge’s Maximum Theorem, both the unique maximizer $\theta\mapsto q^*(\theta)$ and the value function
\[\theta \mapsto\sum_{i=1}^n \int_0^{q_i^*(\theta)} P_i(x;\theta) dx - \sum_{i=1}^n c_i(\theta) q_i^*(\theta)\text{,}
\]
are continuous in $\theta$. On any region where the optimizer is interior, the multiplier $\theta \mapsto \lambda(\theta)$ is also continuous, as $\lambda(\theta)=P_i(q_i^*(\theta);\theta)-c_i(\theta)$.

Now suppose we have a binding ceiling $\bar p$.
Define $\bar q_i(\theta) \coloneqq D_i(\bar p;\theta)$ and aggregate quantity $\bar{Q}\left(\theta\right) \coloneqq S(\bar p;\theta)$.
The feasible set is
\[
\mathcal{F}(\theta)=\left\{q\in\mathbb{R}^n_+ \colon \ \sum_{i=1}^n q_i=\bar{Q}\left(\theta\right),\ 0 \leq q_i \leq \bar q_i(\theta) \ \forall i\right\}\text{.}
\]
We make the nondegeneracy assumptions
\begin{enumerate}
\item \label{assa} $0<\bar{Q}\left(\theta\right)<\sum_{i=1}^n \bar q_i(\theta)$;
\item \label{assb} $\bar{Q}\left(\theta\right)\neq \sum_{i\in S}\bar q_i(\theta)$ for every $S\subset\{1,\dots,n\}$.
%\item $c_i(\theta)\neq c_j(\theta)$ for $i\neq j$.
\end{enumerate}
\iffalse
We solve
\(\min_{q\in\mathcal{F}(\theta)}\ c(\theta)\cdot q\)
and order costs as $c_{(1)}(\theta)<\cdots<c_{(n)}(\theta)$ and let $k$ denote the unique index with $\sum_{s=1}^{k-1} \bar q_{(s)}(\theta) <\bar{Q}\left(\theta\right)< \sum_{s=1}^k \bar q_{(s)}(\theta)$.
Then the unique optimizer is
\[
q^*_{(r)}(\theta)=
\begin{cases}
\bar q_{(r)}(\theta), & r<k,\\
\bar{Q}\left(\theta\right)-\sum_{s=1}^{k-1} \bar q_{(s)}(\theta), & r=k,\\
0, & r>k\text{,}
\end{cases}
\]
i.e., we “fill in increasing $c$” with at most one partial coordinate. %Changing $\theta$ that breaks a cost tie ($c_i(\theta) = c_j(\theta)$) that changes which markets are filled causes a discrete jump in allocations.
\fi

Suppressing the dependence on the parameter $\theta$, the optimization problem is
\[\tag{$\star$}\label{lp} \min\left\{c \cdot q \colon \ q\in \mathcal{F}\right\}\text{,}
\]
For Proposition~\ref{prop:greedy-allocation}, we also assume $c_i(\theta)\neq c_j(\theta)$ for $i\neq j$. For its proof, recall also the notation: $c_{(1)}(\theta)<\cdots<c_{(n)}(\theta)$ and let $k$ denote the unique index with $\sum_{s=1}^{k-1} \bar q_{(s)}(\theta) <\bar{Q}\left(\theta\right)< \sum_{s=1}^k \bar q_{(s)}(\theta)$
\begin{proof}[Proof of Proposition~\ref{prop:greedy-allocation}]
We continue to suppress \(\theta\) for notational ease, and write \(\bar Q\coloneqq\bar Q(\theta)\), \(\bar q_i\coloneqq\bar q_i(\theta)\), and \(c_i\coloneqq c_i(\theta)\). The price-control optimizing problem is a linear program
\[
\min\left\{\sum_{i=1}^n c_ix_i \colon \  \sum_{i=1}^n x_i= \bar Q,\ 0\le x_i\le \bar q_i\text{ for }i=1,\dots,n\right\}.
\]
Because \(0< \bar Q<\sum_{i=1}^n\bar q_i\), the feasible set is nonempty. It is also compact, so an optimal solution exists. The dual problem is
\[
\max\left\{\lambda \bar Q-\sum_{i=1}^n\mu_i\bar q_i\colon \ \lambda-\mu_i\le c_i\text{ for }i=1,\dots,n,\ \mu_i\ge0\text{ for }i=1,\dots,n\right\}.
\]
Let \(q\) be a primal optimizer and \(\left(\lambda,\mu\right)\) a dual optimizer. By strong duality and complementary slackness, \(\mu_i\left(q_i-\bar q_i\right)=0\) and \(q_i\left(c_i-\lambda+\mu_i\right)=0\)
for each \(i\), while dual feasibility means we must have \(c_i-\lambda+\mu_i\ge0\).

If \(c_i<\lambda\), then dual feasibility implies \(\mu_i\ge\lambda-c_i>0\), which implies \(q_i=\bar q_i\).
If \(c_i>\lambda\), then \(c_i-\lambda+\mu_i>0\), and so \(q_i=0\). Therefore, any coordinate with \(0<q_i<\bar q_i\) must satisfy \(c_i=\lambda\).

As the \(c_i\) are pairwise distinct, there is at most one index \(i\) with \(c_i=\lambda\). If there were no such index, then \(q_i\in\left\{0,\bar q_i\right\}\) for every \(i\), so
\(\bar Q=\sum_{i=1}^n q_i=\sum_{i\in S}\bar q_i\) for \(S\coloneqq\left\{i\colon \ q_i=\bar q_i\right\}\), contradicting the hypothesis that \(\bar Q\) is not a subset sum of the \(\bar q_i\). Consequently, there is a unique index \(k\) with \(c_{\left(k\right)}=\lambda\). Moreover, \(q_{\left(k\right)}\notin\left\{0,\bar q_{\left(k\right)}\right\}\), for otherwise again \(\bar Q\) would be a subset sum of the \(\bar q_i\). Therefore,
\(0<q_{\left(k\right)}<\bar q_{\left(k\right)}\).

For \(r<k\) we have \(c_{\left(r\right)}<\lambda\), so \(q_{\left(r\right)}=\bar q_{\left(r\right)}\). For \(r>k\) we have \(c_{\left(r\right)}>\lambda\), so \(q_{\left(r\right)}=0\). The equality constraint then delivers
\[
q_{\left(k\right)}=\bar Q-\sum_{s=1}^{k-1}\bar q_{\left(s\right)} \qquad \Longrightarrow \qquad
\sum_{s=1}^{k-1}\bar q_{\left(s\right)}<\bar Q<\sum_{s=1}^{k}\bar q_{\left(s\right)},
\]
as \(0<q_{\left(k\right)}<\bar q_{\left(k\right)}\). Thus, \(k\) is exactly the unique index in the statement, and \(q\) has the asserted form.
\end{proof}

We now finish developing the apparatus needed to prove the theorem. By our assumption that there exists a local cost-only perturbation at \(\theta^*\) for \(v\) and \(w\), we fix the feasible polytope $\mathcal F= \mathcal F(\theta^*)$ and without loss of generality consider perturbations that act only through the cost vector $c(\theta)$. Thus, we treat $c$ as the effective parameter in \eqref{lp}. By assumption, after shrinking the cost neighborhood if necessary, every local cost perturbation considered below is induced by a smooth parameter perturbation in an arbitrary neighborhood of \(\theta^*\), and so all local conclusions in cost space translate directly to local conclusions in parameter space. 

%We maintain Assumptions \ref{assa} and \ref{assb} 
We note that Assumption \ref{assb} means that every vertex $v$ of $\mathcal{F}$ has exactly one coordinate strictly between its bounds (the partial coordinate), and all other coordinates are at $0$ or at $\bar q_i$. To elaborate, a point \(v \in \mathcal F\) is a vertex if all but one coordinates are at their bounds \(0\) or \(\bar{q}_i\), with the remaining coordinate then pinned down by \(\sum_i q_i = \bar{Q}\). Recall also that two distinct vertices $v$ and $w$ are \textit{adjacent} if they share an edge: along the edge exactly the pair
$(r,s)$ of coordinates moves (with all others fixed at bounds) and
\[\label{edge-diff}\tag{$A5$}
w - v = \Delta (e_s - e_r)\qquad \text{for some } r\neq s \text{ and } \Delta > 0,\]
where we label $r$ as the coordinate that decreases ($w_r < v_r$) and $s$ as the coordinate that increases ($w_s > v_s$). Along the edge from $v$ to $w$, only coordinates $r$ and $s$ move, with all other coordinates fixed at bounds. Since each vertex has exactly one partial coordinate, at each endpoint exactly one of $r$ or $s$ is partial and the other is at a bound. Note moreover, that the identity of the partial coordinate need not change between endpoints. %At $v$, $r$ is the unique partial index (so $0<v_r<\bar q_r$); at $w$, one of $r$ or $s$ is the new partial and the other hits a bound.  
Equivalently, edges are exactly the one-dimensional faces on which two coordinates are free (subject to the sum constraint), and the others are fixed at bounds.

For any $v\in \mathcal{F}$, the (outer) normal cone is
\[
\mathcal{N}_v = \left\{c\in\R^n \colon \ c\cdot(x-v)\ge 0 \ \forall x\in \mathcal F\right\}\text{.}
\]
The following is standard:
\begin{remark}\label{opt-nc}
    $v$ solves \eqref{lp} if and only if $c \in \mathcal{N}_v$.
\end{remark}
Moreover,
\begin{remark}\label{tie-orth}
    If a given $c$ makes exactly two distinct vertices $v,w$ optimal, then $v,w$ are adjacent and $c \cdot \left(w-v\right) = 0$.
\end{remark}
Denote the set of such common-optimal costs by
\[
\mathcal{H}_{vw} \coloneqq \mathcal{N}_v\cap \mathcal{N}_w\cap \left\{c \colon \ c\cdot(w-v)=0\right\}.
\]
We say that $c^\star$ is a \emph{simple tie} for $v,w$ if
\[
c^\star\in \operatorname{relint}(\mathcal{H}_{vw}) \quad \text{and} \quad c^\star\notin \mathcal{N}_u \text{ for any other vertex } u.
\]
We denote the open ball of radius \(\rho > 0\) around \(c^\star\) by
\[
B\left(c^\star,\rho\right)\coloneqq\left\{c\in\mathbb R^n\colon \ \left\|c-c^\star\right\|<\rho\right\}.
\]
The proofs of the next two results (Lemma \ref{lem:trans-app} and Corollary \ref{cor:trans-app}) lie in the supplementary appendix.

%\begin{remark}\label{rem:simple-tie-app} If \(c\) makes \(v\) and \(w\) optimal but the tie is not simple, then every neighborhood of \(c\) contains some \(c'\) that is a simple tie for \(v,w\). Hence, Lemma~\ref{lem:trans-app} applies arbitrarily close to any \((v,w)\) tie. \end{remark}

\begin{lemma}\label{lem:trans-app}
Let $v\neq w$ be adjacent vertices and $c^\star$ a simple tie for $v,w$.
Then for any neighborhood $U\subset\R^n$ of $c^\star$ there exists $\varepsilon > 0$ and a smooth function
$\gamma \colon \left[-\varepsilon,\varepsilon\right] \to U$ such that \(\gamma(0)=c^\star\), and the unique optimizer of \eqref{lp} equals $v$ for $t<0$ and $w$ for $t>0$,
while at $t=0$ both $v$ and $w$ are optimal.
\end{lemma}

\begin{corollary}\label{cor:trans-app}
Let \(v\neq w\) be adjacent vertices, and let \(c^\star\) be any cost vector for which both \(v\) and \(w\) are optimal for \eqref{lp}. Then for every open neighborhood \(U\subset\mathbb{R}^n\) of \(c^\star\) there exist \(\hat c\in U\), \(\varepsilon>0\), and a smooth function
\(\gamma\colon\left[-\varepsilon,\varepsilon\right]\to U\) such that \(\gamma(0)=\hat c\), both \(v\) and \(w\) are optimal at \(t=0\), the unique optimizer of \eqref{lp} equals \(v\) for \(t<0\), and the unique optimizer equals \(w\) for \(t>0\).
\end{corollary}

Let \(\gamma(t)\) be a smooth function as in Corollary~\ref{cor:trans-app}. The next lemma shows that the welfare jump at \(t\) passes through \(0\) equals \(\Delta W\). The \(-c\cdot q\) contribution cancels in the limit, as \(\gamma(\cdot)\) is continuous at \(0\) and \(\gamma(0)\cdot(w-v)=0\).
\iffalse
\begin{lemma}\label{lem:jump-app}
We have
\[\tag{$A6$}\label{eq:jump}
\lim_{t\downarrow0}\left(\sum_i\int_0^{w_i}P_i(x) dx-\gamma(t)\cdot w\right)-\lim_{t\uparrow0}\left(\sum_i\int_0^{v_i}P_i(x) dx-\gamma(t)\cdot v\right)=\Delta W\text{.}
\]
\end{lemma}\fi

\begin{proof}[Proof of Theorem~\ref{thm:chaos-main}]
Fix a neighborhood \(U\) of \(\theta^*\). By our assumption that there is a local cost-only perturbation, after shrinking the cost neighborhood if necessary, every smooth cost variation with image in that neighborhood is induced by a smooth parameter perturbation taking values in \(U\). Since \(v\) and \(w\) are both optimal at \(\theta^*\), Corollary~\ref{cor:trans-app} yields a smooth cost variation near \(c(\theta^*)\) along which the unique optimizer is \(v\) for \(t<0\) and \(w\) for \(t>0\). The welfare jump along this cost variation equals
\[
\sum_{i=1}^n \int_0^{w_i} P_i(x) dx
-
\sum_{i=1}^n \int_0^{v_i} P_i(x) dx
=
\int_{v_s}^{w_s} P_s(x) dx
-
\int_{w_r}^{v_r} P_r(x) dx
=
\Delta W,
\]
which is nonzero by hypothesis. Moreover, our cost-only perturbation assumption guarantees that \(\mathcal{F}(\theta^\ast)\), \(P_r(\cdot;\theta^\ast)\), and \(P_s(\cdot;\theta^\ast)\) are unaffected by the induced parameter perturbation. Thus, as \(w_i=v_i\) for \(i\notin\{r,s\}\), the continuity of \(\theta\mapsto \int_0^{v_i} P_i(x;\theta) dx\) implies that the remaining integrated welfare terms have identical one-sided limits, so the welfare jump for the induced parameter perturbation is also \(\Delta W\).

Reversing the variation yields a smooth parameter perturbation along which the unique optimizer is \(w\) for \(t<0\) and \(v\) for \(t>0\), and the corresponding welfare jump is \(-\Delta W\). Hence, every neighborhood of \(\theta^*\) contains smooth parameter perturbations along which welfare jumps by \(+\left|\Delta W\right|\) and \(-\left|\Delta W\right|\), so welfare is not locally monotone at \(\theta^*\).
\end{proof}

\begin{proof}[Proof of Corollary~\ref{cor:chaos-allocation}]
Fix a neighborhood \(U\) of \(\theta^\ast\). By part I of Theorem \ref{thm:chaos-main}, there exists a smooth parameter perturbation taking values in \(U\) along which the optimal allocation jumps between \(v\) and \(w\). Choose \(\theta^-,\theta^+\in U\) on opposite sides of the jump such that \(v\) is optimal at \(\theta^-\) and \(w\) is optimal at \(\theta^+\). Since \(v\) and \(w\) are adjacent, we have
\(w-v=\Delta\left(e_s-e_r\right)\) and so the allocation changes by reallocating \(\Delta\) units from market \(r\) to market \(s\), with all other market quantities unchanged.
\end{proof}

\section{Proposition \ref{prop:reductioninformal} and Theorem \ref{thm:extremalsinformal} Formal Statements and Proofs}\label{app:extremals-proofs}
Recall our assumption:
{
\renewcommand{\theassumption}{2}
\begin{assumption}\label{mainass}
For all $i$, $P_i$ is nonincreasing and satisfies
\begin{enumerate}
    \item \label{ass:a1_appendix} $P_i(q_i^{\mathrm{obs}})=p_{0,i}=\bar p-b_i \quad\text{with}\quad b_i\in[\underline b_i,\bar b_i].$
  \item \label{ass:a2_appendix} There exist $g_{i,L} < g_{i,U}<0$ with $g_{i,L}\le P'_i(q)\le g_{i,U}$ for a.e. $q\in\left[0, q_i^{\max}\right]$.
  \item \label{ass:a3_appendix} $P_i(0)\le M_i<\infty$
  \item \label{ass:a4_appendix} $P_i\in W^{1,\infty}\left(\left[0, q_i^{\max}\right]\right)$, the space of bounded Lipschitz-continuous functions on $\left[0, q_i^{\max}\right]$.
\end{enumerate}
\end{assumption}
}
Recall also that \(q_i(\cdot)\) denotes the left-continuous generalized inverse of \(P_i(\cdot)\):
\[
q_i(p)\coloneqq
\begin{cases}
\inf\left\{x\in\left[0,q_i^{\max}\right]\colon P_i(x)\le p\right\},&\text{if the set is nonempty},\\
q_i^{\max},&\text{otherwise}.
\end{cases}
\]
We first compute easy bounds for \(q_i\). By the fundamental theorem of calculus,
\(P_i(q)=p_{0,i}+\int_{q_i^{\mathrm{obs}}}^{q}P_i'(s)d s\). By Assumption~\ref{mainass}\ref{ass:a2_appendix}, \(P_i(q)\) lies between the two affine functions \(p_{0,i}+g_{i,L}(q-q_i^{\mathrm{obs}})\) and \(p_{0,i}+g_{i,U}(q-q_i^{\mathrm{obs}})\). Rearranging these inequalities yields
\[\min\left\{q_i^{\mathrm{obs}}+\frac{p-p_{0,i}}{g_{i,U}},q_i^{\mathrm{obs}}+\frac{p-p_{0,i}}{g_{i,L}}\right\} \leq q_i(p) \leq \max\left\{q_i^{\mathrm{obs}}+\frac{p-p_{0,i}}{g_{i,U}},q_i^{\mathrm{obs}}+\frac{p-p_{0,i}}{g_{i,L}}\right\}.\]
To incorporate the choke-price restriction \(P_i(0)\le M_i\) (Assumption~\ref{mainass}\ref{ass:a3_appendix}), we explicitly enforce that demand is zero at prices \(p\ge M_i\) and rule out \((q,p)\) pairs that cannot be connected to \(q=0\) without exceeding \(M_i\) under the least-negative slope \(g_{i,U}\).

We summarize the bounds with two functions. For \(p \in \mathbb{R}\), we define
\[
\ell_i(p)\coloneqq
\begin{cases}
0,&p\ge M_i,\\[4pt]
\max\left\{0,\min\left\{q_i^{\mathrm{obs}}+\frac{p-p_{0,i}}{g_{i,U}},q_i^{\mathrm{obs}}+\frac{p-p_{0,i}}{g_{i,L}}\right\}\right\},&p<M_i,
\end{cases}
\]
and
\[
u_i(p)\coloneqq
\begin{cases}
0,&p\ge M_i,\\[6pt]
\min\left\{q_i^{\max},\max\left\{q_i^{\mathrm{obs}}+\frac{p-p_{0,i}}{g_{i,U}},q_i^{\mathrm{obs}}+\frac{p-p_{0,i}}{g_{i,L}}\right\},\frac{M_i-p}{-g_{i,U}}\right\},&p<M_i.
\end{cases}
\]

We also characterize set of candidate common shadow prices consistent with the aggregate quantity constraint as follows. Define \(L(p)\coloneqq\sum_i\ell_i(p)\), \(U(p)\coloneqq\sum_i u_i(p)\), and
\[
\mathcal{I}\coloneqq\left\{p\in\mathbb{R}\colon L(p)\le\bar{Q}\le U(p)\right\}.
\]
Then \(q_i(p)\in\left[\ell_i(p),u_i(p)\right]\) for all \(p\in\mathcal{I}\).

To summarize things: at price \(p\), \(\ell_i(p)\) and \(u_i(p)\) are the smallest and largest (respectively) quantities in market \(i\) that are consistent with \(q_i(\cdot)\) (i) passing through \(\left(q_i^{\mathrm{obs}},p_{0,i}\right)\), (ii) obeying the slope bounds, and (iii) satisfying the choke-price restriction \(P_i(0)\le M_i\).
The set \(\mathcal{I}\) is the set of candidate common shadow prices \(p\) for which the interval constraints across markets are compatible with \(\sum_i q_i(p)=\bar{Q}\).

Next, for \(P\in\mathcal{P}\), define
\[
\Phi(P)\coloneqq\sum_{i=1}^n\int_{q_i^{\mathrm{obs}}}^{q_i^{*}(P)}\left(P_i(x)-p^{*}(P)\right)d x,
\]
where \(q^*(P)\) is the equal-shadow-price split of total quantity \(\bar{Q}\) and \(p^*(P)\) is its Lagrange multiplier.
Note that \(\sum_i q_i^*(P)=\bar{Q}\) and \(\sum_i q_i^{\mathrm{obs}}=\bar{Q}\), so using \(p^*(P)\) (rather than \(\bar{p}\)) is harmless since \(\sum_i\left(q_i^*(P)-q_i^{\mathrm{obs}}\right)=0\).
Recall \(\overline{\Phi}\coloneqq\max_{P\in\mathcal{P}}\Phi(P)\), and \(\underline{\Phi}\coloneqq\min_{P\in\mathcal{P}}\Phi(P)\).

We make the following interiority assumption

\begin{assumption}\label{ass:pathA}
The set \(\mathcal{I}\) is nonempty and compact.
Moreover, for every \(i\) and every \(p\in\mathcal{I}\), and every \(s\in\left[\min\left\{p,p_{0,i}\right\},\max\left\{p,p_{0,i}\right\}\right]\),
\[
0<q_i^{\mathrm{obs}}+\frac{s-p_{0,i}}{g_{i,U}}<q_i^{\max},\qquad
0<q_i^{\mathrm{obs}}+\frac{s-p_{0,i}}{g_{i,L}}<q_i^{\max},\qquad
s<M_i,
\]
and
\[
\max\left\{q_i^{\mathrm{obs}}+\frac{s-p_{0,i}}{g_{i,U}},q_i^{\mathrm{obs}}+\frac{s-p_{0,i}}{g_{i,L}}\right\}<\frac{M_i-s}{-g_{i,U}}.
\]
\end{assumption}
For our purposes, this assumption means that \(\ell_i(\cdot)\) and \(u_i(\cdot)\) simplify to
\[\tag{\(B1\)}\label{eq:simpell}\ell_i(p) = \min\left\{q_i^{\mathrm{obs}}+\frac{p-p_{0,i}}{g_{i,U}},q_i^{\mathrm{obs}}+\frac{p-p_{0,i}}{g_{i,L}}\right\},\]
and
\[\tag{\(B2\)}\label{eq:simpu}u_i(p) = \max\left\{q_i^{\mathrm{obs}}+\frac{p-p_{0,i}}{g_{i,U}},q_i^{\mathrm{obs}}+\frac{p-p_{0,i}}{g_{i,L}}\right\},\]
and take values in \(\left(0,q_i^{\max}\right)\). Henceforth, we take these simplified expressions to be the definitions of \(\ell_i\) and \(u_i\). Moreover, on \(p \in \mathcal{I}\), \(q_i(p)\) simplifies to
\(q_i(p) = P^{-1}_i(p)\).

Next, we introduce several more objects: for each \(i\), define
\[
\alpha_i\coloneqq\frac{1}{g_{i,U}},\qquad
\beta_i\coloneqq\frac{1}{g_{i,L}},\quad \text{and} \quad
d_i\coloneqq\beta_i-\alpha_i=\frac{1}{g_{i,L}}-\frac{1}{g_{i,U}}>0.
\]
For each \(p\in\mathcal{I}\), define the baseline endpoint choices
\[
q_i^{+}(p)\coloneqq
\begin{cases}
\ell_i(p),&p\ge p_{0,i},\\
u_i(p),&p<p_{0,i},
\end{cases}
\quad \text{and} \quad
q_i^{-}(p)\coloneqq
\begin{cases}
u_i(p),&p\ge p_{0,i},\\
\ell_i(p),&p<p_{0,i}.
\end{cases}
\]
Let
\[
\Delta^+(p)\coloneqq\bar{Q}-\sum_{i=1}^n q_i^{+}(p),\quad \text{and} \quad
\Delta^-(p)\coloneqq\bar{Q}-\sum_{i=1}^n q_i^{-}(p).
\]
and define active sets
\[
\mathcal{A}^+(p)\coloneqq
\begin{cases}
\left\{i\colon p\ge p_{0,i}\right\},&\Delta^+(p)>0,\\
\left\{i\colon p<p_{0,i}\right\},&\Delta^+(p)<0,\\
\emptyset,&\Delta^+(p)=0,
\end{cases}
\quad \text{and} \quad
\mathcal{A}^-(p)\coloneqq
\begin{cases}
\left\{i\colon p<p_{0,i}\right\},&\Delta^-(p)>0,\\
\left\{i\colon p\ge p_{0,i}\right\},&\Delta^-(p)<0,\\
\emptyset,&\Delta^-(p)=0.
\end{cases}
\]
Define also the triangle-penalty value functions
\[
\Psi^+(p)\coloneqq
\min_{\left\{\delta_i\right\}_{i\in\mathcal{A}^+(p)}} \left\{ \sum_{i\in\mathcal{A}^+(p)}\frac{\delta_i^2}{2d_i}\right\}
\quad\text{s.t.}\quad
\sum_{i\in\mathcal{A}^+(p)}\delta_i=\left|\Delta^+(p)\right|,\qquad
0\le\delta_i\le u_i(p)-\ell_i(p),
\]
and
\[
\Psi^-(p)\coloneqq
\min_{\left\{\delta_i\right\}_{i\in\mathcal{A}^-(p)}} \left\{ \sum_{i\in\mathcal{A}^-(p)}\frac{\delta_i^2}{2d_i}\right\}
\quad\text{s.t.}\quad
\sum_{i\in\mathcal{A}^-(p)}\delta_i=\left|\Delta^-(p)\right|,\qquad
0\le\delta_i\le u_i(p)-\ell_i(p).
\]
For a fixed \(p \in \mathcal{I}\) we term the solution to \(\max_{\substack{P \in \mathcal{P}\\ p^*(P) = p}} \Phi(P)\) (in turn, \(\min_{\substack{P \in \mathcal{P}\\ p^*(P) = p}} \Phi(P)\)) an \textit{inner optimizer}.
% Formal version of Proposition 3
{
\renewcommand{\theproposition}{3*}
\begin{proposition}\label{prop:reduction}
Posit Assumptions~\ref{mainass} and~\ref{ass:pathA}.
Then
\[
\overline\Phi=\max_{p\in\mathcal{I}}\left\{\bar{Q}p-\sum_{i=1}^n q_i^{\mathrm{obs}}p_{0,i}
-\sum_{i\colon p\ge p_{0,i}}\int_{p_{0,i}}^{p}\ell_i(s)d s
+\sum_{i\colon p<p_{0,i}}\int_{p}^{p_{0,i}}u_i(s)d s
-\Psi^+(p)\right\},
\tag{$B3$}\label{eq:Phi_sup_corr}
\]
and
\[
\underline\Phi=\min_{p\in\mathcal{I}}\left\{\bar{Q}p-\sum_{i=1}^n q_i^{\mathrm{obs}}p_{0,i}
-\sum_{i\colon p\ge p_{0,i}}\int_{p_{0,i}}^{p}u_i(s)d s
+\sum_{i\colon p<p_{0,i}}\int_{p}^{p_{0,i}}\ell_i(s)d s
+\Psi^-(p)\right\}.
\tag{$B4$}\label{eq:Phi_inf_corr}
\]
Moreover, for any \(p\in\mathcal{I}\), the inner optimizer may be chosen so that each inverse-quantity function \(q_i(\cdot)\) is continuous and piecewise affine on \(\left[\min\left\{p,p_{0,i}\right\},\max\left\{p,p_{0,i}\right\}\right]\), has a.e.\ slope in \(\left\{\alpha_i,\beta_i\right\}\), and is piecewise affine with at most one kink.
\end{proposition}
}
We first assemble several lemmas, deferring the proofs of these lemmas to the supplementary appendix.
\begin{lemma}\label{lem:a}
    Under Assumptions \ref{mainass} and \ref{ass:pathA}, the functions \(\Psi^+(\cdot)\) and \(\Psi^-(\cdot)\) are continuous on \(\mathcal I\).
\end{lemma}
Next, for any \(p\in\mathcal{I}\), define
\(\mathcal{I}_+(p)\coloneqq\left\{i\colon p\ge p_{0,i}\right\}\) and \(\mathcal{I}_-(p)\coloneqq\left\{i\colon p<p_{0,i}\right\}\).
\begin{lemma}\label{lem:2}
    Posit Assumptions \ref{mainass} and \ref{ass:pathA} and fix \(p \in \mathcal{I}\). Then, the inner problem of maximizing \(\Phi(P)\) over \(P \in \mathcal{P}\) satisfying \(p^*(P) = p\) is equivalent to minimizing \(\int_{p_{0,i}}^{p}q_i\) for \(i\in\mathcal{I}_+(p)\) and maximizing \(\int_{p}^{p_{0,i}}q_i\) for \(i\in\mathcal{I}_-(p)\), subject to feasibility and \(\sum_i q_i(p)=\bar{Q}\).
\end{lemma}

\begin{lemma}\label{lem:3}
    \(q_i'(s)\in\left[\alpha_i,\beta_i\right]\) for a.e. \(s\in\left[\min\left\{p,p_{0,i}\right\},\max\left\{p,p_{0,i}\right\}\right]\).
\end{lemma}

\begin{lemma}\label{lem:4}
    Now let \(t_0<t_1\), \(\alpha<\beta\), and let function \(q\) be absolutely continuous on \(\left[t_0,t_1\right]\) with \(q(t_0)=q_0\) and \(q'(t)\in\left[\alpha,\beta\right]\) a.e. Define \(\underline{q}(t)\coloneqq q_0+\alpha\left(t-t_0\right)\).
    If \(q(t_1)=\underline{q}(t_1)+\delta\) for some \(\delta\in\left[0,\left(\beta-\alpha\right)\left(t_1-t_0\right)\right]\), then
\[
\int_{t_0}^{t_1}q(t)d t\ge\int_{t_0}^{t_1}\underline{q}(t)d t+\frac{\delta^2}{2\left(\beta-\alpha\right)}.
\tag{B5}\label{eq:triangle-app}
\]
Moreover, equality holds if and only if \(q\) is ''bang-bang;'' \(q'(t)=\alpha\) a.e. on \(\left[t_0,t_1-h\right]\) and \(q'(t)=\beta\) a.e. on \(\left[t_1-h,t_1\right]\), where \(h=\delta/\left(\beta-\alpha\right)\).
\end{lemma}

\begin{proof}[Proof of Proposition \ref{prop:reduction}]
Fix \(p\in\mathcal{I}\) and consider the inner maximization problem of maximizing \(\Phi(P)\), which by Lemma \ref{lem:2} is equivalent to minimizing \(\int_{p_{0,i}}^{p}q_i\) for \(i\in\mathcal{I}_+(p)\) and maximizing \(\int_{p}^{p_{0,i}}q_i\) for \(i\in\mathcal{I}_-(p)\), subject to \(\sum_i q_i(p)=\bar{Q}\). Ignoring the constraint \(\sum_i q_i(p)=\bar{Q}\) and optimizing the integral term in \eqref{eq:id_sum} with \(p^*=p\), \(\sum_{i=1}^n\int_{p_{0,i}}^{p^*}q_i(s)d s\), the solution is \(q_i=\ell_i\) for \(i\in\mathcal{I}_+(p)\) and \(q_i=u_i\) for \(i\in\mathcal{I}_-(p)\), corresponding to \(q_i^{+}(p)\).

If this sum does not equal \(\bar{Q}\), the constraint \(\sum_i q_i(p)=\bar{Q}\) forces deviations from the endpoints (\(\ell_i\) and \(u_i\)) at \(p\) on exactly the set \(\mathcal{A}^+(p)\), as follows (with Claim \ref{clm:active_set_deviations} formally proved in the supplementary appendix). Note that \(\sum_i q^+_i(p) \neq \bar{Q}\) if and only if \(\Delta^+(p) \neq 0\).

\begin{claim}\label{clm:active_set_deviations}
Fix \(p\in \mathcal I\).
\begin{enumerate}
    \item If \(\Delta^+(p)>0\), there exist numbers \(\left\{\delta_i\right\}_{i\in \mathcal{I}_+(p)}\) such that \(0\le\delta_i\le u_i(p)-\ell_i(p)\) for all \(i\in \mathcal{I}_+(p)\) and \(\sum_{i\in \mathcal{I}_+(p)}\delta_i=\Delta^+(p)\), and the endpoint vector \(q(p)\) defined by
\[
q_i(p)=\ell_i(p)+\delta_i\ \text{for }i\in \mathcal{I}_+(p),\qquad q_i(p)=u_i(p)\ \text{for }i\in \mathcal{I}_-(p)
\]
satisfies \(\sum_{i=1}^n q_i(p)=\bar{Q}\). Moreover, in the inner problem, there is an optimizer with \(q_i(p)=u_i(p)\) for all \(i\in \mathcal{I}_-(p)\); equivalently, deviations from the baseline endpoints \(q_i^+(p)\) may be taken without loss of optimality to occur only on \(\mathcal{A}^+(p)=\mathcal{I}_+(p)\).

\item If \(\Delta^+(p)<0\), there exist numbers \(\left\{\delta_i\right\}_{i\in \mathcal{I}_-(p)}\) such that \(0\le\delta_i\le u_i(p)-\ell_i(p)\) for all \(i\in \mathcal{I}_-(p)\) and \(\sum_{i\in \mathcal{I}_-(p)}\delta_i=-\Delta^+(p)\), and the endpoint vector \(q(p)\) defined by
\[
q_i(p)=\ell_i(p)\ \text{for }i\in \mathcal{I}_+(p),\qquad q_i(p)=u_i(p)-\delta_i\ \text{for }i\in \mathcal{I}_-(p)
\]
satisfies \(\sum_{i=1}^n q_i(p)=\bar{Q}\). Moreover, in the inner problem, there is an optimizer with \(q_i(p)=\ell_i(p)\) for all \(i\in \mathcal{I}_+(p)\); equivalently, deviations from the baseline endpoints \(q_i^+(p)\) may be taken without loss of optimality to occur only on \(\mathcal{A}^+(p)=\mathcal{I}_-(p)\).
\end{enumerate}
\end{claim}

Now fix \(p\in\mathcal{I}\). By Claim \ref{clm:active_set_deviations}, we may restrict attention to endpoint deviations supported on \(\mathcal{A}^+(p)\).
Let \(\left\{\delta_i\right\}_{i\in\mathcal{A}^+(p)}\) be feasible for \(\Psi^+(p)\), and set \(\delta_i\coloneqq 0\) for \(i\notin\mathcal{A}^+(p)\).
Then \(q_i(p)=\ell_i(p)+\delta_i\) for all \(i\) with \(p\ge p_{0,i}\), and \(q_i(p)=u_i(p)-\delta_i\) for all \(i\) with \(p<p_{0,i}\).

By Lemma \ref{lem:3}, \(q_i'(s) \in \left[\alpha_i,\beta_i\right]\) a.e. on \(s\in\left[\min\left\{p,p_{0,i}\right\},\max\left\{p,p_{0,i}\right\}\right]\), so if \(p\ge p_{0,i}\) and \(q_i(p)=\ell_i(p)+\delta_i\), we can apply Lemma \ref{lem:4} with \(\alpha=\alpha_i\), \(\beta=\beta_i\), \(t_0=p_{0,i}\), \(t_1=p\) to obtain
\[
\int_{p_{0,i}}^{p}q_i(s) ds\ge\int_{p_{0,i}}^{p}\ell_i(s) ds+\frac{\delta_i^2}{2d_i}.
\tag{\(B6\)}\label{eq:tri_plus}
\]
Analogously, if \(p<p_{0,i}\) and \(q_i(p)=u_i(p)-\delta_i\), then by a change of variables reversing the integration limits, applying Lemma \ref{lem:4} yields
\[
\int_{p}^{p_{0,i}}q_i(s) ds\le\int_{p}^{p_{0,i}}u_i(s) ds-\frac{\delta_i^2}{2d_i}.
\tag{\(B7\)}\label{eq:tri_minus}
\]

Using Lemma \ref{lem:2} and substituting \eqref{eq:tri_plus} and \eqref{eq:tri_minus} into \eqref{eq:id_sum} shows that for fixed \(p\in\mathcal{I}\),
\[\begin{split}
    \sup_{\substack{P\in\mathcal{P}\\ p^*(P)=p}}\Phi(P)
\leq \bar{Q}p-\sum_{i=1}^n q_i^{\mathrm{obs}}p_{0,i}
&-\sum_{i\colon p\ge p_{0,i}}\int_{p_{0,i}}^{p}\ell_i(s) ds
+\sum_{i\colon p<p_{0,i}}\int_{p}^{p_{0,i}}u_i(s) ds\\
&-\Psi^+(p).
\end{split}\tag{\(B8\)}\label{eq:ub}
\]

Continue to fix \(p\in\mathcal{I}\) and let \(\left\{\delta_i^*\right\}_{i\in\mathcal{A}^+(p)}\) attain \(\Psi^+(p)\)
(existence follows from compactness of the feasible set and continuity of the objective).
Extend by setting \(\delta_i^*\coloneqq 0\) for \(i\notin\mathcal{A}^+(p)\), and define \(h_i\coloneqq \delta_i^*/d_i\) for each \(i\).
Construct \(q_i^*(\cdot)\) on \(\left[\min\left\{p,p_{0,i}\right\},\max\left\{p,p_{0,i}\right\}\right]\) by the equality cases in \eqref{eq:tri_plus} and \eqref{eq:tri_minus}:
\begin{itemize}
\item If \(p\ge p_{0,i}\), set \(q_i^*(s)=\ell_i(s)\) for \(s\in\left[p_{0,i},p-h_i\right]\), and set
\[
q_i^*(s)\coloneqq \ell_i\left(p-h_i\right)+\beta_i\left(s-\left(p-h_i\right)\right)
\quad \text{for } s\in\left[p-h_i,p\right].
\]
\item If \(p<p_{0,i}\), set \(q_i^*(s)=u_i(s)\) for \(s\in\left[p+h_i,p_{0,i}\right]\), and set
\[
q_i^*(s)\coloneqq u_i\left(p+h_i\right)+\beta_i\left(s-\left(p+h_i\right)\right)
\quad \text{for } s\in\left[p,p+h_i\right].
\]
\end{itemize}
By construction, each \(q_i^*(\cdot)\) is continuous and piecewise affine, has a.e. slope in \(\left\{\alpha_i,\beta_i\right\}\), and satisfies \(q_i^*(p_{0,i})=q_i^{\mathrm{obs}}\).
Moreover, \(\sum_{i=1}^n q_i^*(p)=\sum_{i=1}^n q_i^+(p)+\Delta^+(p)=\bar{Q}\), and \(q^*\) attains equality in \eqref{eq:tri_plus} and \eqref{eq:tri_minus}.
Therefore, the upper bound \eqref{eq:ub} is tight and \(q^*(\cdot)\) attains the inner optimum at this \(p\).

Define \(P_i^*\) as the inverse of \(q_i^*\).
Then \(P_i^*\) is continuous and piecewise affine on \(\left[0,q_i^{\max}\right]\), with a.e. derivative in \(\left\{g_{i,L},g_{i,U}\right\}\), satisfies \(P_i^*\left(q_i^{\mathrm{obs}}\right)=p_{0,i}\), and satisfies \(P_i^*(0)\le M_i\) because \(q_i^*(p)=0\) for all \(p\ge M_i\) is feasible under the choke-aware bounds.
Thus, \(P^*\in\mathcal{P}\) attains the inner optimum at \(p\).

Finally, we maximize over \(p \in \mathcal{I}\). Define \(F^+(p)\) as the bracketed expression in \eqref{eq:Phi_sup_corr}. By Lemma \ref{lem:a}, \(\Psi^+(p)\) is continuous in \(p\), so \(F^+(p)\) is continuous on the compact set \(\mathcal{I}\) and attains a maximum, proving \eqref{eq:Phi_sup_corr}.

The argument for the lower bound \(\underline{\Phi}\) is identical, \textit{mutatis mutandis}.
\end{proof}

{
\renewcommand{\thetheorem}{3*}
\begin{theorem}\label{thm:extremals}
Posit Assumptions~\ref{mainass} and~\ref{ass:pathA}.
Both \(\overline{\Phi}\) and \(\underline{\Phi}\) are attained by some \(P^*\in\mathcal{P}\).
Moreover, each \(P_i^*\) may be chosen continuous and piecewise affine on \(\left[0,q_i^{\max}\right]\), with a.e. derivative \(P_i^{* \prime}(q)\in\left\{g_{i,L},g_{i,U}\right\}\), and satisfying \(P_i^*(0)\le M_i\).
\end{theorem}
}

\begin{proof}
By Proposition~\ref{prop:reduction}, the objective in \eqref{eq:Phi_sup_corr} is continuous on the compact set \(\mathcal{I}\), so it admits a maximizer \(p^*\in\mathcal{I}\).
Let \(\left\{\delta_i^*\right\}\) attain \(\Psi^+(p^*)\).
The proof of Proposition~\ref{prop:reduction} constructs inverse-quantity functions \(q_i^*(\cdot)\) that attain the inner optimum at \(p^*\) and are continuous and piecewise affine with slopes in \(\left\{\alpha_i,\beta_i\right\}\) on \(\left[\min\left\{p^*,p_{0,i}\right\},\max\left\{p^*,p_{0,i}\right\}\right]\), with at most one slope switch.
Outside this interval, extend each \(q_i^*(\cdot)\) to all prices by any monotone continuous continuation that keeps slopes in \(\left\{\alpha_i,\beta_i\right\}\) while \(q_i^*(p)>0\), and then sets \(q_i^*(p)=0\) for all \(p\ge M_i\). This extension does not change \(\Phi\), since \(\Phi\) depends on \(P_i\) only on \(\left[\min\left\{q_i^{\mathrm{obs}},q_i^*(p^*)\right\},\max\left\{q_i^{\mathrm{obs}},q_i^*(p^*)\right\}\right]\).
Define \(P_i^*\) as the generalized inverse of this extended \(q_i^*\).
Then \(P_i^*\) is continuous and piecewise affine on \(\left[0,q_i^{\max}\right]\), with a.e.\ derivative in \(\left\{g_{i,L},g_{i,U}\right\}\), satisfies \(P_i^*\left(q_i^{\mathrm{obs}}\right)=p_{0,i}\), and satisfies \(P_i^*(0)\le M_i\).
Therefore, \(P^*\in\mathcal{P}\) attains \(\overline{\Phi}\).
The construction for \(\underline{\Phi}\) is analogous using \eqref{eq:Phi_inf_corr} and \(\Psi^-(p)\).
\end{proof}

\section{Additional Robustness Figures}\label{app:additional-figures}

\begin{figure}[htbp]
    \centering
    \includegraphics[width=\textwidth]{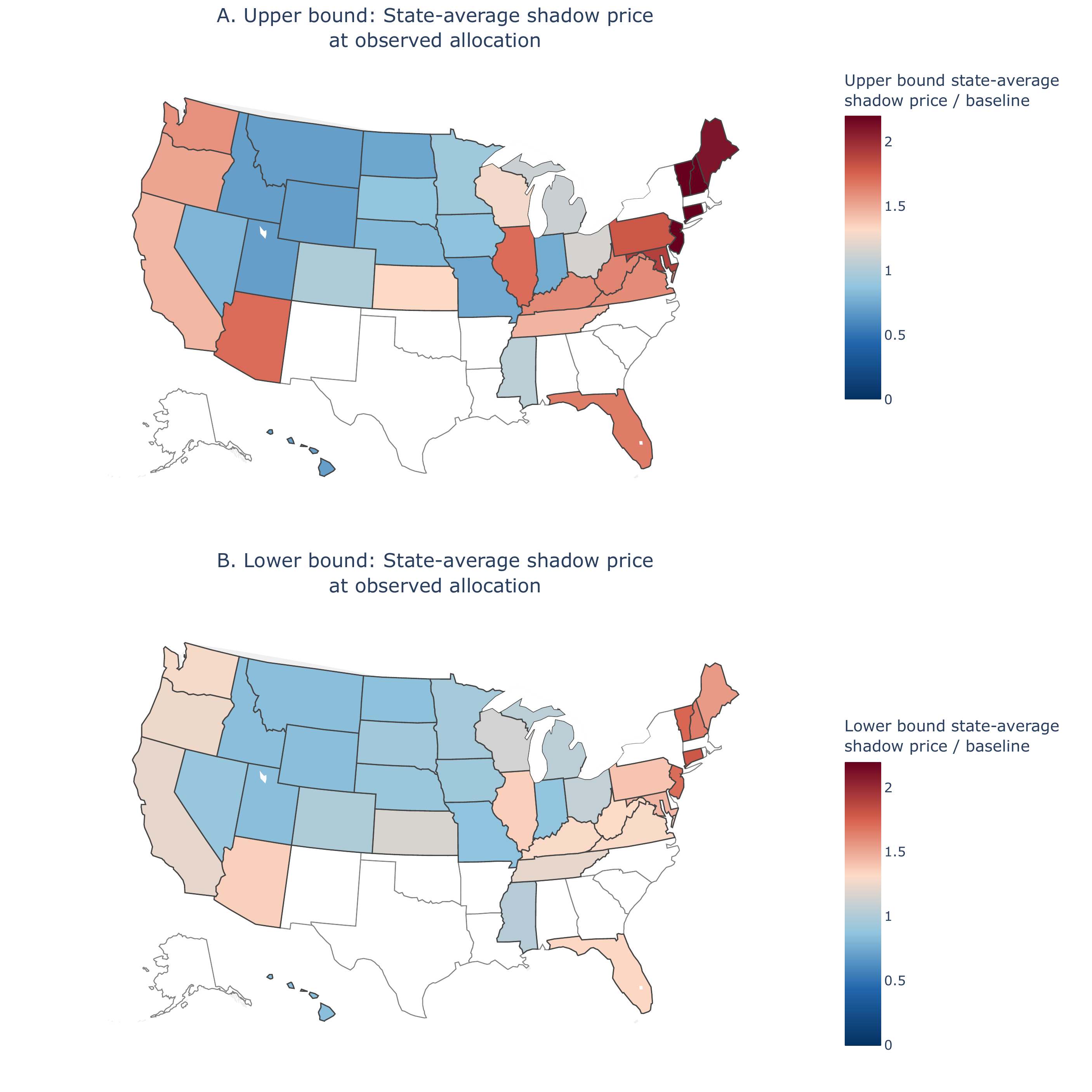}
    \caption{\textbf{State-Average Shadow Prices: No Imputation.} Same construction as Figure~\ref{fig:shadow_price_map}, but using only the 36 states with observed 1972 gallon sales (no imputation). States shown in white lack gallon data. The pattern is similar to the main specification with imputed states included.}
    \label{fig:shadow_price_map_known_exact}
\end{figure}

\newpage

\section{Supplementary Appendix}

\subsection{Proofs of Chaos Theorem Lemmas}

\paragraph{Lemma \ref{lem:trans-app}.} \textit{Let $v\neq w$ be adjacent vertices and $c^\star$ a simple tie for $v,w$. Then for any neighborhood $U\subset\R^n$ of $c^\star$ there exists $\varepsilon > 0$ and a smooth function $\gamma \colon \left[-\varepsilon,\varepsilon\right] \to U$ such that \(\gamma(0)=c^\star\); and the unique optimizer of \eqref{lp} equals $v$ for $t<0$ and $w$ for $t>0$,
while at $t=0$ both $v$ and $w$ are optimal.}

\begin{proof}[Proof of Lemma \ref{lem:trans-app}]
As \(c^\star\in\operatorname{relint}\mathcal H_{vw}\) and no third cone meets at \(c^\star\), and since \(U\) is a neighborhood of \(c^\star\), there exists \(\rho>0\) such that \(B\left(c^\star,\rho\right)\subset U\) and
\[
B\left(c^\star,\rho\right)=\left(B\left(c^\star,\rho\right)\cap\operatorname{int}\mathcal N_v\right)\dot\cup\left(B\left(c^\star,\rho\right)\cap\mathcal H_{vw}\right)\dot\cup\left(B\left(c^\star,\rho\right)\cap\operatorname{int}\mathcal N_w\right).
\]
Choose any \(\eta\) with \((w-v)\cdot\eta<0\), and define \(\gamma(t)=c^\star+t\eta\). Then
\[
(w-v)\cdot\gamma(t)=t(w-v)\cdot\eta,
\]
so \(\gamma\) crosses \(\mathcal H_{vw}\) with a sign change. Hence, for \(|t|\) sufficiently small,
\[
\gamma(t)\in\operatorname{int}\mathcal N_v\text{ for }t<0,\qquad \gamma(0)\in\mathcal H_{vw},\qquad \text{and} \qquad \gamma(t)\in\operatorname{int}\mathcal N_w\text{ for }t>0.
\]
Choosing \(\varepsilon>0\) small enough that \(\gamma\left(\left[-\varepsilon,\varepsilon\right]\right)\subset B\left(c^\star,\rho\right)\), the claim follows. Optimality follows from Remark~\ref{opt-nc}.
\end{proof}

\paragraph{Corollary \ref{cor:trans-app}.} \textit{Let \(v\neq w\) be adjacent vertices, and let \(c^\star\) be any cost vector for which both \(v\) and \(w\) are optimal for \eqref{lp}. Then for every open neighborhood \(U\subset\mathbb{R}^n\) of \(c^\star\) there exist \(\hat c\in U\), \(\varepsilon>0\), and a smooth function
\(\gamma\colon\left[-\varepsilon,\varepsilon\right]\to U\) such that \(\gamma(0)=\hat c\), both \(v\) and \(w\) are optimal at \(t=0\), the unique optimizer of \eqref{lp} equals \(v\) for \(t<0\), and the unique optimizer equals \(w\) for \(t>0\).}

\begin{proof}[Proof of Corollary \ref{cor:trans-app}]
If \(c^\star\) is a simple tie for \(v,w\), the conclusion is exactly Lemma~\ref{lem:trans-app}. Otherwise, since \(v\) and \(w\) are adjacent, the segment \([v,w]\) is an edge of \(\mathcal F\), and \(\mathcal H_{vw}\) equals the normal cone of that edge. In particular, every \(c\in\operatorname{relint}\left(\mathcal H_{vw}\right)\) exposes exactly the edge \([v,w]\); equivalently, \(v\) and \(w\) are both optimal and no other vertex is optimal. Since \(c^\star\in\mathcal H_{vw}\) and \(\operatorname{relint}\left(\mathcal H_{vw}\right)\) is dense in \(\mathcal H_{vw}\), every open neighborhood \(U\subset\mathbb R^n\) of \(c^\star\) contains some \(\hat c\in U\cap\operatorname{relint}\left(\mathcal H_{vw}\right)\). Thus \(\hat c\) is a simple tie for \(v,w\). Applying Lemma~\ref{lem:trans-app} at \(\hat c\) and shrinking \(\varepsilon\) if necessary yields the stated variation with image contained in \(U\).
\end{proof}

\iffalse
\paragraph{Lemma \ref{lem:jump-app}.} \textit{\(\lim_{t\downarrow0}\left(\sum_i\int_0^{w_i}P_i(x)\,dx-\gamma(t)\cdot w\right)-\lim_{t\uparrow0}\left(\sum_i\int_0^{v_i}P_i(x)\,dx-\gamma(t)\cdot v\right)=\Delta W\).}

\begin{proof}[Proof of Lemma \ref{lem:jump-app}]
By definition,
\[
\begin{aligned}
&\lim_{t\downarrow0}\left(\sum_i\int_0^{w_i}P_i(x)\,dx-\gamma(t)\cdot w\right)-\lim_{t\uparrow0}\left(\sum_i\int_0^{v_i}P_i(x)\,dx-\gamma(t)\cdot v\right)\\
&=\sum_i\left(\int_0^{w_i}P_i(x)\,dx-\int_0^{v_i}P_i(x)\,dx\right)-\lim_{t\to0}\gamma(t)\cdot(w-v)\\
&=\sum_i\int_{v_i}^{w_i}P_i(x)\,dx-\gamma(0)\cdot(w-v)\\
&=\int_{v_s}^{w_s}P_s(x)\,dx-\int_{w_r}^{v_r}P_r(x)\,dx\\
&=\Delta W\text{,}
\end{aligned}
\]
where we used the definition of a simple tie and Remark~\ref{tie-orth} to eliminate \(\gamma(0)\cdot(w-v)=0\), and the last equality is the definition of \(\Delta W\).
\end{proof}
\fi

\subsection{Proofs of Robust Bounds Lemmas}

\paragraph{Lemma \ref{lem:a}.} \textit{Under Assumptions \ref{mainass} and \ref{ass:pathA}, the functions \(\Psi^+(\cdot)\) and \(\Psi^-(\cdot)\) are continuous on \(\mathcal I\).}
    
\begin{proof}[Proof of Lemma \ref{lem:a}]
    We prove continuity of \(\Psi^+\); the argument for \(\Psi^-\) is identical.
    
    Under Assumption \ref{ass:pathA}, on \(\mathcal{I}\), \(\ell_i(\cdot)\) and \(u_i(\cdot)\) are given by \eqref{eq:simpell} and \eqref{eq:simpu}, and so are continuous on \(\mathcal{I}\). Evidently, so too are \(q_i^+(\cdot)\), \(q_i^-(\cdot)\), \(\Delta^+(\cdot)\) and \(\Delta^-(\cdot)\).

    For each \(i\) and \(p \in \mathcal{I}\), define the continuous functions
    \[
\underline z_i^+(p)\coloneqq\ell_i(p)-q_i^+(p),\qquad \text{and} \qquad \overline z_i^+(p)\coloneqq u_i(p)-q_i^+(p).
\]
Note that \(\underline z_i^+(p)\le0\le\overline z_i^+(p)\) for all \(i\), and
\[
\sum_{i=1}^n\underline z_i^+(p)=L(p)-\sum_{i=1}^n q_i^+(p)\le\bar Q-\sum_{i=1}^n q_i^+(p)=\Delta^+(p)\le U(p)-\sum_{i=1}^n q_i^+(p)=\sum_{i=1}^n\overline z_i^+(p),
\]
where the inequality uses \(p \in \mathcal{I}\), i.e., \(L(p)\le\bar Q\le U(p)\). Thus, the set
\[
Z^+(p)\coloneqq\left\{z\in\mathbb R^n\colon\sum_{i=1}^n z_i=\Delta^+(p),\ \underline z_i^+(p)\le z_i\le\overline z_i^+(p)\ \forall i\right\}
\]
is nonempty and compact. Define the optimization-problem
\[\tag{\(D1\)}\label{eq:psiplus}
\widetilde\Psi^+(p)\coloneqq\min_{z\in Z^+(p)}\sum_{i=1}^n\frac{z_i^2}{2d_i}.
\]
\begin{claim}
    \(\widetilde\Psi^+(p)=\Psi^+(p)\).
\end{claim}
\begin{proof}
If \(\Delta^+(p)=0\), both values equal \(0\).
If \(\Delta^+(p)>0\), let \(z\) be a minimizer of \(\widetilde\Psi^+(p)\). If some coordinate satisfies \(z_i<0\), define \(N\coloneqq\left\{i\colon z_i<0\right\}\) and \(m\coloneqq-\sum_{i\in N}z_i>0\). Since \(\sum_i z_i=\Delta^+(p)>0\), we have \(\sum_{i\colon z_i>0}z_i=\Delta^+(p)+m\). Construct \(\tilde z\) by setting \(\tilde z_i=0\) for \(i\in N\) and reducing the positive coordinates (each of which has lower bound \(0\)) by total mass \(m\) so that \(\sum_i\tilde z_i=\Delta^+(p)\). This preserves feasibility because it only increases negative coordinates up to \(0\) and decreases positive coordinates down toward \(0\). Moreover, \(\sum_i \tilde z_i^2/(2d_i)<\sum_i z_i^2/(2d_i)\), contradicting optimality. Thus every minimizer must satisfy \(z_i\ge0\) for all \(i\). For indices with \(p<p_{0,i}\) we have \(\overline z_i^+(p)=u_i(p)-q_i^+(p)=0\), so \(z_i=0\) there. For indices with \(p\ge p_{0,i}\) we have \(\underline z_i^+(p)=\ell_i(p)-q_i^+(p)=0\) and \(\overline z_i^+(p)=u_i(p)-\ell_i(p)\). Therefore \eqref{eq:psiplus} reduces to
\[
\min_{\left\{z_i\right\}_{i\colon p\ge p_{0,i}}}\left\{\sum_{i\colon p\ge p_{0,i}}\frac{z_i^2}{2d_i}\right\}\quad\text{s.t.}\quad \sum_{i\colon p\ge p_{0,i}}z_i=\Delta^+(p),\quad 0\le z_i\le u_i(p)-\ell_i(p),
\]
which is exactly \(\Psi^+(p)\) when \(\Delta^+(p)>0\) (since then \(A^+(p)=\left\{i\colon p\ge p_{0,i}\right\}\) and \(\left|\Delta^+(p)\right|=\Delta^+(p)\)).
The case \(\Delta^+(p)<0\) is symmetric: any minimizer must satisfy \(z_i\le0\) for all \(i\); then \(z_i=0\) for all \(i\) with \(p\ge p_{0,i}\) (because \(\underline z_i^+(p)=0\) there), and setting \(\delta_i\coloneqq-z_i\ge0\) on \(\left\{i\colon p<p_{0,i}\right\}\) yields exactly the definition of \(\Psi^+(p)\) when \(\Delta^+(p)<0\) (since then \(A^+(p)=\left\{i\colon p<p_{0,i}\right\}\) and \(\left|\Delta^+(p)\right|=-\Delta^+(p)\)).\end{proof}
Since \(Z^+(p)\) is a nonempty compact polytope defined by continuous bounds and a continuous affine equality, the correspondence \(p \mapsto Z^+(p)\) is continuous. Consequently, by Berge's Maximum theorem \(\widetilde\Psi^+(\cdot)\) (and so \(\Psi^+(\cdot)\)) is continuous on \(\mathcal{I}\).
\end{proof}

\paragraph{Lemma \ref{lem:2}.} \textit{Posit Assumptions \ref{mainass} and \ref{ass:pathA} and fix \(p \in \mathcal{I}\). Then, the inner problem of maximizing \(\Phi(P)\) over \(P \in \mathcal{P}\) satisfying \(p^*(P) = p\) is equivalent to minimizing \(\int_{p_{0,i}}^{p}q_i\) for \(i\in\mathcal{I}_+(p)\) and maximizing \(\int_{p}^{p_{0,i}}q_i\) for \(i\in\mathcal{I}_-(p)\), subject to feasibility and \(\sum_i q_i(p)=\bar{Q}\).}
\begin{proof}[Proof of Lemma \ref{lem:2}]
    Fix \(p \in \mathcal{I}\) and \(i\). Under Assumption~\ref{ass:pathA}, the equal-shadow allocation satisfies \(0<q_i^*(P)<q_i^{\max}\), so \(P_i\left(q_i^*(P)\right)=p^*\). By the fundamental theorem of calculus,
\[
\int_{q_i^{\mathrm{obs}}}^{q_i^*(P)}\left(P_i(x)-p^*\right)d x
=q_i^{\mathrm{obs}}\left(p^*-p_{0,i}\right)-\int_{p_{0,i}}^{p^*}q_i(s)d s;
\tag{\(D2\)}\label{eq:id_i}
\]
and, summing \eqref{eq:id_i} over \(i\) and using \(\sum_i q_i^{\mathrm{obs}}=\bar{Q}\) we obtain
\[
\Phi(P)=\bar{Q}p^*-\sum_{i=1}^n q_i^{\mathrm{obs}}p_{0,i}-\sum_{i=1}^n\int_{p_{0,i}}^{p^*}q_i(s)d s.
\tag{\(D3\)}\label{eq:id_sum}
\]

%Next, as \(q_i(p^*)=q_i^*(P)\) and \(q_i(p)\in\left[\ell_i(p),u_i(p)\right]\) for all \(p\),
Next, as \(q_i(p^*)=q_i^*(P)\) and \(q_i(p^*)\in\left[\ell_i(p^*),u_i(p^*)\right]\) for all \(i\)
\[
L(p^*)\le\sum_{i=1}^n q_i(p^*)=\sum_{i=1}^n q_i^*(P)=\bar{Q}\le U(p^*) \quad \Longrightarrow \quad p^*\in\mathcal{I}.
\]
Observe that
\[
-\int_{p_{0,i}}^{p}q_i(s)d s=
\begin{cases}
-\int_{p_{0,i}}^{p}q_i(s)d s,&i\in\mathcal{I}_+(p),\\[4pt]
\int_{p}^{p_{0,i}}q_i(s)d s,&i\in\mathcal{I}_-(p).
\end{cases}
\]
Consequently, conditional on \(p\), maximizing \(\Phi(P)\) is equivalent to minimizing \(\int_{p_{0,i}}^{p}q_i\) for \(i\in\mathcal{I}_+(p)\) and maximizing \(\int_{p}^{p_{0,i}}q_i\) for \(i\in\mathcal{I}_-(p)\), subject to \(\sum_i q_i(p)=\bar{Q}\).
\end{proof}

\paragraph{Lemma \ref{lem:3}.}
    \textit{\(q_i'(s)\in\left[\alpha_i,\beta_i\right]\) for a.e. \(s\in\left[\min\left\{p,p_{0,i}\right\},\max\left\{p,p_{0,i}\right\}\right]\).}
\begin{proof}[Proof of Lemma \ref{lem:3}]
For any \(x_1<x_2\),
\[
g_{i,L}\le\frac{P_i(x_2)-P_i(x_1)}{x_2-x_1}\le g_{i,U},
\]
and inverting this secant inequality yields
\[
\alpha_i\le\frac{q_i(p_2)-q_i(p_1)}{p_2-p_1}\le\beta_i
\quad\text{for all }p_1<p_2\text{ in the relevant range}.
\]
Take limits along the differentiability points of the Lipschitz function \(q_i\).
\end{proof}

\paragraph{Lemma \ref{lem:4}.} \textit{Now let \(t_0<t_1\), \(\alpha<\beta\), and let function \(q\) be absolutely continuous on \(\left[t_0,t_1\right]\) with \(q(t_0)=q_0\) and \(q'(t)\in\left[\alpha,\beta\right]\) a.e. Define \(\underline{q}(t)\coloneqq q_0+\alpha\left(t-t_0\right)\).
    If \(q(t_1)=\underline{q}(t_1)+\delta\) for some \(\delta\in\left[0,\left(\beta-\alpha\right)\left(t_1-t_0\right)\right]\), then}
\[
\int_{t_0}^{t_1}q(t)d t\ge\int_{t_0}^{t_1}\underline{q}(t)d t+\frac{\delta^2}{2\left(\beta-\alpha\right)}.
\tag{\(D4\)}\label{eq:triangle}
\]
\textit{Moreover, equality holds if and only if \(q\) is ``bang-bang;'' \(q'(t)=\alpha\) a.e. on \(\left[t_0,t_1-h\right]\) and \(q'(t)=\beta\) a.e. on \(\left[t_1-h,t_1\right]\), where \(h=\delta/\left(\beta-\alpha\right)\).}
\begin{proof}[Proof of Lemma \ref{lem:4}]
Write \(q'(t)=\alpha+u(t)\) with \(u(t)\in[0,\beta-\alpha]\) a.e. on \([t_0,t_1]\). The two constraints on \(u\) are:
\[
0\le u(t)\le\beta-\alpha\quad\text{for a.e. }t\in[t_0,t_1]\qquad\text{(box constraint),}
\]
and
\[
\int_{t_0}^{t_1}u(t)dt=\delta\qquad\text{(total mass constraint).}
\]
Let \(\underline{q}(t)\coloneqq q_0+\alpha(t-t_0)\). Then
\[
q(t)=\underline{q}(t)+\int_{t_0}^{t}u(s)ds.
\]
Integrating over \([t_0,t_1]\) and using Fubini gives
\[
\int_{t_0}^{t_1}q(t)dt=\int_{t_0}^{t_1}\underline{q}(t)dt+\int_{t_0}^{t_1}\left(\int_{t_0}^{t}u(s)ds\right)dt
=\int_{t_0}^{t_1}\underline{q}(t)dt+\int_{t_0}^{t_1}(t_1-s)u(s)ds.
\]
Thus, conditional on \(\delta\), minimizing \(\int_{t_0}^{t_1}q(t)dt\) is equivalent to minimizing
\[
\min_{u}\ \int_{t_0}^{t_1}(t_1-s)u(s)ds
\quad\text{s.t.}\quad
0\le u(s)\le\beta-\alpha\ \text{a.e.},\ \int_{t_0}^{t_1}u(s)ds=\delta.
\tag{$D5$}\label{eq:aux1}
\]
Define the cumulative mass \(U(t)\coloneqq\int_{t_0}^{t}u(s)ds\). \(U\) is absolutely continuous, nondecreasing, satisfies \(U(t_0)=0\), \(U(t_1)=\delta\), and
\(U'(t)=u(t)\in[0,\beta-\alpha]\) a.e. Integrating by parts:
\[
\int_{t_0}^{t_1}(t_1-s)u(s)ds=\int_{t_0}^{t_1}(t_1-s)U'(s)ds
=\left[(t_1-s)U(s)\right]_{t_0}^{t_1}-\int_{t_0}^{t_1}(-1)U(s)ds
=\int_{t_0}^{t_1}U(s)ds.
\]
So the program in \eqref{eq:aux1} is equivalent to
\[
\min_{U}\ \int_{t_0}^{t_1}U(t)dt
\quad\text{s.t.}\quad
U(t_0)=0,\ U(t_1)=\delta,\ 0\le U'(t)\le\beta-\alpha\ \text{a.e.}
%\tag{$A8$}\label{eq:aux2}
\]
Let \(h\coloneqq\delta/(\beta-\alpha)\in[0,t_1-t_0]\). Since \(U'(s)\le\beta-\alpha\) a.e., for any \(t\in[t_0,t_1]\),
\[
\delta-U(t)=U(t_1)-U(t)=\int_{t}^{t_1}U'(s)ds\le\int_{t}^{t_1}(\beta-\alpha)ds=(\beta-\alpha)(t_1-t),
\]
hence,
\[
U(t)\ge\delta-(\beta-\alpha)(t_1-t)=(\beta-\alpha)\left(t-(t_1-h)\right).
\]
Also \(U(t)\ge0\). Therefore, for all \(t\in[t_0,t_1]\),
\[
U(t)\ge U^{*}(t)\coloneqq\max\left\{0,(\beta-\alpha)\left(t-(t_1-h)\right)\right\}.
\]
Integrating both sides yields
\[
\int_{t_0}^{t_1}U(t)dt\ge\int_{t_0}^{t_1}U^{*}(t)dt
=\int_{t_1-h}^{t_1}(\beta-\alpha)\left(t-(t_1-h)\right)dt
=\frac{(\beta-\alpha)h^2}{2}
=\frac{\delta^2}{2(\beta-\alpha)}.
\]
Moreover, equality holds if and only if \(U(t)=U^{*}(t)\) for all \(t\), which forces
\(U'(t)=0\) a.e. on \([t_0,t_1-h]\) and \(U'(t)=\beta-\alpha\) a.e. on \([t_1-h,t_1]\), i.e.
\[
u(t)=
\begin{cases}
0,&t\in[t_0,t_1-h]\ \text{a.e.},\\
\beta-\alpha,&t\in[t_1-h,t_1]\ \text{a.e.}
\end{cases}
\]
That is, the optimizer concentrates all mass \(\delta\) on the terminal interval \([t_1-h,t_1]\), saturating the box constraint.
This is a special case of the bathtub principle (the Hardy-Littlewood rearrangement) for minimizing a linear functional under an integral constraint.
\end{proof}

\paragraph{Claim \ref{clm:active_set_deviations}}
\textit{Fix \(p\in \mathcal I\).}
\begin{enumerate}
    \item \textit{If \(\Delta^+(p)>0\), there exist numbers \(\left\{\delta_i\right\}_{i\in \mathcal{I}_+(p)}\) such that \(0\le\delta_i\le u_i(p)-\ell_i(p)\) for all \(i\in \mathcal{I}_+(p)\) and \(\sum_{i\in \mathcal{I}_+(p)}\delta_i=\Delta^+(p)\), and the endpoint vector \(q(p)\) defined by}
\[
q_i(p)=\ell_i(p)+\delta_i\ \text{for }i\in \mathcal{I}_+(p),\qquad q_i(p)=u_i(p)\ \text{for }i\in \mathcal{I}_-(p)
\]
\textit{satisfies \(\sum_{i=1}^n q_i(p)=\bar{Q}\). Moreover, in the inner problem, there is an optimizer with \(q_i(p)=u_i(p)\) for all \(i\in \mathcal{I}_-(p)\); equivalently, deviations from the baseline endpoints \(q_i^+(p)\) may be taken without loss of optimality to occur only on \(\mathcal{A}^+(p)=\mathcal{I}_+(p)\).}

\item \textit{If \(\Delta^+(p)<0\), there exist numbers \(\left\{\delta_i\right\}_{i\in \mathcal{I}_-(p)}\) such that \(0\le\delta_i\le u_i(p)-\ell_i(p)\) for all \(i\in \mathcal{I}_-(p)\) and \(\sum_{i\in \mathcal{I}_-(p)}\delta_i=-\Delta^+(p)\), and the endpoint vector \(q(p)\) defined by}
\[
q_i(p)=\ell_i(p)\ \text{for }i\in \mathcal{I}_+(p),\qquad q_i(p)=u_i(p)-\delta_i\ \text{for }i\in \mathcal{I}_-(p)
\]
\textit{satisfies \(\sum_{i=1}^n q_i(p)=\bar{Q}\). Moreover, in the inner problem, there is an optimizer with \(q_i(p)=\ell_i(p)\) for all \(i\in \mathcal{I}_+(p)\); equivalently, deviations from the baseline endpoints \(q_i^+(p)\) may be taken without loss of optimality to occur only on \(\mathcal{A}^+(p)=\mathcal{I}_-(p)\).}
\end{enumerate}

\begin{proof}[Proof of Claim \ref{clm:active_set_deviations}]
We prove (i) as the proof of (ii) is identical with signs reversed.

Assume \(\Delta^+(p)>0\). Since \(p\in \mathcal{I}\), we have \(\bar{Q}\le U(p)=\sum_{i=1}^n u_i(p)\). Using
\[
\sum_{i=1}^n q_i^+(p)=\sum_{i\in \mathcal{I}_+(p)}\ell_i(p)+\sum_{i\in \mathcal{I}_-(p)}u_i(p),
\]
we obtain
\[
\Delta^+(p)=\bar{Q}-\sum_{i\in \mathcal{I}_+(p)}\ell_i(p)-\sum_{i\in \mathcal{I}_-(p)}u_i(p)
\le\sum_{i\in \mathcal{I}_+(p)}\left(u_i(p)-\ell_i(p)\right).
\]
Thus the total upward slack available in the coordinates \(i\in \mathcal{I}_+(p)\) is at least \(\Delta^+(p)\), so we can choose \(\left\{\delta_i\right\}_{i\in \mathcal{I}_+(p)}\) with \(0\le\delta_i\le u_i(p)-\ell_i(p)\) and \(\sum_{i\in \mathcal{I}_+(p)}\delta_i=\Delta^+(p)\). Defining \(q(p)\) as in the statement then yields \(\sum_i q_i(p)=\bar{Q}\) by construction.

For the “without loss” part, let \(\left\{q_i(\cdot)\right\}_{i=1}^n\) be an optimizer for the inner problem at this \(p\).
Conditional on \(p\), the objective depends on \(q_i(\cdot)\) only through the signed integrals
\[
-\int_{p_{0,i}}^{p}q_i(s) ds\ \text{for }i\in \mathcal{I}_+(p),\qquad \text{and} \quad \int_{p}^{p_{0,i}}q_i(s) ds\ \text{for }i\in \mathcal{I}_-(p).
\]
By Lemmas \ref{lem:3} and \ref{lem:4}, for each market \(k\) and each admissible endpoint value \(q_k(p)\in\left[\ell_k(p),u_k(p)\right]\),
the choice of \(q_k(\cdot)\) that optimizes its signed integral term is the Lemma \ref{lem:4} equality-case (“bang-bang”) form.
Moreover, the resulting value depends on the endpoint only through a quadratic penalty in the endpoint gap:
it is strictly decreasing in \(q_k(p)-\ell_k(p)\) when \(k\in\mathcal{I}_+(p)\), and strictly decreasing in \(u_k(p)-q_k(p)\) when \(k\in\mathcal{I}_-(p)\).

If an optimizer had some \(i\in \mathcal{I}_-(p)\) with \(q_i(p)<u_i(p)\), then (because \(\sum_{k=1}^n q_k(p)=\bar{Q}>\sum_{k=1}^n q_k^+(p)\))
there must exist some \(j\in \mathcal{I}_+(p)\) with \(q_j(p)>\ell_j(p)\).
Take
\[
\varepsilon\in\left(0,\min\left\{u_i(p)-q_i(p),\,q_j(p)-\ell_j(p)\right\}\right],
\]
and modify only these two endpoints by setting
\[
\tilde q_i(p)\coloneqq q_i(p)+\varepsilon,\qquad \tilde q_j(p)\coloneqq q_j(p)-\varepsilon,
\]
leaving \(\tilde q_k(p)\coloneqq q_k(p)\) for all \(k\notin\left\{i,j\right\}\).
This preserves \(\sum_{k=1}^n \tilde q_k(p)=\bar{Q}\) and the box constraints \(\tilde q_k(p)\in\left[\ell_k(p),u_k(p)\right]\).
Now choose for markets \(i\) and \(j\) the Lemma \ref{lem:4} equality-case forms under the modified endpoints (leaving all other markets unchanged).
Since the \(\varepsilon\)-transfer strictly reduces both endpoint gaps \(u_i(p)-q_i(p)\) and \(q_j(p)-\ell_j(p)\), it strictly reduces the associated quadratic penalties
and therefore strictly increases the objective, contradicting optimality.
Hence, every optimizer must satisfy \(q_i(p)=u_i(p)\) for all \(i\in\mathcal{I}_-(p)\), so deviations from \(q_i^+(p)\) may be taken to occur only on
\(\mathcal{I}_+(p)=\mathcal{A}^+(p)\).
\end{proof}

	\newpage
	\singlespacing 

			\renewcommand\refname{References}		
			\printbibliography

@misc{jec1973gasoline,
    address = {Washington, DC},
    title = {The gasoline and fuel oil shortage: {Hearings} before the subcommittee on consumer economics of the joint economic committee, congress of the united states, ninety-third congress, first session, may 1, 2, and june 2, 1973},
    publisher = {U.S. Government Printing Office},
    author = {{U.S. Congress, Joint Economic Committee, Subcommittee on Consumer Economics}},
    year = {1973},
    note = {Number: 99-740 O
tex.entrytype: government},
}

@techreport{federal_energy_office_review_1974,
    title = {Review of {Complaints} {Concerning} the {Mandatory} {Petroleum} {Allocation} {Program} {And} {The} {Regulation} of {Petroleum} {Pricing}},
    url = {https://www.gao.gov/assets/b-178205-090294.pdf},
    number = {090294},
    author = {Federal Energy Office},
    year = {1974},
}

@misc{time_controls_1973,
    title = {{CONTROLS}: {A} {Threat} of {Food} {Shortage}},
    shorttitle = {{CONTROLS}},
    url = {https://time.com/archive/6841371/controls-a-threat-of-food-shortage/},
    urldate = {2025-03-19},
    journal = {TIME},
    author = {TIME},
    month = jul,
    year = {1973},
}

@article{kang_robustness_2025,
	title = {Robustness {Measures} for {Welfare} {Analysis}},
	volume = {115},
	number = {8},
	journal = {American Economic Review},
	author = {Kang, Zi Yang and Vasserman, Shoshana},
	year = {2025},
	pages = {2449--2487},
}

@misc{eia_petroleum_products_supplied_annual,
	title = {Monthly Energy Review, Table 3.1: Petroleum Overview, Annual (1949--2024)},
	author = {{U.S. Energy Information Administration}},
	year = {2025},
	url = {https://www.eia.gov/totalenergy/data/monthly/index.php},
	note = {Accessed via Monthly Energy Review tables; products supplied series used for pre-shock trend comparison.},
}

@misc{fhwa_net_volume_gallons_taxed,
	title = {Highway Statistics, Table MF-202: Motor-Fuel Use},
	author = {{Federal Highway Administration}},
	year = {2012},
	url = {https://www.fhwa.dot.gov/policyinformation/statistics/2012/mf202.cfm},
	note = {Annual net volume of gasoline taxed; used as a gasoline-specific cross-check on 1974 trend shortfall.},
}

@article{new_york_times_baby_1973,
    chapter = {Archives},
    title = {Baby Chicks Killed and Cooked for Feed},
    issn = {0362-4331},
    url = {https://www.nytimes.com/1973/06/25/archives/baby-chicks-killed-and-cooked-for-feed.html},
    urldate = {2026-02-01},
    journal = {The New York Times},
    author = {{The New York Times}},
    month = jun,
    year = {1973},
}

@article{associated_press_chicks_1973,
    chapter = {Archives},
    title = {Chicks {Smothered} as {Farmers} {Decry} {Price} {Freeze}},
    issn = {0362-4331},
    url = {https://www.nytimes.com/1973/06/27/archives/chicks-smothered-as-farmers-decry-price-freeze.html},
    urldate = {2025-03-19},
    journal = {The New York Times},
    author = {{Associated Press}},
    month = jun,
    year = {1973},
}

@book{bradley_oil_1997,
    title = {Oil, {Gas} and {Government}: {The} {U}.{S}. {Experience}},
    isbn = {0-8476-8108-4},
    shorttitle = {Oil, {Gas} and {Government}},
    publisher = {Rowman \& Littlefield Inc},
    author = {Bradley, Robert L.},
    month = oct,
    year = {1997},
}

@misc{time_shortages_1974,
    title = {{SHORTAGES}: {Gas} {Fever}: {Happiness} {Is} a {Full} {Tank}},
    shorttitle = {{SHORTAGES}},
    url = {https://time.com/archive/6875854/shortages-gas-fever-happiness-is-a-full-tank/},
    urldate = {2026-02-07},
    journal = {TIME},
    author = {TIME},
    month = feb,
    year = {1974},
}

@article{gupte_fuel_1979,
    chapter = {Archives},
    title = {{FUEL} {CRISIS} {TERMED} {CRITICAL} {AT} {PUMPS} {IN} {NEW} {YORK} {REGION}},
    issn = {0362-4331},
    url = {https://www.nytimes.com/1979/06/25/archives/fuel-crisis-termed-critical-at-pumps-in-new-york-region-long-lines.html},
    urldate = {2026-02-07},
    journal = {The New York Times},
    author = {Gupte, Pranay B.},
    month = jun,
    year = {1979},
}

@misc{turkel_2023,
    title = {Fifty years ago, the energy crisis officially began. And it hasn't ended.},
    url = {https://themainemonitor.org/fifty-years-ago-the-climate-crisis-began-and-it-hasnt-ended/},
    journal = {The Maine Monitor},
    author = {Turkel, Tux},
    month = sep,
    year = {2023},
    urldate = {2026-02-10},
}

@article{mulligan_supply_chain_2025,
    title = {Equilibrium Responses to Price Controls: A Supply-Chain Approach},
    volume = {203},
    journal = {Public Choice},
    author = {Mulligan, Casey B.},
    year = {2025},
    pages = {23--52},
    doi = {10.1007/s11127-024-01196-8},
}

@unpublished{kang_robust_regulation_2026,
    title = {Robust Regulation: Prices vs.\ Quantities},
    author = {Kang, Zi Yang},
    year = {2026},
    note = {Working paper, University of Toronto},
}

@techreport{aaa_aaa_1974,
    address = {Falls Church, VA},
    title = {{AAA} {Fuel} {Gauge} {Report} {II}},
    number = {5},
    author = {AAA},
    month = feb,
    year = {1974},
}

@misc{zotero-item-8347,
    title = {Emergency {Petroleum} {Allocation} {Extension} {Act} of 1974: {Hearing}},
    publisher = {U.S. Govt. Print. Off., 1975},
    author = {US Congress},
    month = jul,
    year = {1974},
}

@article{hausman_individual_2016,
	title = {Individual {Heterogeneity} and {Average} {Welfare}},
	volume = {84},
	number = {3},
	journal = {Econometrica},
	author = {Hausman, Jerry A. and Newey, Whitney K.},
	year = {2016},
	pages = {1225--1248},
}

@article{hausman_nonparametric_1995,
	title = {Nonparametric {Estimation} of {Exact} {Consumers} {Surplus} and {Deadweight} {Loss}},
	volume = {63},
	number = {6},
	journal = {Econometrica},
	author = {Hausman, Jerry A. and Newey, Whitney K.},
	year = {1995},
	pages = {1445--1476},
}

@techreport{kremer_worst-case_2018,
	type = {Working {Paper}},
	title = {Worst-{Case} {Bounds} on {R}\&{D} and {Pricing} {Distortions}},
	number = {25119},
	institution = {NBER},
	author = {Kremer, Michael and Snyder, Christopher M.},
	year = {2018},
}

@article{buurma-olsen_quantifying_2025,
	title = {Quantifying misallocation of public housing},
	volume = {242},
	issn = {0047-2727},
	url = {https://www.sciencedirect.com/science/article/pii/S0047272724002081},
	doi = {10.1016/j.jpubeco.2024.105272},
	urldate = {2026-01-14},
	journal = {Journal of Public Economics},
	author = {Buurma-Olsen, Jennifer and Koster, Hans R. A. and van Ommeren, Jos and Damsté, Jort Sinninghe},
	month = feb,
	year = {2025},
	pages = {105272},
}

@article{deacon_rationing_1985,
	title = {Rationing by {Waiting} and the {Value} of {Time}: {Results} from a {Natural} {Experiment}},
	volume = {93},
	number = {4},
	journal = {Journal of Political Economy},
	author = {Deacon, Robert T. and Sonstelie, Jon},
	year = {1985},
	pages = {627--647},
}

@article{akbarpour_redistributive_2024,
	title = {Redistributive {Allocation} {Mechanisms}},
	volume = {132},
	number = {6},
	journal = {Journal of Political Economy},
	author = {Akbarpour, Mohammad and Dworczak, Piotr and Kominers, Scott Duke},
	year = {2024},
	pages = {1831--1875},
}

@article{condorelli_market_2013,
	title = {Market and {Non}-{Market} {Mechanisms} for the {Optimal} {Allocation} of {Scarce} {Resources}},
	volume = {82},
	journal = {Games and Economic Behavior},
	author = {Condorelli, Daniele},
	year = {2013},
	pages = {582--591},
}

@article{barzel_theory_1974,
	title = {A {Theory} of {Rationing} by {Waiting}},
	volume = {17},
	number = {1},
	journal = {Journal of Law and Economics},
	author = {Barzel, Yoram},
	year = {1974},
	pages = {73--95},
}

@article{poole_fuel_1973,
	title = {Fuel {Shortage}: {A} {Prehistory}},
	journal = {Reason},
	author = {Poole, Robert W.},
	year = {1973},
}

@book{grunbaum_convex_2003,
	address = {New York},
	edition = {2nd},
	series = {Graduate {Texts} in {Mathematics}},
	title = {Convex {Polytopes}},
	volume = {221},
	publisher = {Springer},
	author = {Grünbaum, Branko},
	year = {2003},
}

@article{deacon_price_1980,
	title = {Price {Controls} and {International} {Petroleum} {Product} {Prices}},
	author = {Deacon, Robert T. and Mead, Walter J. and Agarwal, Vinod B.},
	year = {1980},
	note = {Publisher: Department of Energy},
}

@article{verleger_us_1979,
	title = {The {U}.{S}. {Petroleum} {Crisis} of 1979},
	volume = {1979},
	number = {2},
	journal = {Brookings Papers on Economic Activity},
	author = {Verleger, Philip K.},
	year = {1979},
	pages = {463--476},
}

@article{kholodilin_rent_2024,
	title = {Rent control effects through the lens of empirical research: {An} almost complete review of the literature},
	volume = {63},
	journal = {Journal of Housing Economics},
	author = {Kholodilin, Konstantin A.},
	year = {2024},
	pages = {101976},
}

@unpublished{kominers_price_2025,
	title = {A {Price} {Theory} of {Price} {Gouging}},
	author = {Kominers, Scott Duke and Dworczak, Piotr},
	year = {2025},
}

@article{brunt_moving_2025,
	title = {Moving beyond {Harberger}'s {Triangle} to present the inefficiency from misallocated market transactions under price controls},
	volume = {56},
	issn = {0022-0485, 2152-4068},
	url = {https://www.tandfonline.com/doi/full/10.1080/00220485.2025.2524631},
	doi = {10.1080/00220485.2025.2524631},
	language = {en},
	number = {4},
	urldate = {2025-12-26},
	journal = {The Journal of Economic Education},
	author = {Brunt, Christopher S.},
	month = oct,
	year = {2025},
	pages = {322--328},
}

@article{deacon_welfare_1989,
	title = {The {Welfare} {Costs} of {Rationing} by {Waiting}},
	volume = {27},
	journal = {Economic Inquiry},
	author = {Deacon, Robert T. and Sonstelie, Jon},
	year = {1989},
	pages = {179--196},
}

@article{weitzman_is_1977,
	title = {Is the {Price} {System} or {Rationing} {More} {Effective} in {Getting} a {Commodity} to {Those} {Who} {Need} {It} {Most}?},
	volume = {8},
	journal = {Bell Journal of Economics},
	author = {Weitzman, Martin L.},
	year = {1977},
	pages = {517--524},
}

@book{yergin_prize_1991,
	address = {New York},
	title = {The prize: the epic quest for oil, money, and power},
	isbn = {978-0-671-50248-5},
	shorttitle = {The prize},
	language = {eng},
	publisher = {Simon \& Schuster},
	author = {Yergin, Daniel},
	year = {1991},
}

@article{davis_allocative_2011,
	title = {The {Allocative} {Cost} of {Price} {Ceilings} in the {U}.{S}. {Residential} {Market} for {Natural} {Gas}},
	volume = {119},
	issn = {0022-3808, 1537-534X},
	url = {https://www.journals.uchicago.edu/doi/10.1086/660124},
	doi = {10.1086/660124},
	language = {en},
	number = {2},
	urldate = {2025-05-15},
	journal = {Journal of Political Economy},
	author = {Davis, Lucas W. and Kilian, Lutz},
	month = apr,
	year = {2011},
	pages = {212--241},
}

@article{hughes_evidence_2008,
	title = {Evidence of a {Shift} in the {Short}-{Run} {Price} {Elasticity} of {Gasoline} {Demand}},
	volume = {29},
	number = {1},
	journal = {The Energy Journal},
	author = {Hughes, Jonathan E. and Knittel, Christopher R. and Sperling, Daniel},
	year = {2008},
	pages = {113--134},
}

@article{diamond_effects_2019,
	title = {The {Effects} of {Rent} {Control} {Expansion} on {Tenants}, {Landlords}, and {Inequality}: {Evidence} from {San} {Francisco}},
	volume = {109},
	doi = {10.1257/aer.20181289},
	number = {9},
	journal = {American Economic Review},
	author = {Diamond, Rebecca and McQuade, Tim and Qian, Franklin},
	year = {2019},
	pages = {3365--3394},
}

@article{glaeser_misallocation_2003,
	title = {The {Misallocation} of {Housing} under {Rent} {Control}},
	volume = {93},
	doi = {10.1257/000282803769206188},
	number = {4},
	journal = {American Economic Review},
	author = {Glaeser, Edward L. and Luttmer, Erzo F. P.},
	year = {2003},
	pages = {1027--1046},
}

@book{rockoff_drastic_1984,
	address = {New York},
	title = {Drastic {Measures}: {A} {History} of {Wage} and {Price} {Controls} in the {United} {States}},
	publisher = {Cambridge University Press},
	author = {Rockoff, Hugh},
	year = {1984},
}

@article{harberger_monopoly_1954,
	title = {Monopoly and {Resource} {Allocation}},
	volume = {44},
	number = {2},
	journal = {American Economic Review},
	author = {Harberger, Arnold C.},
	year = {1954},
	pages = {77--87},
}

@article{gordon_response_1973,
	title = {The {Response} of {Wages} and {Prices} to the {First} {Two} {Years} of {Controls}},
	volume = {1973},
	number = {3},
	journal = {Brookings Papers on Economic Activity},
	author = {Gordon, Robert J.},
	year = {1973},
	pages = {765--815},
}

@article{oi_measuring_1976,
	title = {On {Measuring} the {Impact} of {Wage}-{Price} {Controls}: {A} {Critical} {Appraisal}},
	volume = {2},
	doi = {10.1016/S0167-2231(76)80003-6},
	journal = {Carnegie-Rochester Conference Series on Public Policy},
	author = {Oi, Walter Y.},
	year = {1976},
	pages = {7--64},
}

@article{frech_welfare_1987,
	title = {The {Welfare} {Cost} of {Rationing}-{By}-{Queuing} {Across} {Markets}: {Theory} and {Estimates} from the {U}. {S}. {Gasoline} {Crises}*},
	volume = {102},
	issn = {0033-5533},
	shorttitle = {The {Welfare} {Cost} of {Rationing}-{By}-{Queuing} {Across} {Markets}},
	url = {https://doi.org/10.2307/1884682},
	doi = {10.2307/1884682},
	number = {1},
	urldate = {2025-04-23},
	journal = {The Quarterly Journal of Economics},
	author = {Frech, III, H. E. and Lee, William C.},
	month = feb,
	year = {1987},
	pages = {97--108},
}

@book{sowell_knowledge_1980,
	address = {New York},
	title = {Knowledge and {Decisions}},
	publisher = {BasicBooks. Reissued 1996},
	author = {Sowell, Thomas},
	year = {1980},
}

@article{Bulow2012,
	title = {Regulated {Prices}, {Rent} {Seeking}, and {Consumer} {Surplus}},
	volume = {120},
	number = {1},
	urldate = {2018-03-09},
	journal = {Journal of Political Economy},
	author = {Bulow, Jeremy I. and Klemperer, Paul D.},
	year = {2012},
	pages = {160--186},
}

@article{Murphy1992,
	title = {The {Transition} to a {Market} {Economy}: {Pitfalls} of {Partial} {Reform}},
	volume = {107},
	doi = {10.2307/2118367},
	number = {3},
	urldate = {2017-10-31},
	journal = {The Quarterly Journal of Economics},
	author = {Murphy, Kevin M. and Shleifer, Andrei and Vishny, Robert W.},
	month = aug,
	year = {1992},
	note = {Publisher: Oxford University Press},
	pages = {889--906},
}

\end{document}